\documentclass[11pt]{article}

\usepackage{amssymb,amsfonts,amsmath,amsthm,amscd}
\usepackage[margin=1in]{geometry}
\usepackage{graphicx,float,epsfig}
\usepackage{wrapfig}
\usepackage{hyperref}
\usepackage{multirow}
\usepackage{breakcites}
\usepackage{bbm}
\usepackage{authblk}
\usepackage{enumitem}

\usepackage{color}
\usepackage{xcolor}

\DeclareMathAlphabet{\mathpzc}{OT1}{pzc}{m}{it}

\newtheorem{propo}{Proposition}
\newtheorem{prop}[propo]{Proposition}
\newtheorem{lemma}[propo]{Lemma}
\newtheorem{definition}[propo]{Definition}
\newtheorem{coro}[propo]{Corollary}

\newtheorem{theorem}[propo]{Theorem}

\newtheorem{remark}[propo]{Remark}

\newcommand{\st}{\sigma}

\newcommand{\1}{\mathbf{1}}

\newcommand{\E}{\mathbb{E}}
\newcommand{\N}{\mathbb{N}}

\newcommand{\R}{\mathbb{R}}
\renewcommand{\P}{\mathbb{P}}

\newcommand{\wt}{\widetilde}

\def\P{\mathbb{P}}
\def\E{\mathbb{E}}

\def\cS{\mathcal{S}}

\def\cP{\mathcal{P}}

\def\cL{\mathcal{L}}

\def\cI{\mathcal{I}}

\def\cG{\mathcal{G}}
\def\cF{\mathcal{F}}
\def\cE{\mathcal{E}}

\def\cA{\mathcal{A}}

\newcommand{\eps}{\epsilon}

\newcommand{\Tasp}{T_{\op{aspirant}}}

\newcommand{\Tinf}{T^1_{\op{informed}}}
\newcommand{\Tinff}{T^2_{\op{informed}}}

\newcommand{\Tterm}{T^1_{\op{terminal}}}
\newcommand{\init}{1}

\newcommand{\CCC}{5}
\newcommand{\DeltaT}{T}

\def\@rst #1 #2other{#1}
\newcommand\MR[1]{\relax\ifhmode\unskip\spacefactor3000 \space\fi
	\MRhref{\expandafter\@rst #1 other}{#1}}\newcommand{\MRhref}[2]{\href{http://www.ams.org/mathscinet-getitem?mr=#1}{MR#2}}

\def\MR#1{\href{http://www.ams.org/mathscinet-getitem?mr=#1}{MR#1}}
\newcommand{\aryb}{\begin{eqnarray*}}
	\newcommand{\arye}{\end{eqnarray*}}
\def\alb#1\ale{\begin{align*}#1\end{align*}}
\newcommand{\eqb}{\begin{equation}}
\newcommand{\eqe}{\end{equation}}
\newcommand{\eqbn}{\begin{equation*}}
\newcommand{\eqen}{\end{equation*}}

\newcommand{\op}{\operatorname}

\newcommand{\frk}{\mathfrak}

\newcommand{\ep}{\epsilon}
\newcommand{\rta}{\rightarrow}

\newcommand{\wh}{\widehat}

\newcommand{\old}[1]{}

\begin{document}


\title{Communication cost of consensus for nodes with limited memory}
\author{
	Giulia Fanti\thanks{Carnegie Mellon University; gfanti@andrew.cmu.edu}\qquad 
	Nina Holden\thanks{ETH Z\"urich; holdenn@eth-its.ethz.ch}\qquad 
	Yuval Peres\thanks{yperes@gmail.com}\qquad 
	Gireeja Ranade\thanks{UC Berkeley; gireeja@eecs.berkeley.edu. }}

\date{}
\maketitle

\begin{abstract}
	Motivated by applications in blockchains and sensor networks, we consider a model of $n$ nodes trying to reach consensus on their majority bit. Each node $i$ is assigned a bit at time zero, and is a finite automaton with $m$ bits of memory (i.e., $2^m$ states) and a Poisson clock. When the clock of $i$ rings, $i$ can choose to communicate, and is then matched to a uniformly chosen node $j$. The nodes $j$ and $i$ may update their states based on the state of the other node. Previous work has focused on minimizing the time to consensus and the probability of error, while our goal is minimizing the number of communications. We show that when $m>3 \log\log\log(n)$, consensus can be reached at linear communication cost, but this is impossible if $m<\log\log\log(n)$. We also study a synchronous variant of the model, where our upper and lower bounds on $m$ for achieving linear communication cost are $2\log\log\log(n)$ and $\log\log\log(n)$, respectively. A key step is to distinguish when nodes can become aware of knowing the majority bit and stop communicating. We show that this is impossible if their memory is too low.
\end{abstract}




\newpage

\section{Introduction}
Consensus algorithms are useful in distributed systems that require coordination, such as cryptocurrencies and filesharing systems. 
Many distributed systems today are run on resource-constrained networks with limited bandwidth, computation, power, or storage. 
Despite this, consensus algorithms are often designed for resource-rich environments. That is, they minimize time to consensus without considering other costs such as communication and storage. 
Some algorithms do optimize communication costs, but typically under the assumption that nodes always communicate whenever they are allowed to. 
This is not representative of resource-constrained networks, because distributed systems are increasingly being deployed on wireless networks of battery-powered devices (e.g., the Internet of Things). On such devices, the high power demands of communication can quickly drain battery life, thus incentivizing nodes to remain silent whenever possible. 
Low-power wireless devices are also more likely to have limited storage than traditional computers.

In this work, we consider a communication model that is motivated by a wireless network of resource-constrained devices. 
We make three primary modeling assumptions: (1) nodes are storage-constrained, 
(2) nodes refrain from communicating whenever possible, and
(3) the dominant cost of communication is setting up the connection.\footnote{For example, when two mobile devices exchange a message of less than 1 kB in a line-of-sight setting, the initial TLS handshake comprises over 85\% of the power overhead \cite{miranda2011tls}. As such, our model penalizes the establishment of a communication channel, but not the number of bits sent over that channel. 
Further, although we do not explicitly charge the number of bits sent in our protocol, our protocols transmit well under 1 kB for reasonable network sizes, so we are operating in a regime where establishing a connection is the energy bottleneck.
}
Our goal is to design consensus protocols that obey memory constraints while simultaneously minimizing the total communication cost over all nodes. 

\vspace{0.05in}
\noindent \textbf{Model}
We summarize our model, which is fully specified in Section \ref{sec:model}.
Consider a set of $n$ nodes in a complete graph topology, each of which can be in one of $s$ possible states.\footnote{Note that a node needs $\lceil \log_{2} s \rceil$ bits of memory to store its state.} At the beginning of the protocol, each node $i$ is assigned a bit $\frk b_i\in\{0,1 \}$ which is stored in its memory. Let $\frk b$ be the majority bit, and let $p\in(1/2,1)$ 
be the fraction of the nodes for which $\frk b_i=\frk b$. We assume $p \in [\frac{1}{2} + \eps, 1 - \eps]$, where $\eps \in (0, \frac{1}{4})$ is known to the protocol. We call $p-\frac{1}{2}$ the \emph{initial advantage}.

In the \emph{asynchronous} variant of the model each node $i$ has an independent, unit rate Poisson clock. When $i$'s clock rings, $i$ may either do nothing (which costs 0) or initiate a communication (which costs 1).
If $i$ chooses to communicate it will be connected with another node $j$ chosen uniformly at random, and the two nodes update their states based on the state of the other node. We also study a \emph{synchronous} variant of the model where the nodes are allowed to communicate at every integer time. Note that we do not use the word ``asynchronous" in the sense of unbounded communication delays, but simply to describe a continuous-time communication model.

At any time $t\geq 0$ each node $i$ has an estimate for $\frk b$, which we call the \emph{belief bit} of $i$. We have reached \emph{consensus} when all nodes have belief bit equal to $\frk b$. We say that a node is in a \emph{terminal state} if nodes in this state will never change state and never initiate further communications. We say that we have reached \emph{terminal consensus} if all nodes are in a terminal state and have belief bit equal to $\frk b$. The goal is to reach consensus or terminal consensus with high probability (w.h.p.), meaning with probability $1-o(1)$, while minimizing communication cost. 

We say that a state is \emph{aware} if a node in this state will never change its belief bit.  Notice that when we reach terminal consensus all nodes are in aware states, while this is not necessarily the case when we reach consensus.

\subsection{Main results}
It is immediate that any protocol, regardless of the memory constraint $s$, must incur a communication cost of $\Omega(n)$. Our main results provide upper and lower bounds for the \emph{threshold} on $s$ above which $\Theta(n)$ communications are sufficient. Earlier literature has studied consensus protocols for the asynchronous model with $\Theta(n\log n)$ communications and $O(1)$ (e.g.\ $s=3$) states of memory \cite{angluin2008simple,perron2009using,cruise2014probabilistic}. Synchronous variants of such protocols achieve consensus with $\Theta(n\log\log n)$ communications and $O(1)$ states of memory.  We obtain lower bounds on the number of communications needed under arbitrary memory constraints which, in particular, show that these earlier studied protocols are optimal (up to multiplication by a constant) for the case where $s=O(1)$.
Our results for the asynchronous model are summarized in Figure \ref{fig:mainresult}.
\begin{theorem}[Upper bound, asynchronous model]\label{prop:upperasync2}
	For any $\eps\in(0,1/4)$ there exists a constant $C>0$ and an asynchronous consensus protocol such that w.h.p., terminal consensus is achieved with $Cn$ communications using $s = \lceil C(\log \log n)^3\rceil$ states of memory per node if $p$ is in $[1/2+\eps,1-\eps]$.
\end{theorem}
\begin{theorem}[Upper bound, synchronous model]\label{prop:uppersync}
	For any $\eps\in(0,1/4)$ there exists a constant $C>0$ and a synchronous consensus protocol such that w.h.p., terminal consensus is achieved with $Cn$ communications using $s = \lceil C(\log \log n)^2\rceil$ states of memory per node if $p$ is in $[1/2+\eps,1-\eps]$.
	\label{thm:upper_sync2}
\end{theorem}
These upper bounds are proved by describing and analyzing explicit consensus protocols. See Sections \ref{sec:intro-upper}, \ref{app:uppersync}, and \ref{app:upperasync2}. 
Although it is not our goal to minimize running time, we remark that the asynchronous protocol terminates in time $\wt O(\log n)$ w.h.p., while the synchronous protocol terminate in time $O( (\log\log n)^3 )$ w.h.p.
We also present a simpler protocol for the asynchronous model.
\begin{prop}[Simpler upper bound, asynchronous model]
	\label{prop:upperasync1}
	For any $\eps\in(0,1/4)$ there exists a constant $C>0$ and an asynchronous consensus protocol such that w.h.p., terminal consensus is achieved with $Cn$ communications using $s = \lceil C(\log n)^2\rceil$ states of memory per node if the p is in $[1/2+\eps,1-\eps]$.
\end{prop}
\begin{figure}
	\centering
	\includegraphics[]{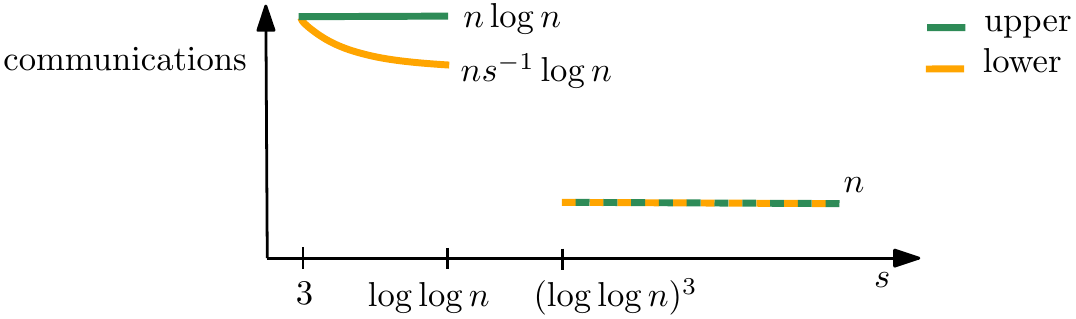}
	\caption{The figure gives an overview of our upper and lower bounds for the number of required communications in the asynchronous model, given the number of states of memory per node $s$.}
	\label{fig:mainresult}
\end{figure}

The following theorems provide lower bounds on the communication cost for nodes with a given memory constraint $s$. In particular, the theorems imply that consensus among nodes with $o(\log\log n)$ states of memory cannot  be achieved with $\Theta(n)$ communication cost.
\begin{theorem}[Lower bound, asynchronous model]\label{prop:lower-async}
	For any $\eps\in(0,1/4)$ consider an arbitrary asynchronous consensus protocol which achieves consensus on the correct bit with probability greater than $1/2$ for any $n\in\N$, $\frk b\in\{0,1 \}$, and $p\in[1/2+\eps,1-\eps]$. There is a constant $c>0$ depending only on $\eps$ such that w.h.p.\ and for $s<\log\log n-c^{-1}$, the protocol incurs communication cost at least $cns^{-1}\log n$. Furthermore, for $s<\log\log n-c^{-1}$ it holds w.h.p.\ that no node is ever in an aware state.
\end{theorem}
\begin{theorem}[Lower bound, synchronous model]\label{prop:lower}
	For any $\eps\in(0,1/4)$ consider an arbitrary synchronous consensus protocol which achieves consensus on the correct bit with probability greater than $1/2$ for any $n\in\N$, $\frk b\in\{0,1 \}$, and $p\in[1/2+\eps,1-\eps]$. There is a constant $c>0$ depending only on $\eps$ such that w.h.p., the protocol incurs communication cost at least $c(ns^{-1}\log\log n \vee  n)$. Furthermore, for $s<\log\log\log n-c^{-1}$ it holds w.h.p.\ that no node is ever in an aware state.
\end{theorem}


\subsection{Related work}
\label{sec:related}
The cost of majority consensus has been widely studied, 
and can be categorized by communication/timing model, consensus problem formulation, and cost metrics. 
We do not discuss related (more difficult) problems like leader election \cite{berenbrink2018simple} and plurality consensus  \cite{becchetti2015plurality,GP16}. 
We study  two main communication/timing models: synchronous (discrete-time) and asynchronous (continuous-time). 
Synchronous models may allow nodes to communicate with multiple nodes per time step\footnote{Our model differs  in that it allows only one communication per node per discrete time step.}, whereas asynchronous communication models generally assume gossip communication where each node can contact at most one other node per communication event. 
Metrics of interest typically include the probability of consensus, the communication cost, and the time to consensus, while constraints on communication and storage capacity are common.
We summarize relevant results in Table \ref{tab:related}, with a more detailed comparison of proof techniques and algorithms  in Section \ref{sec:related2}.
Table \ref{tab:related} uses \emph{wall-clock time} to refer to the global convergence time (expected or w.h.p., depending on the paper). 
In population protocols, this is often called \emph{parallel convergence time}, 
defined as the expected number of interactions needed for consensus, divided by $n$.
Since interactions happen concurrently in most population protocols, parallel time is related to wall-clock time by a constant factor w.h.p. 
However our protocols do not require nodes to communicate at each clock tick; as such, parallel time and  wall-clock time are not necessarily proportional in our protocols. 

Much of the relevant work is related to \emph{population protocols} \cite{angluin2006computation}, in which nodes (finite-state automata), engage in random pairwise interactions determined by a random scheduler, and update their states according to the state machine. 
Majority consensus is widely studied under the population protocol model, in two variants: 
\emph{exact majority} refers to protocols that converge to the majority bit with probability 1, whereas \emph{approximate majority} protocols can converge to the incorrect answer with positive (possibly vanishing) probability. 
In this work, we focus on approximate majority, which has received  less attention. 
Table \ref{tab:related} lists various exact consensus protocols aiming to optimize convergence time and/or storage complexity \cite{draief2012convergence,mertzios2014determining,alistarh2015fast,alistarh2017time,bilke2017population,alistarh2018space,berenbrink2018simple}. 
To date, the sharpest such result that holds for \emph{any} initial advantage is due to Berenbrink, Kaaser, Kling, and Otterbach \cite{berenbrink2018simple}, which has an optimal storage cost of $O(\log n)$ states (optimal for exact consensus) and $O(\log^{5/3}n)$ time complexity. 

In parallel, researchers have studied approximate majority protocols, mainly in the asynchronous setting, which is a more natural model for population protocols.
Angluin \emph{et al.} proposed a protocol requiring only 3 states and converging in logarithmic time \cite{angluin2008simple}, but this protocol requires the initial majority advantage to be $\Omega(\sqrt{n \log n})$.
More recently, \cite{kosowski2018population} proposed a protocol that achieves approximate majority consensus for any nonzero initial advantage, incurring constant storage cost, polylogarithmic convergence time, and $O(n \log^3 n)$ communication cost.
As these protocols were designed to optimize the time-storage tradeoff, they incur unnecessary communication cost. 
In this paper, we propose a protocol that instead achieves $O(n)$ communication cost while using $O((\log \log n)^2)$ memory states in the synchronous setting, and $O((\log \log n)^3)$ in the asynchronous setting.
Compared to \cite{kosowski2018population}, this incurs a polyloglog penalty in storage, in exchange for polylogarithmic savings in communication. 

To the best of our knowledge, relevant lower bounds have been proved only for exact consensus. 
In particular, a series of papers \cite{alistarh2015fast,alistarh2017time,bilke2017population} culminate in a result by Alistarh, Aspnes, and Gelashvili \cite{alistarh2018space} showing that to achieve exact consensus in $O(n^{1-c})$ parallel time for some $c>0$, the memory needed is $\Omega(\log n)$ states. 
We show that this is not true for approximate consensus; indeed, in a comparable asynchronous model, one can achieve consensus with $\wt O(\log n)$ parallel time using only $O((\log \log n)^3)$ states of memory and $O(n)$ messages. 
We compare the proof techniques (and protocols) of these papers more carefully in Section \ref{sec:related2}.

\begin{table}[]
\small
\begin{tabular}{|c|c|c|c|c|c|}
\hline
\multicolumn{2}{|c|}{\begin{tabular}[c]{@{}c@{}}Result\\ Type\end{tabular}} & \begin{tabular}[c]{@{}c@{}}Memory\\ ($s$ states)\end{tabular}                                                           & \begin{tabular}[c]{@{}c@{}}Communication\\ (Message)\\ Complexity\end{tabular}                                                                                                     & \begin{tabular}[c]{@{}c@{}}Time\\ Complexity\\ (Wall-clock)\end{tabular}                                                                                                    & Reference                                                                                                                                                                                                                              \\ \hline
\multirow{2}{*}{Exact}                                                         & Upper  & \begin{tabular}[c]{@{}c@{}}4\\ $O(n)$\\ $O(\log^2n)$\\ $O(\log^2 n)$\\ $O(\log n)$\\ $O(\log n)$\end{tabular} & \begin{tabular}[c]{@{}c@{}}$O(n\log n /\eps_{\star})$\\ $O(n\log n(\frac{1}{s\eps_{\star}} + \log s))$\\ $O(n\log^3n)$\\ $O(n\log^2 n)$\\ $O(n\log^2 n)$\\ $O(n\log^{5/3} n)$\end{tabular} & \begin{tabular}[c]{@{}c@{}}$O(\log n/\eps_{\star})$\\ $O(\log n(\frac{1}{s\eps_{\star}} + \log s))$\\ $O(\log^3n)$\\ $O(\log^2 n)$\\ $O(\log^2 n)$\\ $O(\log^{5/3} n)$\end{tabular} & \begin{tabular}[c]{@{}c@{}}\cite{draief2012convergence,mertzios2014determining}\\ \cite{alistarh2015fast}\\ \cite{alistarh2017time}\\ \cite{bilke2017population}\\ \cite{alistarh2018space}\\ \cite{berenbrink2018simple}\end{tabular} \\ \cline{2-6} 
                                                                               & Lower  & \begin{tabular}[c]{@{}c@{}}$\leq 4$\\ any $s$\\ $O(\log \log n)$\\ $\Omega(\log n)$\end{tabular}                        & \begin{tabular}[c]{@{}c@{}}$\Omega(n/\eps_{\star})$\\ $\Omega(n\log n)$\\ $\Omega(\frac{n^2}{(K^s+\eps_{\star} n)^2})$\\ $O(n^{2-c}),~c>0$\end{tabular}                                    & \begin{tabular}[c]{@{}c@{}}$\Omega(1/\eps_{\star})$\\ $\Omega(\log n)$\\ $\Omega(\frac{n}{(K^s+\eps_{\star} n)^2})$\\ $O(n^{1-c}),~c>0$\end{tabular}                                & \begin{tabular}[c]{@{}c@{}}\cite{alistarh2015fast}\\ \cite{alistarh2015fast}\\ \cite{alistarh2017time}\\ \cite{alistarh2018space}\end{tabular}                                                                                         \\ \hline
\multirow{2}{*}{\begin{tabular}[c]{@{}c@{}}Approx.\\ (sync)\end{tabular}}      & Upper  & $O((\log \log n)^2)$                                                                                                    & $O(n)$                                                                                                                                                                             & $O((\log \log n)^3)$                                                                                                                                                        & \multirow{2}{*}{This paper}                                                                                                                                                                                                            \\ \cline{2-5}
                                                                               & Lower  & any $s$                                                                                                                 & $\Omega \left(\frac{n \log \log n \vee n)}{s} \right)$                                                                                                                                          & ---                                                                                                                                                                         &                                                                                                                                                                                                                                        \\ \hline
\multirow{2}{*}{\begin{tabular}[c]{@{}c@{}}Approx.  \\ (async)\end{tabular}}   & Upper  & \begin{tabular}[c]{@{}c@{}}$O(1)$\\ $O(1)$\\ $O((\log \log n)^3)$\end{tabular}                                          & \begin{tabular}[c]{@{}c@{}}$O(n\log n)$\\ $O(n\log^3 n)$\\ $O(n)$\end{tabular}                                                                                                     & \begin{tabular}[c]{@{}c@{}}$O(\log n)$\\ $O(\log ^3 n)$\\ $\wt O(\log n)$\end{tabular}                                                                                   & \begin{tabular}[c]{@{}c@{}}\cite{angluin2008simple,perron2009using,cruise2014probabilistic}\\ \cite{kosowski2018population}\\ This paper\end{tabular}                                                                                  \\ \cline{2-6} 
                                                                               & Lower  & $O(\log \log n)$                                                                                                        & $\Omega \left( \frac{n\log n}{s} \right )$                                                                                                                                                          & ---                                                                                                                                                                         & This paper                                                                                                                                                                                                                             \\ \hline
\end{tabular}
\caption{Comparison of related work on majority consensus. We study approximate majority consensus (upper and lower bounds), under a synchronous (sync) and asynchronous (async) communication model. The number of nodes is denoted by $n$, the initial advantage is $\eps_{\star} = p- \frac{1}{2}$, and $K, c$ are constants. 
Lower bounds should be interpreted as follows: any protocol consuming $O(\cdot)$ of one resource (e.g., storage) requires $\Omega(\cdot)$ of another (e.g., time); upper bounds instead imply the existence of a protocol that achieves resource costs in complexity class $O(\cdot)$.
}
\label{tab:related}
\end{table}

 \vspace{0.05in}
\noindent \textbf{Outline}
 We  precisely define our model in Section \ref{sec:model}, and give brief proof outlines for our main results in Section \ref{sec:intro-upper}. We prove our upper bounds Proposition \ref{prop:upperasync1}, Theorem \ref{prop:uppersync}, and Theorem \ref{prop:upperasync2} in Sections \ref{app:upperasync1}, \ref{app:uppersync}, and \ref{app:upperasync2}, respectively. Our lower bounds (Theorems \ref{prop:lower-async} and \ref{prop:lower}) are proved in Section \ref{sec:lower}.


\section{The model}
\label{sec:model}

Consider a set of $n$ nodes connected in a complete graph topology, enumerated by \mbox{$[n]=\{1,2,\ldots, n\}$}. These indices are only for our own bookkeeping, and cannot be used by nodes during the protocol.
At any point in time a node $i\in[n]$ has a state chosen from a set $\cS$ of cardinality $s\in\{2,3,\dots \}$. 
We may assume each state is a binary string of $\lceil\log_2(s)\rceil$ bits. 
For a node $i\in[n]$ and a time $t\geq 0$, let $\st(i,t)\in\cS$ denote the state of node $i$ at time $t$. 
All logarithms we consider throughout the paper will be in base $2$, i.e., $\log x=\log_2 x$ for any $x>0$.

At the beginning of the protocol, each node $i$ is assigned a bit $\frk b_i\in\{0,1 \}$ which is stored in its memory. The state of $i$ at time $t=0$ can for example be represented as a single bit $\frk b_i$ followed by  $\lceil\log_2(s)\rceil-1$ bits 0. Let $\frk b$ be the \emph{majority bit}, i.e., $\frk b=0$ if and only if\footnote{Note that to resolve draws, we define $\frk b=0$ if there are equally many nodes for which $\frk b_i=0$ and $\frk b_i=1$.} 
\eqbn
\#\{i\in[n]\,:\,b_i=0 \}\geq \#\{i\in[n]\,:\,b_i=1 \},
\eqen
where  $\# A\in\N\cup\{0,\infty\}$ denotes the cardinality of a  set $A$ and $\N=\{1,2,\dots \}$.
Let $p\in[1/2,1]$ 
be the fraction of nodes for which $\frk b_i=\frk b$, i.e.,
$
p = n^{-1} \cdot\#\{ i\in[n]\,:\,\frk b_i=\frk b \}.
$

Each node $i$ has an independent unit rate Poisson clock $\cP_i$. We identify $\cP_i\subset\R_+$ with the set of times that the clock rings. 
Whenever $i$'s clock rings (i.e., at every time $t\geq 0$ such that $t\in\cP_i$) the node is allowed to communicate with another node. The node chooses based on its current state whether to initiate a communication with another node.
 In other words, there is a set of states $\cS'\subset\cS$ such that a node $i\in[n]$ initiates a communication with another node $j$ at time $t\in\cP_i$ if and only if $\sigma(i,t^-)\in\cS'$, where $\sigma(i,t^-)\in\cS'$ is the state of $i$ infinitesimally before time $t$. The node $j$ is always chosen uniformly at random from $[n]\setminus\{i \}$, independently of all other randomness. For each $i\in[n]$ and $t\in\cP_i$ let $\frk r(i,t)\in[n]$ denote the node which $i$ would contact at time $t$ if $\sigma(i,t^-)\in\cS'$. The process of initiating a communication  has unit cost.

When a connection is established between nodes $i$ and $j$, each node observes the state of the other node and the nodes update their states to reflect any new information gained during the interaction. The new states of the nodes are a deterministic function of the state of each node before the communication, i.e., there is a function $\Lambda:\cS'\times \cS\to\cS^2$ such that if $i$ was the initiator of the communication,
\eqbn
(\st(i,t),\st(j,t)) = \Lambda(\st(i,t^-),\st(j,t^-)).
\eqen
Let $\Lambda_1:\cS'\times \cS\to\cS$ and $\Lambda_2:\cS'\times \cS\to\cS$ denote the coordinate functions of $\Lambda$, such that $\Lambda(\sigma_1,\sigma_2)=(\Lambda_1(\sigma_1,\sigma_2),\Lambda_2(\sigma_1,\sigma_2))$ for all $\sigma_1\in\cS'$ and $\sigma_2\in\cS$. Let $\Theta_i\subset\N$ denote the set of times at which node $i$ initiates a communication, i.e.,
$
\Theta_i = \{ t\in\cP_i\,:\, \sigma(i,t^-)\in\cS' \}.
$
A node $i$ that does not initiate a communication at time $t\in\cP_i$ may also update its state. 
More precisely, there is a function\footnote{Note that for the asynchronous model defined here it is sufficient to define $\Lambda'|_{\cS\setminus\cS'}$. However, we choose to let the domain of $\Lambda'$ be $\cS$ since we use the same function for the synchronous model, which is defined later in this section.} $\Lambda':\cS\to\cS$ such that if $\sigma(i,t^-)\not\in\cS'$ (so $i$ does not communicate with any other node at time $t$),
$
\st(i,t) = \Lambda'(\st(i,t^-)).
$

At any time $t\geq 0$, each node $i$ has an estimate for $\frk b$, which we call the \emph{belief bit} of $i$ and denote by $\wh\sigma(i,t)\in\{0,1 \}$. We have reached \emph{consensus} 
when 
all nodes have belief bit equal to $\frk b$ for the remainder of the protocol, i.e., consensus is reached at the time $\tau_{\op{consensus}}$ defined by
\eqbn
\tau_{\op{consensus}} = \inf\{ t\geq 0\,:\,\wh\sigma(i,t')=\frk b,\,\forall i\in[n],\,t'\geq t \},
\eqen 
where the infimum of an empty set is $\infty$. For $t\geq 0$ let $N(t)$ denote the number of communications initiated before or at time $t$, i.e., 
$
N(t)=\sum_{i\in[n]}\#(\Theta_i\cap [0,t] ).
$
The cost until consensus is the random variable $N_{\op{consensus}}$ defined by
$
N_{\op{consensus}} = N(\tau_{\op{consensus}}),
$
i.e., $N_{\op{consensus}}$ is the number of communications required to reach consensus.

\emph{Terminal consensus} is a stronger notion of consensus. To define this, we first need to introduce the notion of a \emph{terminal state}. A state $\sigma\in \cS$ is a terminal state if a node in this state will never change state and never initiate further communications, i.e.,
\eqbn
\sigma\not\in \cS'
\qquad\text{and}\qquad
\Lambda_2(\sigma',\sigma)=\sigma,\,\,\forall \sigma'\in\cS'.
\eqen
Let $\cS_\infty\subset\cS$ denote the (possibly empty) set of terminal states. We say that we have reached terminal consensus if all nodes are in a terminal state and have belief bit equal to $\frk b$, i.e., terminal consensus is reached at the time $\tau_{\op{terminal}}$ defined by
\eqbn
\tau_{\op{terminal}} = \inf\{ t\geq 0\,:\, 
\sigma(i,t)\in\cS_\infty 
\text{\,\,and\,\,}
\wh\sigma(i,t)=\frk b,
\,\forall i\in[n] \},
\eqen
where the infimum of an empty set is $\infty$. The cost until terminal consensus is the random variable $N_{\op{consensus}}$ defined by
$N_{\op{terminal}} = N(\tau_{\op{terminal}})$.

We say that an event happens with high probability (w.h.p.) if it happens with probability $1-o(1)$, i.e., with probability converging to 1 as $n\rta\infty$. Our goal is to find a protocol which achieves consensus or terminal consensus w.h.p.\ while minimizing communication cost (i.e., minimizing $N_{\op{consensus}}$ or $N_{\op{terminal}}$).
Note that nodes have no perception of time beside the information stored in their memory. Nodes can obtain an estimate for the time by counting their own clock rings or by receiving such estimates from other nodes.

\vspace{0.05in}
\noindent \textbf{Synchronous model}
The synchronous model is defined just as the asynchronous model, except that nodes are allowed to communicate at each time in $\N$. 
However, in this model multiple nodes may try to communicate with the same node simultaneously, which leads to collisions.
Collisions are handled as follows: if there are nodes $i_1,\dots,i_\ell$ for $\ell\in\N$ which initiate a communication with a node $j$ at time $t\in\N$ then one of two possibilities occurs:
(a) If $\sigma(j,t^-)\in\cS'$, so that $j$ initiates a communication with another node at time $t$, then $j$ will not communicate with any of the nodes $i_1,\dots,i_\ell$ at time $t$. Still, each of the communications initiated by the nodes  $i_1,\dots,i_\ell$ will have unit cost.
(b) If $\sigma(j,t^-)\not\in\cS'$, so that $j$ does not initiate a communication with another node at time $t$, then $j$ establishes a connection with a uniformly chosen node $i'\in\{ i_1,\dots,i_\ell\}$. The other $\ell-1$ nodes that initiated a communication with $j$ do not exchange any information with $i$, but each of their initiated communications still have unit cost. 
Note that under these rules, any node communicates with at most one other node at a time. The nodes update their state as specified by the functions $\Gamma$ and $\Gamma'$ above, and again the goal is to minimize  $N_{\op{consensus}}$ or $N_{\op{terminal}}$.

\vspace{0.05in}
\noindent \textbf{Awareness}
We say that a state is \emph{aware} if a node in this state will always keep its belief bit for the remainder of the protocol. 
In other words, a state $\sigma\in\cS$ is aware if a node $i$ in this state at time $t$  satisfies $\wh\sigma(i,s) = \wh\sigma(i,t)$ for all $s\geq t$, no matter which other nodes it communicates with at times $>t$.
When we reach consensus (as defined by $\tau_{\op{consensus}}$) all nodes have belief bit equal to the majority bit, but the nodes are not necessarily aware that they have identified the majority bit. 
A node in a terminal state, on the other hand, never updates its belief bit and is therefore aware.
Notice that when we reach terminal consensus, all nodes are in aware states, but this is not necessarily the case when we reach consensus. Not all aware states are terminal states, since nodes in aware states may change their state (only the belief bit must stay fixed) and they may initiate communications with other nodes.



\section{Proof outlines}
\label{sec:intro-upper}
In Sections \ref{sec:intro-async1}, \ref{sec:intro-sync}, and \ref{sec:intro-async2} we present the consensus protocols used in Proposition \ref{prop:upperasync1}, Theorem \ref{prop:uppersync}, and Theorem \ref{prop:upperasync2}, respectively. The precise descriptions and analysis of the protocols are deferred  to Sections \ref{app:upperasync1}, \ref{app:uppersync}, and \ref{app:upperasync2}, respectively. Section \ref{sec:intro-lower} gives a brief proof outline for our lower bounds.

\subsection{First asynchronous upper bound for $s = C(\log n)^2$}
\label{sec:intro-async1}
All of the nodes are assigned types that describe their behavior: aspirant, expert, regular, or terminal. Aspirants aspire to be experts, and experts are the knowledgeable nodes that spread information about the correct bit. We describe below the four phases of the protocol and the behavior of each type of node. The phases are partly overlapping in time due to the asynchronous nature of the communications. See Figure \ref{fig:pollpush} for an illustration of the phases.

\vspace{0.05in}
\noindent \textbf{Expert selection phase} At time $t=0$ all the nodes are aspirants. Each aspirant $i$ repeatedly obtains an ordered tuple of bits $(b',b'')$ by asking two other uniformly chosen nodes for their belief bit in consecutive clock rings. If it observes $\log\log n$ tuples $(0,1)$ before the first tuple $(1,0)$ then it becomes an expert; otherwise it becomes a regular node.

Note that each time a node obtains a tuple $(b',b'')$ it is equally likely that $(b',b'')=(0,1)$ and that $(b',b'')=(1,0)$ (see von Neumann's unbiasing \cite{vN51}). Therefore an aspirant turns into an expert with probability $0.5^{\lceil\log\log n\rceil}\approx 1/\log n$, so we create approximately $n/\log n$ experts w.h.p.

\vspace{0.05in}
\noindent \textbf{Estimation phase} Each expert $i$ contacts a uniformly chosen node $j$ at each of $C\log n$ consecutive clock rings for sufficiently large $C$, and stores the initial  bit $\frk b_j$ of each node $j$. At the end of the estimation phase, the expert $i$ calculates the majority bit among the $\frk b_j$'s, and this becomes the new belief bit of $i$.
By a Chernoff bound and a union bound, w.h.p.\ all the experts estimate the majority bit correctly in the estimation phase if $C$ is chosen sufficiently large (depending only on $\ep$, where $\ep$ is as defined in Theorem \ref{prop:upperasync2}).

\vspace{0.05in}
\noindent \textbf{Pushing phase} Each expert $i$ initiates a communication with a uniformly sampled node $j$ at each of $\log n$ consecutive clock rings. The expert $i$ sends its estimate of the majority bit to $j$, and $j$ adopts this estimate and becomes a terminal node. Terminal nodes do not initiate any communications and do not change their state if other nodes initiate communications with them. After the $\log n$ clock rings, $i$ also becomes a terminal node.
Since there are $\Theta(n/\log n)$ experts and each expert contacts $\log n$ nodes, one can argue that w.h.p.\ a constant fraction of the nodes become a terminal node in this phase.

\vspace{0.05in}
\noindent \textbf{Pulling phase} Each regular node $i$ initiates a communication with another node every $\log n$ clock rings\footnote{In fact, regular nodes initiate a communication with another node every $\log n$ clock rings throughout the full protocol. However, only in the pulling phase and the latter part of the estimation phase they are likely to encounter a terminal node.} until it encounters a terminal node $j$. When $i$ succeeds it adopts the estimate of $j$ for the majority bit and becomes a terminal node. The protocol ends when all the nodes are terminal. 

The communication cost in this phase is $O(n)$ since a uniformly positive fraction of the nodes are terminal nodes at the beginning of the phase, so the number of trials of each regular node is stochastically dominated by a geometric random variable with uniformly positive success probability, which has expectation $O(1)$.

\begin{figure}
	\centering
	\includegraphics[scale=1]{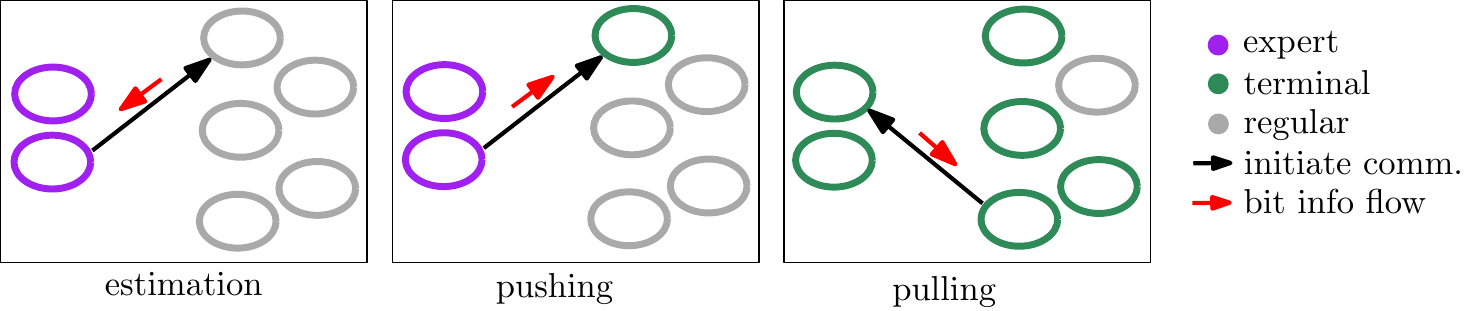}
	\caption{The figure illustrates three of the phases of the protocol described in Section \ref{sec:intro-async1}: the estimation phase, the pushing phase, and the pulling phase. In the estimation phase each expert asks $C\log n$ nodes for their bit, and each expert calculates the majority bit among the asked nodes. In the pushing phase each expert informs $\log n$ nodes about the bit calculated in the estimation phase, and these nodes become terminal nodes. In the pulling phase uninformed nodes initiate communications until they encounter a terminal node.}
	\label{fig:pollpush}
\end{figure}

\begin{figure}
	\centering
	\includegraphics{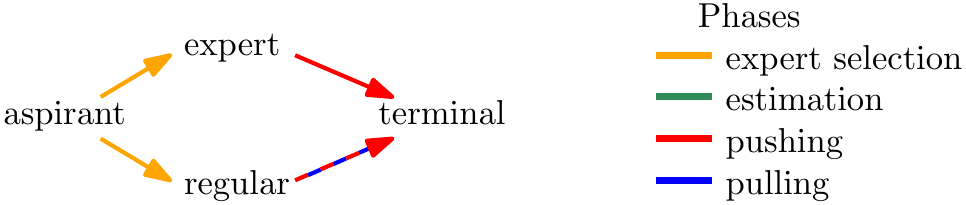}
	\caption{The figure shows the four types of nodes considered in the proof of Proposition \ref{prop:upperasync1}, and which types the nodes can move between in the various phases.}
\end{figure}

\subsection{Lower bounds proof outline}
\label{sec:intro-lower}
We first outline the lower bound for the asynchronous model (Theorem \ref{prop:lower-async}), and then we explain which changes are needed to adapt it to the synchronous case (Theorem \ref{prop:lower}).
The notion of \emph{passive} and \emph{active} states play an essential role in both  proofs. A state $\sigma\in\cS$ is \emph{passive} if a node in this state will not initiate communication until another node has contacted it. A state called \emph{active} if it is not passive. 
Since $\#\cS=s$,  active nodes must be involved in a communication (either as initiator or recipient) at least every $s$ clock rings. 
Passive states are essential for reducing the number of communications in the protocols described in Sections \ref{sec:intro-async1}, \ref{sec:intro-sync}, and \ref{sec:intro-async2}. On the other hand, as we discuss below, it is costly to have many nodes in passive states unless they have a correct estimate for the majority bit. 

Let $\cS_0\subset\cS$ be set of all states that are attained with positive probability, i.e., $$\cS_0=\{\sigma_0\in\cS\,:\,\exists t\geq 0,i\in[n]\text{\,\,such\,\,that\,\,}\P[\sigma(i,t)=\sigma_0]>0 \}.$$
Consider two cases: (i) all nodes are active at all times a.s., i.e., all states in $\cS_0$ are active, and (ii)~nodes in passive states arise with positive probability, i.e., $\cS_0$ contains at least one passive state.


For case (i) we know that even if all nodes were to initiate a communication every time their clock rings, w.h.p.\ there are nodes that do not communicate a single time before time $t=\Omega(\log n)$. Therefore $\tau_{\op{consensus}}=\Omega(\log n)$ w.h.p. This immediately implies the theorem in case (i), since we have $n$ nodes which communicate for time $\tau_{\op{consensus}}=\Omega(\log n)$ at rate at least $1/s$, so the total number of communications is $\Omega(ns^{-1}\log n)$.

In case (ii) we show that if $\sigma_0\in\cS_0$ is passive, then w.h.p.\ there are $n^{0.9}$ nodes\footnote{The exponent 0.9 is somewhat arbitrary; we can obtain any fixed power of $n$ by adjusting the constant $c$ in the statement of the theorem.} in state $\sigma_0$ at time $s$, independently of whether the true majority bit $\frk b=0$ or $\frk b=1$. 
We first explain how to conclude the proof once we have established this result. If $\frk b\neq\wh\sigma(i,t)$ for a node $i$ in state $\sigma_0$ at time $t$, then,  to reach consensus, all the $n^{0.9}$ nodes with state $\sigma_0$ must be reached by other nodes  to reach consensus. By a coupon collector argument,  $\Theta(n^{0.1}n^{0.9}\log n^{0.9}) = \Theta(n \log n)$ 
communications are necessary to reach the $n^{0.9}$ nodes. 

To prove that there are $n^{0.9}$ nodes in state $\sigma_0$ at time $s$, we show that w.h.p., for all $\sigma\in\cS_0$, there are at least $n^{0.9}$ nodes in state $\sigma$ at time $s$. 
Let $A(0)\subset\cS$ be the set of the two initial states that the nodes can take at time $t=0$. We define $A(k)\subset\cS$ inductively as the set of states that may be attained from states in $A(k-1)$, i.e., the set of all possible states that may arise from a set of nodes with states in $A(k-1)$ after one clock ring. We note that $A(k)$ is obtained deterministically from $A(k-1)$ and does not depend on the actual clock rings/communications that happen or the majority bit. Also, the number of elements in $A(k)$ are increasing with $k$, since it is always possible that a node does not change state after a clock rings. We use this and the bound on the total number of states, $\#\cS_0\leq s$, to show that in fact $A(k)=\cS_0$ for all $k\geq s$. 

We see that all states in $\cS_0$ can be present at time $s$, regardless of whether the majority bit $\frk b=0$ or $\frk b=1$. As a result we cannot have any states that are aware in $\cS_0$, i.e., states that never change their belief bit. Thus, states with the incorrect belief bit that are passive must be contacted to achieve consensus.

\begin{figure}
	\centering
	\includegraphics{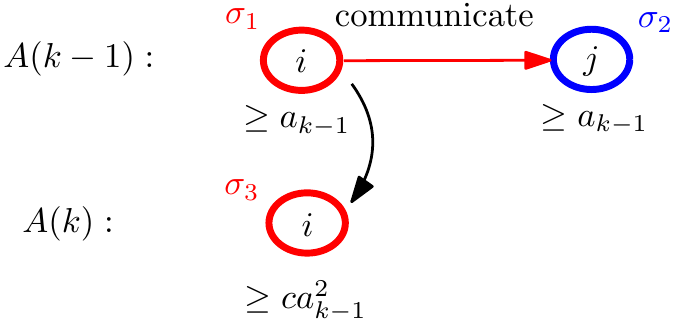}
	\caption{All states $\sigma_{1}$ and $\sigma_{2}$ in $A(k-1)$ have frequency at least $a_{k-1}$ w.h.p. Thus, all states in $A(k)$ have frequency at least $c a_{k-1}^{2}$ w.h.p.}
	\label{fig:lower}
\end{figure}

Now consider the deterministic set $A(k-1)$ at time $k-1$. Suppose that each of the states in $A(k-1)$ occurs with frequency at least $a_{k-1}$ (i.e., in at least $n a_{k-1}$ nodes) at time $k-1$ as illustrated in Fig.~\ref{fig:lower}. Then w.h.p.\ all states in $A(k)$ are found in the protocol at time $k$ with frequency  at least $c_0a_{k-1}^2$ for some constant $c_0>0$. To see why this is true, we consider all possible interactions between pairs of states in $A(k-1)$ in the unit time interval between $k-1$ and $k$. Let $\sigma_{1} \in A(k-1)$. Then, if the node in state $\sigma_{1}$ initiates communication, the probability that it interacts with some state in $A(k-1)$ is at least $a_{k-1}$. Therefore, the frequency of states in $A(k)$ that will be present at time $k$ is at least $c n a_{k-1}^{2}$ w.h.p., where the constant $c$ depends on the various probabilities of communications happening during that unit time interval from $k$ to $k+1$. 

Applying this bound on the frequency of states from $A(k)$ iteratively and using $s\leq \log\log n-c^{-1}$ we get that all states in $\cS_0=A(s)$ are found with frequency at least $(c')^{2^{s}}>n^{-1}\cdot n^{0.9}$ at time $s$ for a constant $c'>0$ w.h.p.

The proof for the synchronous model has many similar ideas: Again we define sets $A(k)$ inductively; now $A(k)$ describes the set of states which occur at time $k$ with positive probability. By a similar argument as before, for all $\sigma\in A(s)$ there are at least $n^{0.9}$ nodes in state $\sigma$ at time $\log\log n-c^{-1}$. Furthermore, we may assume no states in $\wh A:=\bigcup_{k\leq \log\log n-c^{-1}}A(k)$ are passive, since this would give $\Omega(n\log n)$ communications by a coupon collector argument. Therefore we have $n$ nodes communicating at rate at least $1/s$ for time $\log\log n-c^{-1}$, which gives a total of $\Theta(ns^{-1}\log\log n)$ communications. However, there are  some differences between the synchronous and asynchronous case: In the synchronous case the sets $\wh A$ and $\cS_0$ are not necessarily the same, since the sets $A(k)$ may not be increasing. Furthermore, consensus may be reached in only time $\Theta(\log\log n)$ (rather than $\Theta(\log n)$) in the synchronous model.

\subsection{Synchronous upper bound for $s=C (\log \log n)^2$}
\label{sec:intro-sync}
In this section we describe the protocol used in the proof of Theorem \ref{prop:uppersync}. As in the description of the first asynchronous protocol in Section~\ref{sec:intro-async1}, we rely on node types to describe the behavior of the nodes; we use aspirant, expert (at different levels), expert candidate, regular, informed, or terminal nodes. \textbf{Define} 
$M=\lceil 2\log\log n\rceil$ and $K=\lceil 5\log\log n \rceil$.

\vspace{0.05in}
\noindent \textbf{Expert selection phase} At $t=0$ all nodes are aspirants, and are differentiated to be either experts or regular nodes by the end of the expert selection phase. Approximately $n0.5^K=\Theta(n/(\log n)^5)$ nodes become level 0 experts, and the remaining nodes become regular nodes. The selection is done by a variant of von Neumann unbiasing as in Section \ref{sec:intro-async1}. However, we have to introduce some new tricks because no information is exchanged if all the nodes initiate a communication simultaneously. The protocol is described in detail in Section~\ref{app:uppersync}. 

\vspace{0.05in}
\noindent \textbf{Estimation phase} The estimation phase is divided into $M$ rounds as described below, where each round lasts for time $2K+3$. At the beginning of round $m$ there are approximately $n0.5^K$ level $m-1$ experts, while the remaining nodes are regular nodes. 
\begin{enumerate}
	\item In the first three steps of round $m$, each level $m-1$ expert $i$ initiates a communication with a uniformly chosen node $j$. A node $j$ which is contacted by a level $m-1$ expert in all three time steps becomes a level $m$ expert. Letting $b,b',b''$ denote the belief bits of the three experts contacting $j$, the node $j$ updates its belief bit to be the majority bit in $\{b,b',b'' \}$. A node which receives a bit from a level $m-1$ expert in the first step waiting to receive two more bits is called a level $m$ expert candidate.
	\item At time step 4 the level $m-1$ experts and level $m$ expert candidates change their type to regular nodes. Now all nodes are either level $m$ experts or regular nodes.
	\item At time steps 4 to $2K+3$ each level $m$ expert $i$ initiates a communication with a uniformly chosen node $j$. The node $j$ also becomes a level $m$ expert and sets its belief bit equal to the belief bit $\wh\sigma(i,t)$ of $i$.
\end{enumerate}
One can show that w.h.p.\ there are approximately $n0.5^K$ level $M$ experts, and that all level $M$ experts have a correct estimate for the majority bit (see the end of this subsection for an analysis). At the end of round $M$, the experts change their type to informed. Now all nodes are either informed or regular.


\vspace{0.05in}
\noindent \textbf{Pushing phase} In the pushing phase each informed node $i$ initiates a communication with a uniformly chosen node $j$ every time its clock rings. If $j$ is a regular node then $j$ becomes informed with the same state as $i$. If $j$ is terminal then $i$ becomes terminal and $j$ does not change its state. 
If the communication with $j$ is rejected (e.g., due to $j$ initiating its own communication), then $i$ becomes terminal, and $j$ does not change its state. The pushing phase has duration of order $\log\log n$, and at the end of the pushing phase a uniformly positive fraction of the nodes are terminal nodes.

\vspace{0.05in}
\noindent \textbf{Pulling phase} Regular nodes initiate a communication every $3MK=\Theta( (\log\log n)^2 )$ time steps throughout the protocol. The first time a regular node $i$ encounters a terminal node $j$, it adapts the majority bit estimate of $j$ and becomes a terminal node itself. 
W.h.p.\ no regular node initiates a communication until the pulling phase.
Any fixed regular node $i$ typically needs $O(1)$ trials before it succeeds in contacting a terminal node, since it contacts a terminal node with uniformly positive probability every time it initiates a communication.

\begin{figure}
	\centering
	\includegraphics{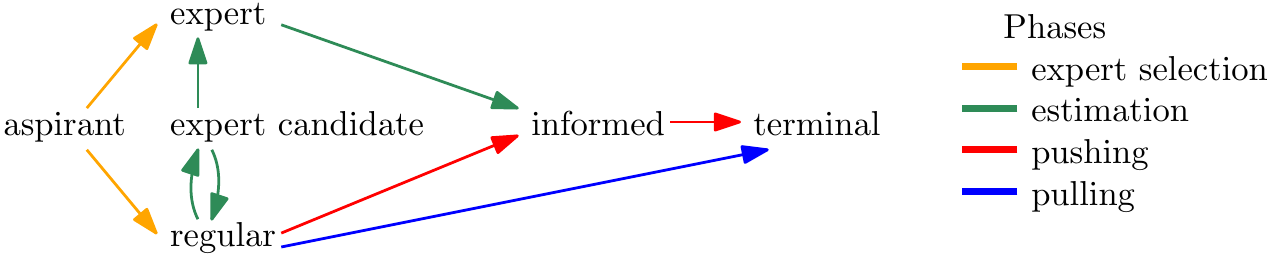}
	\caption{The figure shows the six types of nodes considered in the proof of Theorems \ref{prop:upperasync2} and \ref{prop:uppersync}, and the types the nodes can move between in the various phases.}
\end{figure}

\vspace{0.05in}
\noindent \textbf{Analysis of estimation phase.} If a fraction $\delta\ll 1$ of the level $m-1$ experts have a wrong estimate for the majority bit, then the fraction of wrong level $m$ experts will be approximately $3\delta^2$. Therefore, a level $M$ expert is wrong with probability approximately $(1-p)^{2^M}\ll n^{-1}$, so w.h.p.\ all level $M$ experts will have a correct estimate for the majority bit. 

Recall that at the beginning of round $m$ there are approximately $n0.5^K$ level $m-1$ experts. We will have approximately $n0.5^{3K}$ level $m$ experts after the first three time steps of the round, since the probability that any given node $i$ is contacted by an expert in all three steps is approximately $0.5^{3K}$. In each of the later time steps of round $m$ the number of experts approximately doubles, which gives that the number of experts at the end of the round is approximately $n0.5^{3K}\cdot 2^{2K}=n0.5^K$. 

Notice that if the number of level $m-1$ experts is $(1+\eps)n0.5^K$ for some $\eps\in(-0.5,0.5)$ then the number of level $m$ experts after the first three time steps is typically about $(1+\eps)^3n0.5^{3K} = (1+3\eps+o(\eps)) n0.5^{3K}$. In particular, the percentage-wise error $\eps$ \emph{triples}, so it grows exponentially in the round number. Therefore we need very good concentration estimates when we make rigorous the heuristic estimates of the preceding paragraph. In particular, we show that the number of collisions (which happen when an expert contacts a node which is already an expert or is contacted by another expert at the same time) is sufficiently small to be ignored.

\vspace{0.05in}
\noindent \textbf{Memory usage} Aspirants, experts, and regular nodes all require $\Omega( (\log\log n)^2 )$ states of memory (see Lemma \ref{prop42}).

\subsection{Asynchronous upper bound for $s = C (\log \log n)^3$}
\label{sec:intro-async2}
The protocol for the asynchronous model (which is used to prove Theorem \ref{prop:upperasync2}) has the same overall structure as the protocol for the synchronous case in Theorem \ref{prop:uppersync}. First there is an expert selection phase, followed by an estimation phase, a pushing phase, and a pulling phase, respectively. Due to the asynchronous nature of the Poisson clocks, the phases are partly overlapping in time. At any point in time each node is one of the following types: aspirant, expert, expert candidate, regular, informed, or terminal.

\vspace{0.05in}
\noindent \textbf{Expert selection phase} All nodes are aspirants in the beginning of the expert selection phase. The purpose of this phase is to select approximately $n2^{-K}$ level 0 experts for $K=\Theta(\log\log n)$. Nodes which do not become experts become regular nodes. The selection of experts is done by von Neumann unbiasing. 

\vspace{0.05in}
\noindent \textbf{Estimation phase} The estimation phase consists of $M=2\log\log n$ rounds. Level $m$ is associated with a set of approximately $n2^{-K}$ nodes that we call \emph{level $m$ experts}. As before, a node may become a level $m$ expert upon being contacted by at least three level $m-1$ experts, or upon being contacted by one level $m$ expert. There are approximately $n2^{-3K}$ level $m$ experts of the former kind, and their belief  bit is obtained by calculating the majority bit among the three belief bits received from level $m-1$ experts. Each of these level $m$ experts create approximately $2^{2K}$ new level $m$ experts by ``rumor-spreading'' their belief bit for $2K$ clock rings. As in the synchronous case, w.h.p.\ all level $M$ experts will identify the majority bit $\frk b$. At the end of the estimation phase all level $M$ experts become informed nodes.

One substantial challenge in the asynchronous case as compared to the synchronous case is the creation of level $m$ experts from level $m-1$ experts. In the protocol for the synchronous model a level $m$ expert is created when three level $m-1$ experts contact a node during a time interval of three clock rings, but this event is too unlikely in the asynchronous case since the levels are not synchronized; a node will be contacted by three level $m-1$ experts with about the same probability as before, but the time between each contact will typically be much larger. 
A node must remain an expert candidate for a sufficient amount of time to allow other level $m-1$ experts to contact it. However, it should not remain an expert candidate indefinitely. We let an expert candidate convert to a regular node after $\Theta((\log\log n)^2)$ clock rings, which gives sufficient time to be contacted by three level $m-1$ experts, since this is the duration of the estimation phase for most nodes. 

\vspace{0.05in}
\noindent \textbf{Pushing phase} Informed nodes spread the bit $\frk b$ until a constant fraction of the nodes are terminal nodes with the bit $\frk b$. More precisely, every time the clock of an informed node rings it contacts a uniformly chosen node, and if this node is a regular node it transforms into an informed node. Similarly as in the synchronous model, the spreading slows down when a sufficiently high fraction of the nodes are terminal nodes, since an informed node transforms into a terminal node when it contacts a terminal node.

\vspace{0.05in}
\noindent \textbf{Pulling phase} Throughout the protocol each regular node initiates a communication every $\Theta((\log\log n)^2)$ clock rings until it encounters a terminal node, upon which it also becomes a terminal node. By comparison with geometric random variables as before, we get that the number of communications in this phase is $O(n)$.

\vspace{0.05in}
\noindent \textbf{Memory usage}
Among the six types of nodes we have introduced, the expert candidates require the most memory: They use $\Theta((\log\log n)^2)$ states to count down time to the conversion to a regular node, and they use $\Theta(\log\log n)$ states to store the level number, which gives a total of $\Theta((\log\log n)^3)$ states of memory.

\section{Comparison of algorithms and proof techniques}
\label{sec:related2}

In this section, we provide a more detailed comparison of our algorithms and proof techniques with those of prior related papers, particularly those highlighted in Table \ref{tab:related}. 

\subsection{Achievable results (upper bounds)} 
Our protocols rely  on a dominant primitive from the probabilistic consensus literature:  \emph{polling}.
That is, nodes should request the opinions of their peers and adopt the majority opinion. 
Many papers  analyze polling-based protocols. 
In these protocols, nodes are arranged in a graph, and at each clock ring (synchronous or not), each node contacts a random subset of its neighbors on the graph and adopts the majority opinion among them. 
Such models have been studied on complete graphs \cite{cruise2014probabilistic}, infinite trees \cite{kanoria2011majority}, graphs of fixed degree sequence \cite{abdullah2015global}, and social networks \cite{mossel2014majority}.
Results often characterize the probability of consensus and convergence time for different graph structures and/or initial majority advantages \cite{kanoria2011majority,abdullah2015global,mossel2014majority,cruise2014probabilistic}.
We have avoided these complexities by assuming a complete graph.

None of the above papers constrains the storage available at each node explicitly. 
Among models that explicitly constrain node memory, many papers (particularly in the population protcol literature) initially considered \emph{constant} memory constraints \cite{nakata1999probabilistic,hassin2001distributed,aldous2002reversible,angluin2006computation,benezit2009interval,benezit2011distributed,draief2012convergence,shang2013upper,angluin2008simple,perron2009using,cooper2016discordant,mossel2017opinion}. 
In the constant-memory regime, many results have focused on demonstrating that consensus can be achieved, either w.h.p.\ or exactly \cite{benezit2009interval,benezit2011distributed,angluin2008simple,perron2009using,aldous2002reversible,cooper2009multiple},
as well as (upper) bounding the time to consensus \cite{draief2012convergence,shang2013upper,angluin2008simple,perron2009using}.
Despite considering slightly different models and/or problem formulations, these upper bounds tend to show that when the initial advantage is large (i.e., bounded away from $1/2$), consensus on a complete graph is achievable within $O(\log n)$ wall-clock time and/or $O(n \log n)$ interactions, as seen in Table \ref{tab:related} \cite{angluin2008simple,perron2009using,draief2012convergence,shang2013upper,cruise2014probabilistic}. 

A natural question is whether this upper bound can be tightened by giving nodes more memory;
this is the topic of our paper. 
Several papers have proposed exact majority protocols with memory constraints that grow with $n$, including \cite{alistarh2015fast,alistarh2017time,bilke2017population,alistarh2018space,berenbrink2018simple}.
For exact majority, results have focused on the setting where the initial advantage is small (e.g., as small as $1/n$), with the goal of achieving (exact) consensus  without incurring a linear time complexity. 
As such, many of the achievable schemes have a superlogarithmic time complexity; it is important to note that this arises because they are addressing a harder problem.

Related studies have considered the more difficult problem of plurality consensus \cite{berenbrink2016efficient,ghaffari2016polylogarithmic}.
These papers show that plurality consensus is possible with polylogarithmic storage in polylogarithmic time, w.h.p. 
Again, this body of work optimizes the time to consensus, rather than the communication cost.
Other papers have optimized communication costs among protocols that can withstand robustness to Byzantine faults.
However, there is typically a \emph{storage} cost associated with such robustness; for instance, \cite{GK10} requires each node to store $\Omega(\sqrt{n} \log^2 n)$ bits per node---significantly higher than our proposed protocols, which require as few as $O((\log \log n)^2)$ states.
We do not consider Byzantine fault-tolerant protocols in this work, though it is an interesting direction for future work.

\subsubsection{Comparison of Techniques}
Our protocols use some algorithmic tools that are also used in other majority consensus protocols. We summarize those tools here, while also highlighting the differences with our own protocols. 

\paragraph{Role assignment}
Notice that several of our protocols assign nodes to distinct \emph{roles} (e.g., expert), and define different state transition rules for different roles.
This is useful because it allows system designers to introduce \emph{asymmetry} into the protocol; some nodes can work harder than others. 
Role assignment is becoming a common theme in recent papers, and is used in \cite{AG15,ghaffari2016polylogarithmic,alistarh2018space}.
A key question is how to assign nodes to roles without access to a source of randomness (other than the communication scheduling mechanism). 
This is commonly handled with protocols that use interactions between nodes to infer roles. 
For example, \cite{AG15} uses a protocol where interacting nodes have associated numeric states;  a node can be a ``leader" or a ``minion", and this role is determined by comparing the magnitude of a node's own state with the state of its communication partner.
Other protocols have different ways of using node interactions to determine a node's role;
we use von Neumann unbiasing \cite{vN51}, appealing both for its simplicity and its unbiased outputs.

\paragraph{Push-pull protocols}
A well-known idea in this  literature is the fact that when spreading a rumor, it is more efficient (from a message complexity standpoint) to \emph{push} information in the beginning of the protocol, when most nodes are uninformed, and \emph{pull} information towards the end, when most nodes are already informed. 
This follows from a coupon-collection argument, and is formally analyzed in \cite{karp2000randomized}.
Although consensus is a harder problem than rumor spreading, this idea has been widely used in many consensus algorithms, including ours. 
In particular, the protocol we use to prove Proposition \ref{prop:upperasync1} first elects expert nodes, who inform themselves through polling. 
Those experts then conduct a push-pull protocol to spread their expert opinions to the rest of the network using total communication linear in $n$. 
We adapt these ideas into round-based versions for our other upper bounds, relating to Theorems \ref{prop:upperasync2} and \ref{prop:uppersync}. 

\paragraph{Timekeeping}
The ability of nodes to keep track of time is limited by their memory constraints. 
This problem is especially pronounced in the asynchronous setting, and is the main cause of the higher communication cost of asynchronous consensus compared to synchronous consensus in our results.
 
Recently,  papers have tackled the problem of timekeeping with the notion  of \emph{phase clocks}, a protocol that allows agents to (approximately) synchronize their clocks within rounds, which are defined by a given number of interactions (e.g.\ $n$ interactions).
A key innovation of \cite{alistarh2018space} is to develop a \emph{leaderless} phase clock that is able to maintain this synchronization without electing a leader, which is expensive.
The key idea is to have pairs of nodes alter their local time estimate whenever they meet, using ideas related to the power of two choices.

We do not use a leaderless phase clock to keep time; 
our protocols instead allocate a portion of each node's memory for timekeeping, which is tracked by counting rings of the node's internal clock.
Since clocks can drift apart in the asynchronous setting, the phases of our protocols have some overlap; dealing with this drift is one of the main challenges of moving from the synchronous setting to the asynchronous one.

\subsection{Converse results (lower bounds)} 
Three converse results in particular relate to our work. 
The first two lower bound the time complexity of exact consensus. 
The third bounds the communication complexity of a related problem: randomized rumor spreading.

\paragraph{Alistarh, Gelashvili, and Vojnovic \cite{alistarh2015fast}} This paper considers \emph{exact majority consensus} over a complete graph in an asynchronous setting.
Recall that exact consensus requires consensus on the correct majority bit  with probability 1. 
The authors show a lower bound of $\Omega(\log n)$ parallel convergence time for \emph{any} memory constraint $s$, as well as a scheme called average-and-conquer that achieves this bound.
Here parallel convergence time refers to the wall-clock time to convergence; in a discrete-time setting where all nodes communicate at each clock ring, it is the number of total communication instances divided by $n$.

The lower bound in \cite{alistarh2015fast} follows from a coupon-collecting argument. 
Since each node's clock rings according to a Poisson process, we must wait $\Omega(\log n)$ time before every node's clock has rung at least once w.h.p.
Since we need every node to communicate in order to achieve exact consensus (or approximate consensus, for that matter), this lower bounds the (parallel) time to consensus.

At first glance, the lower bound of \cite{alistarh2015fast}  suggests a necessary communication cost of $\Omega(n\log n)$ for population protocols, since parallel time is defined as the number of interactions divided by $n$ and interchangeable (modulo some constant factor) with the wall-clock time.
However, under our model,  all nodes need not communicate at every time step, leading to a reduced lower bound on  communication costs. 
Note that declining the opportunity to communicate can only increase the time complexity of a protocol;
indeed, the protocols we propose complete in $\tilde O(\log n)$ wall-clock time. 

\paragraph{Alistarh, Aspnes, Eisenstat, Gelashvili, and Rivest \cite{alistarh2017time}}
This paper shows that any exact majority protocol achieving consensus using $O(\log \log n)$ states requires $\Omega(n/\text{polylog}(n))$ expected convergence time.
It builds on the technical building blocks of \cite{doty2018stable}, and is the starting point for the subsequent converse bounds of \cite{alistarh2018space}. 
Although the bounds of \cite{alistarh2018space} are tighter than those in \cite{alistarh2017time}, we focus on \cite{alistarh2017time} here because its proof techniques are similar those in our lower bounds. Also, \cite{alistarh2018space} assumes protocol monotonicity and output dominance; our lower bound requires neither assumption.

The proof of \cite{alistarh2017time} has three main steps. 
The first is to show that for any initial allocation of nodes to states, any consensus protocol must eventually reach a \emph{dense} configuration, in which each state has at least a certain number of nodes in that state. 
The second step is to show a \emph{transition ordering lemma} as in \cite{ccds15}, which shows conditions that a sequence of state transitions must satisfy to eliminate incorrect states quickly. 
For example, the authors define the notion of a \emph{bottleneck} transition, which  is (roughly) a transition that has a low probability of occurring. 
They then show that if a protocol converges quickly, it cannot include any bottleneck transitions. 
Finally, using the transition ordering arguments, the authors show that if a protocol converges too quickly, there must be executions under which it converges to the wrong answer.


As summarized in Section \ref{sec:intro-lower}, our own proof has similarities to \cite{alistarh2017time}. 
First, we show that w.h.p., any protocol must end up in a dense configuration. 
Next, we argue that from such a dense configuration, one cannot reach a correct configuration without incurring the communication lower bounds in Theorems \ref{prop:lower-async} and \ref{prop:lower}. 
However, the second step of our proof is different from that of \cite{alistarh2017time}. 
Rather than invoking a transition ordering lemma, we instead use a coupon-collecting argument to show that
the number of communications needed to eliminate each incorrect state from a dense configuration is $\Theta(n\log n)$.
Such an argument was not possible in \cite{alistarh2017time} because coupon-collecting arguments give high-probability statements, which are not sufficient to reach exact consensus.


\paragraph{Karp, Schindelhauer, Shenker, and Vocking \cite{karp2000randomized}} 
This converse is the oldest of the three, and also applies to a different problem from ours. 
The bound nonetheless has implications on majority consensus.
In \cite{karp2000randomized}, a single node starts with a message; the goal is for every node to obtain the message. 
The authors consider a synchronous model in which nodes can choose not to communicate; 
each node has unlimited memory, but nodes cannot keep track of which nodes have already seen the message.
This is an easier problem than majority consensus because the final result does not depend on local knowledge of other nodes. 
The authors show a lower bound on the communication cost of any such protocol of $\Omega(n\log \log n)$ transmissions. 
This lower bound on an easier problem would seem to contradict our achievable protocol of communication cost $O(n)$.
The discrepancy stems from slight differences in communication models; \cite{karp2000randomized} requires nodes to connect to a peer in every timestep, 
at which point one or both parties can decide to communicate.  
This model can only increase the amount of communication that occurs compared to our model, in which nodes can choose not to connect at all. 

Some aspects of the  proof techniques used in \cite{karp2000randomized} are widely used in the analysis of population protocols. 
In particular, \cite{karp2000randomized} structures their proof by tracking the fraction of nodes that have received the rumor in each \emph{round} of communications, defined as a sequence of $n$ consecutive communications. 
They show that the fraction of uninformed nodes cannot decay too quickly between rounds, which thereby lower bounds the amount of communication needed to reach a fully-informed state.
Although we do not use this structure to show the full lower bound, we use a similar approach to show that the number of nodes in each state is large enough at time $s$, which is the starting point for our coupon-collector argument.

\section{First asynchronous upper bound for $s = C(\log n)^2$}
\label{app:upperasync1}
In this section we first define precisely the protocol introduced in Section \ref{sec:intro-async1}, and then we give a detailed analysis of the protocol, which proves Proposition \ref{prop:upperasync1}.

\subsection{The protocol}
\label{sec:async1-1}

We advise the reader to read the informal presentation of the protocol in  Section \ref{sec:intro-async1} before reading the more formal description here. We specify the protocol precisely by describing exactly the behavior of the nodes of the various types. 

Define the following constant
\eqbn
C_0=\frac{10}{\eps^2}.
\eqen
For each node $i\in[n]$ and time $t\geq 0$ we write the state $\sigma(i,t)\in\cS$ of $i$ at time $t$ as a tuple of integers such that the first element of the tuple indicates the type, 
the last element of the tuple is the belief bit $\wh\sigma(i,t)$, and 
the form of the tuple depends on the type. For an aspirant (resp., expert, regular node, terminal node) the first element of the tuple equals 1 (resp., 2, 3, 4). Let $\sigma_1(i,t)\in[4]$ denote the type of node $i$ at time $t$.

\paragraph{Aspirant} The state of an aspirant $i\in[n]$ at time $t\geq 0$ is of the form $\sigma(i,t)=(1,d,b',b)$, where $d\in[\lceil\log\log n\rceil]$ is the success counter,  $b'\in\{-1,0,1 \}$ is the test bit, and $b=\wh\sigma(i,t)\in\{0,1 \}$ is the belief bit. At $t=0$ each node $i\in[n]$ has state given by $\sigma(i,t)=(1,1,-1,\frk b_i)$. When the clock of an aspirant rings it initiates a communication with a uniformly chosen node $j$ and then it executes the following actions for as long as $d<\lceil\log\log n\rceil$, where $b''$ denotes the belief bit of $j$.
\begin{enumerate}[leftmargin=0.5cm]
	\item If $b'=b''$ then $i$ sets $b'=-1$. 
	\item If $b'=-1$ then $i$ sets $b'=b''$.
	\item If $b'=0$ and $b''=1$ then $d$ is increased by 1 and $b'$ is set to $-1$.
	\item If $b'=1$ and $b''=0$ then $i$ becomes a regular node with state $\sigma(i,t)=(3,1,b)$ and the process described here terminates.
\end{enumerate}
If the above process terminates because $d=\lceil\log\log n\rceil$, then $i$ becomes an expert with state $\sigma(i,t)=(2,1,1,1,b)$. 

\noindent Furthermore, if another node initiates a communication with an aspirant, then the aspirant will not change its state.

\paragraph{Expert} The state of an expert $i\in[n]$ at time $t\geq 0$ is of the form $\sigma(i,t)=(2,\xi,d,d',b)$, where $\xi\in\{1,2 \}$ is the phase, $d\in[\lceil C_0\log n\rceil]$ is the time counter, $d'\in[\lceil C_0\log n\rceil]$ is the 1-counter, and $b=\wh\sigma(i,t)\in\{0,1 \}$ is the belief bit. We say that the expert is in the \emph{estimation} phase (resp.\ \emph{pushing} phase) if $\xi=1$ (resp.\ $\xi=2$). In the analysis section, if $i\in[n]$ is an expert at time $t\geq 0$ let $\xi(i,t)\in\{1,2 \}$ denote the phase of $i$ at time $t$. When the clock of an expert rings then the expert executes the following actions.
\begin{enumerate}[leftmargin=0.5cm]
	\item If $\xi=1$ and $d<\lceil C_0\log n\rceil$ then $i$ initiates a communication with a uniformly chosen node $j$, the time counter $d$ increases by 1, and $d'$ increases by 1 if and only if $\wh\sigma(j,t)=1$.
	\item If $\xi=1$ and $d=\lceil C_0\log n\rceil$ then $i$ sets $\xi=2$ and $d=1$. Furthermore, it sets $b=1$ (resp.\ $b=0$) if $d'> C_0\log n/2$ (resp.\ $d'\leq C_0\log n/2$). 
	\item If $\xi=2$ then $i$ initiates a communication with a uniformly chosen node $j$. 
	\begin{itemize}
		\item If $d<\lceil\log n\rceil$ then $d$ increases by 1.
		\item If $d=\lceil\log n\rceil$ then $i$ becomes a terminal node with state $\sigma(i,t)=(4,b)$.
	\end{itemize}
\end{enumerate}

\noindent If another node initiates a communication with an expert, then the expert will not change its state.

\paragraph{Regular node} The state of a regular node $i\in[n]$ at time $t\geq 0$ is of the form $\sigma(i,t)=(3,d,b)$, where $d\in[\lceil\log n\rceil]$ is the time counter, and $b=\wh\sigma(i,t)\in\{ 0,1\}$ is the belief bit.
\begin{enumerate}[leftmargin=0.5cm]
	\item When the clock of $i$ rings it will initiate a communication with another node $j$ if and only if $d=\lceil\log n\rceil$. If $j$ is a terminal node then $i$ becomes a terminal node with state $\sigma(j,t)$. Otherwise $i$ will not update its state, except that the time counter $d$ increases by 1 (modulo $\lceil\log n\rceil$).
	\item If a node $j$ initiates a communication with $i$ at time $t$ then $i$ will update its state if and only if $j$ is an expert with $\xi=2$. In this case $i$ will become a terminal node with state $(4,\wh\sigma(j,t^-))$.
\end{enumerate}

\paragraph{Terminal nodes} The state of a terminal node $i\in[n]$ at time $t\geq 0$ is of the form $\sigma(i,t)=(4,b)$, where $b=\wh\sigma(i,t)\in\{0,1\}$ is the belief bit. A terminal node does not initiate communications, and does not update its state when contacted by other nodes.

\subsection{Analysis}
We prove that the protocol defined right above satisfies all conditions of Proposition \ref{prop:upperasync1}, which follows by combining Lemmas \ref{prop42}, \ref{prop43}, and \ref{prop44}. 

\begin{lemma}
	In the protocol described in Section \ref{sec:async1-1} it is sufficient that each node has $16 \lceil C_0\log n\rceil^2$ states of memory.
	\label{prop42}
\end{lemma}
\begin{proof}
	By considering each of the four types of nodes separately  we see that the experts have the largest need for memory. By considering the number of possible values that can be taken by any of the components in $\sigma(i,t)=(2,\xi,d,d',b)$, we see that the following number of states of memory is necessary
	\eqbn
	4\cdot 2\cdot \lceil C_0\log n\rceil \cdot \lceil C_{0}\log n\rceil\cdot 2 = 16 \lceil C_0\log n\rceil^2.
	\eqen
\end{proof}

The phases of the protocol described in Section \ref{sec:intro-async1} are partly overlapping in time. However, our next lemma says that w.h.p.\ there is no overlap between the expert selection phase and the pushing phase. Let $\sigma_1$ be the time the expert selection phase ends, i.e., it is the first time that there are no aspirants
\eqbn
\sigma_1 = \inf\{ t\geq 0\,:\,\sigma_1(i,t)\neq 1,\,\forall i\in[n] \}.
\eqen
Let $\sigma_2$ be the time that the pushing phase starts, i.e., it is the first time that we get an expert in the pushing phase
\eqbn 
\sigma_2 = \inf\{ t\geq 0\,:\, \exists i\in[n],\, \sigma_1(i,t)= 2,\, \xi(i,t)=2 \}.
\eqen

\begin{lemma}
	W.h.p.\ $\sigma_1<0.5C_0\log n<\sigma_2$.
	\label{prop45}
\end{lemma}
\begin{proof}
	An expert finishes its estimation phase in $\lceil C_0\log n\rceil$ clock rings. By \cite[Theorem A.1.15]{alon2004probabilistic}, the probability that it takes less than $0.5\lceil C_0\log n\rceil$ units of time to finish the estimation phase is given by the following, where $X\sim\op{Pois}(\lambda)$ for $\lambda=0.5\lceil C_0\log n\rceil$
	\eqbn
	\P[ X > 2\lambda ] = (e/4)^\lambda<n^{-0.25C_0}<n^{-2}. 
	\eqen
	It follows by a union bound over all $i\in[n]$ that $0.5\lceil C_0\log n\rceil<\sigma_2$ w.h.p.
	
	To conclude the proof we need to show that $\sigma_1<0.5 \lceil C_0\log n\rceil$ w.h.p. For any $i\in[n]$ let $A_1(i)$ denote the event that the clock of $i$ rings at least $0.4 \lceil C_0\log n\rceil$ times during $[0,0.5 \lceil C_0\log n\rceil]$.
	Then, with $X$ as in the previous paragraph, $\P[A_1(i)^c]\leq \P[X<0.8\lambda]<e^{0.02 \lambda}$. We conclude with a union bound that w.h.p.\ all the events $A_1(i)$ occur. 
	
	An aspirant repeatedly collects pairs of bits $(b',b'')$. It transforms into a regular node if it observes a pair of bits $(b',b'')=(1,0)$. Let $A_2(i)$ be the event that among the first $0.2 \lceil C_0\log n\rceil$ pairs of bits collected there is at least one pair $(1,0)$. Note that on the event $A_1(i)\cap \{ 0.5 \lceil C_0\log n\rceil<\sigma_2 \}$, all the bits $b',b''$ will have the law of initial bits of uniformly sampled nodes. Therefore, on this event, $(b',b'')=(1,0)$ with probability $p(1-p)(1+O(1/n))>\ep/2$, independently for each pair $(b',b'')$. By this observation, with $R$ denoting a geometric random variable with success probability $\eps/2$, it holds for all sufficiently large $n$ that 
	\eqbn
	\begin{split}
	\P[ A_2(i)^c; \,&A_1(i)\cap \{ 0.5 \lceil C_0\log n\rceil<\sigma_2 \}]\leq \P[R>0.2 \lceil C_0\log n\rceil]\\
	&<(1-\eps/2)^{0.2 \lceil C_0\log n\rceil}<\exp(-0.09\eps \lceil C_0\log n)<n^{-1.1}.
	\end{split}
	\eqen  
	We conclude by a union bound that w.h.p.\ all the events $A_1(i)$, $A_2(i)$, and $0.5 \lceil C_0\log n\rceil<\sigma_2$ occur. On this event we also have $\sigma_1<0.5 \lceil C_0\log n\rceil$, which concludes the proof of the lemma.
\end{proof}

Let $\cE\subset[n]$ denote the set of experts, i.e.,  
\eqbn
\cE = \{ i\in[n]\,:\,\exists t\geq 0
\text{\,\,such\,\,that\,\,}\sigma_1(i,t)=2 \}.
\eqen
\begin{lemma}
	\eqbn
	\P[ |\#\cE-\E[\#\cE]| > n^{-0.55} ] \leq 2\exp(-n^{0.1}/2)
	\qquad\text{and}\qquad
	 \E[\#\cE]\in [n/(2\log n),n/\log n).
	\eqen
	\label{prop38}
\end{lemma}
\begin{proof}
	By Lemma \ref{prop45}, w.h.p.\ any node $i$ has belief bit equal to its initial bit $\frk b_i$ throughout the aspirant phase (i.e., for times $\leq\sigma_1$). The two events 3.\ and 4.\ in the definition of an aspirant are equally likely throughout the expert selection phase. A node becomes an expert if and only if  the event in 3.\ happens $\lceil\log\log n\rceil-1$ times before the event in 4.\ happens for the first time. Therefore the probability that a node becomes an expert is exactly $0.5^{\lceil\log\log n\rceil-1}\in [1/(2\log n),1/\log n)$. Furthermore, this happens independently for each node, so we obtain the lemma by Hoeffding's inequality.
\end{proof}

Let $\cL_t$ denote the union of the terminal nodes and the experts that are in the pushing phase at time $t$. Note that the sets $\cL_t$ are increasing in $t$, i.e.,   $\cL_{t'}\subseteq \cL_t$ for $t'<t$. 
\begin{lemma}
	The protocol terminates in finite time w.h.p., and on this event it holds w.h.p.\ that all nodes have belief bit equal to $\frk b$ when the protocol terminates. In other words, $\tau_{\op{terminal}}<\infty$ w.h.p.
	\label{prop43}
\end{lemma}
\begin{proof}
Lemma \ref{prop38} implies that $\cE\neq\emptyset$ w.h.p., and on the event that $\cE\neq\emptyset$ the protocol terminates in finite time a.s.
	
Consider a sequence of pairs of random variables $\{(I_j,\rho_j)\}_{j=1}^n$, where $I_j \in [n]$ denotes the $j$th node that becomes either a terminal node or an expert in the pushing phase,
and $\rho_j \in \R_+$ denotes the time at which this happens:
$$
\rho_j = \inf \{t\geq 0 \,:\, I_j \in \cL_t\}
=\inf\{ t\geq 0\,:\,\#\cL_t\geq j \}.
$$
We show  that $\P [\wh\sigma(I_j,\rho_j)\neq \frk b]\leq n^{-2}$ for all $j \in [n]$ by induction. This implies the lemma by a union bound.

First observe that $I_1$ must be an expert which starts the pushing phase at time $\rho_1$. Furthermore, all nodes contacted by $I_1$ are nodes for which the belief bit is equal to the initial bit. For any fixed node $i$, the probability that this node encounters a node $j$ with initial bit $\frk b_j=\frk b$ if it initiates a communication is at least $(np-1)/(n-1)$ (more precisely, this bound is sharp if $\frk b_i=\frk b$, and the considered event has probability $np/(n-1)$ otherwise). The expert polls $\lceil C_0\log n \rceil-1$ nodes, and sets its belief bit equal to $\frk b$ if at least $\lceil C_0\log n \rceil/2$ of these nodes have belief bit equal to $\frk b$. Therefore, by Hoeffding's inequality and the definition of $C_0$, for all sufficiently large $n$,
\eqbn
\P\big[\wh\sigma(I_1,\rho_1)\neq \frk b\big] 
\leq
\exp \left \{ -2(p-1/2+O(1/\log n))^2 C_0\log(n)  \right \} \leq n^{-2}.
\eqen
Now suppose that for all $i \in [j-1]$, $\P [\wh\sigma(I_i,\rho_i)\neq \frk b]\leq n^{-2}$.
We want to argue that $\P [\wh\sigma(I_{j},\rho_{j})\neq \frk b]\leq n^{-2}$.
Node $I_j$  is either an expert that finishes the estimation phase at time $\rho_j$, or it is a regular node which becomes a terminal node. 
If the latter, then the claim follows from the induction hypothesis. 
If the former, then $I_j$ must have polled $\lceil C_0 \log n\rceil-1$ nodes $k$ at times $t_k$, each of which is either in 
$[n]\setminus\cL_{t_k}$ (and hence incorrect with probability $1-p$, since its belief bit $\wh\sigma(k,t_k)$ is equal to its initial bit $\frk b_k$) 
or in $\cL_{t_k}$ (and hence incorrect with probability at most $n^{-2}$). 
In both cases, $j$ is incorrect with probability at most $1-p$ for sufficiently large $n$. Using the same argument as for $j=1$, we get $\P[\wh\sigma(I_j,\rho_j)\neq \frk b] \leq n^{-2}$.
\end{proof}

Let $\DeltaT = 2(C_0+1)\log n$. The following lemma will help us to bound the number of communications initiated by regular nodes during $[\DeltaT,\infty)$.

\begin{lemma}
	There exists a constant $q>0$ depending only on $\ep$ such that w.h.p.\ $\# \cL_\DeltaT \geq qn$.
\label{lem:stragglers}
\end{lemma}
\begin{proof}
	Let $E_1$ denote the event that at least $n/(4\log n)$ experts have completed both the estimation phase and the pushing phase by time $\DeltaT$. Let $E_2=\{ |\#\cE-\E[\#\cE]| > n^{-0.55}\}$. Then $E_2$ occurs w.h.p.\ by Lemma \ref{prop38}. 
	
	We will argue that $E_1$ occurs w.h.p. By Lemma \ref{prop45}, it holds w.h.p.\ that the aspirant phase finishes before time $0.5 \lceil C_0\log n\rceil$. Therefore it is sufficient to prove that for at least $n/(4\log n)$ experts the estimation phase and the pushing phase combined take less time than $(1.5C_0+2)\log n$. It takes $a:=\lceil C_0\log n \rceil+\lceil\log n\rceil$ clock rings for an expert to finish both the estimation phase and the pushing phase. Therefore, for a node $i$ sampled uniformly at random from $\cE$, the probability that these two phases take more time than $\lambda=(1.5C_0+2)\log n$ is equal to the following, where $X$ is a Poisson random variable of parameter $\lambda$ 
\eqbn
\P[ X \leq a] \leq (1+a)\P[X=a] = (1+a)\frac{\lambda^a e^{-\lambda}}{a!}.
\eqen
Notice that the right side is smaller than $0.01$ for all sufficiently large $n$. By Lemma \ref{prop38} there are at least $n/(2.1\log n)$ experts w.h.p., and by Hoeffding's inequality and independence of the Poisson clocks we conclude that w.h.p.\ $E_1$ occurs.

Let $\cF$ be the $\sigma$-algebra generated by $\cE$ and by the set of experts which finish both the estimation phase and the pushing phase before time $T$. Note that $E_1$ and $E_2$ are measurable with respect to $\cF$.
On the event $E_2$, the $\geq\frac{n}{4\log n}$ nodes which become experts before time $T$ send their belief bit to $\lceil \log n \rceil$ nodes. Let $Y$ be the number of non-expert nodes which are contacted by at least one of these experts (which means that this node becomes informed before time $T$). Then we clearly have $Y\geq\#\cL_T$. The probability that a non-expert node is contacted by at least one expert (which means that this node becomes informed) is at least $1-(1-1/(n-1))^{ \frac{n}{4\log n}\cdot \lceil \log n \rceil }>0.22$ for all sufficiently large $n$. Therefore, for all sufficiently large $n$,
\eqbn
\E[Y\,|\,\cF]\1_{E_1\cap E_2}>0.22\cdot (n-n/\log n-n^{0.55})>0.21n.
\label{eq136}
\eqen

Let $i_1,i_2,\dots,i_\ell$ be an enumeration of the nodes which are contacted by an expert which finishes both the estimation phase and the pushing phase before time $T$. Conditioned on $\cF$, the random variable $Y$ is a function of $i_1,i_2,\dots,i_\ell$. Furthermore, $\ell<\lceil\log n \rceil\cdot (n/\log n+n^{0.55})$ on $E_2$. It follows from \eqref{eq136} and McDiarmid's inequality that $Y>0.2n$ w.h.p. This concludes the proof since $Y\geq\#\cL_T$.
\end{proof}

\begin{lemma}
	For the protocol described in Section \ref{sec:async1-1} there is a $C>0$ depending only on $\eps$ such that w.h.p.\ the communication cost is smaller than $Cn$.
	\label{prop44}
\end{lemma}
\begin{proof}
The communication cost can be split into three parts, depending on whether the communication was initiated by an aspirant, an expert, or a regular node.

An aspirant repeatedly collects two bits $b'$ and $b''$ by initiating two communications. It transforms into a regular node if it observes a pair of bits $(b',b'')=(1,0)$. Therefore the number of communications of each aspirant (divided by 2) is stochastically dominated by a geometric random variable with success probability $p(1-p)\cdot (1+O(n^{-1}))$ (where the correction term $(1+O(n^{-1}))$ is added since a node cannot initiate a communication with itself). Furthermore, the number of communications is independent for the different nodes. Letting $R_i$ for $i\in[n]$ denote independent geometric random variables with success probability $p(1-p)/2$, we see that the number of communications initiated by aspirants is smaller than $4n/(p(1-p))$, except on an event of probability at most
\eqb
\P\left[ \sum_{i=1}^{n} R_i\geq \frac{4n}{p(1-p)} \right].
\label{eq106}
\eqe
The probability in \eqref{eq106} converges to 0 as $n$ goes to infinity by e.g.\ Chebyshev's inequality, uniformly for all $p\in[1/2+\ep,1-\ep]$.

Each expert initiates $\lceil C_0 \log n\rceil-1$ communications during the estimation phase and $\lceil\log n\rceil$ communications during the pushing phase, so the total number of communication is $\leq \#\cE\cdot ((C_0+1)\log n+2)$. It follows from Lemma \ref{prop38} that w.h.p.\ the total number of communications initiated by experts is smaller than $(C_0+2)n$.

We separate the communication accounting into two time intervals: $[0,\DeltaT]$ and $[\DeltaT,\infty)$, and begin by considering the latter interval.

Lemma \ref{lem:stragglers} implies that the communication required for each regular node, starting at time $\DeltaT$, to reach a node in $\cL_\DeltaT$ is stochastically dominated by a geometric random variable with probability of success $q$. If a regular node $i$ initiates a communication with a node in $\cL_\DeltaT$ at some time $t>\DeltaT$ then $i$ becomes a terminal node.
We upper bound the number of communications initiated by regular nodes during $[\DeltaT,\infty)$ by considering the sum of $n$ independent geometric random variables $\wh R_i,\dots,\wh R_n$ with success probability $q$. By Chebyshev's inequality,
\begin{eqnarray*}
	\P \left [ \sum_{i=1}^n \wh R_i \geq \frac{2n}{q}\right ] 
	\leq
	\frac{\op{Var}[\wh R_i]q^2}{n} = \frac{(1-q)}{n}.
\end{eqnarray*}
We conclude that the regular node communication cost after time $\DeltaT$ is at most  $2n/q$ w.h.p.

Next, we consider the time interval $[0,\DeltaT]$. For each $i\in [n]$ the number of times in $[0,\DeltaT]$ at which $i$ is a regular node and initiates a communication is bounded above by $\#(\cP_i\cap[0,\DeltaT])/\lceil\log n\rceil$. Since $\#(\cP_i\cap[0,\DeltaT])$ has the law of a Poisson random variable of parameter $\DeltaT$, an application of Chebyshev's inequality gives
\eqbn
\P\left[ \sum_{i\in [n]} \frac{\#(\cP_i\cap[0,\DeltaT])}{\lceil\log n\rceil} \geq 2(C_0+2) n   \right ] 
\leq
\frac{\op{Var}[ \#(\cP_i\cap[0,\DeltaT]) ] }{(2(C_0+2)\lceil\log n\rceil)^2n}
=\frac{1}{4(C_0+2)\lceil\log n\rceil n}
.
\eqen
It follows that the regular node communication cost during $[0,\DeltaT]$ is at most  $2(C_0+2)n$ w.h.p.
\end{proof}

\section{Lower bounds}
\label{sec:lower}
In Section \ref{sec:lower-async} we prove Theorem \ref{prop:lower-async}.
In Section \ref{sec:lower-sync} we explain which modifications are needed to prove Theorem \ref{prop:lower}.

\subsection{Asynchronous model}
\label{sec:lower-async}
Let $A(0)\subset\cS$ be the set of size two containing the states that may be attained at time $k=0$. For $k\in\N$ define $A(k)\subset\cS$ inductively by letting $A(t)$ be the set of states that can be obtained via one Poisson clock ring from a group of nodes with states in $A(k-1)$, i.e., with $\Lambda,\Lambda'$ as in Section \ref{sec:model},
\eqbn
A(k) = A(k-1)
\cup\{ \Lambda(\sigma,\sigma')\,:\,\sigma,\sigma'\in A(k-1),\sigma\in\cS' \}
\cup\{ \Lambda'(\sigma)\,:\, \sigma\in A(k-1)\setminus\cS' \}.
\eqen

Observe that the size of sets $A(k)$ is increasing in $k$ and that if $A(T+1)=A(T)$ for some $T\in\N$ then $A(k)=A(T)$ for all $k\in\{T,T+1,T+2,\dots \}$. Since $\#\cS=s$ this implies 
\eqbn
A(k)=A(s),\qquad \text{for}\,\,\,k=s,s+1,s+2,\dots,
\eqen
and further
\eqb
A(s)=\bigcup_{k=0}^\infty A(k).
\label{eq45}
\eqe
The following lemma is immediate by the definition of the sets $A(k)$.
\begin{lemma}
	With probability 1, $\{ \sigma(i,t)\,:\, t\geq 0,i\in[n] \} \subset A(s)$.
	\label{prop35}
\end{lemma}

To study the evolution of the states, we need to understand how some nodes may \emph{influence} the state of other nodes. Recall that $\frk r(i, t) \in[n]$ denotes the node that $i$ contacts at time $t$ on the event that $i$ initiates a communication at time $t$.
\begin{definition}
	Node $i$ \emph{influences} node $i'$ during an interval $J\subset[0,\infty)$ if we can find an increasing sequence of times $t_1<\dots<t_\ell$ in $J$ and a sequence of nodes $i_0,i_1,\dots,i_{\ell-1},i_\ell\in[n]$ such that 
	\begin{itemize}
		\item $i_0=i$, $i_\ell=i'$, and
		\item for $j=1,\dots,\ell$, either $t_j\in\cP_{i_j}$ and $\frk r(i_j,t_j)=i_{j-1}$ or 
		$t_{j}\in\cP_{i_{j-1}}$ and $\frk r(i_{j-1},t_{j})=i_{j}$. 
	\end{itemize}
A node $i$ influences itself during any interval of time.
	
Let $T(i,J)\subset[n]$ denote the set of nodes influenced by node $i$ during $J$.
\label{def:influence}
\end{definition}
Note that some node $i$ may influence some node $j$ by the above definition although the state of $i$ has no actual impact on the state of $j$. The above definition gives an upper bound on the set of nodes whose state could potentially be impacted by $i$, given the set of Poisson clock rings and the random variables $\frk r(i,t)$. If $i$ does not influence $j$ during $J$ according to the definition, then the state of $i$ at the beginning of $J$ has no impact on the state of $j$ at the end of $J$.

Let $E(J)$ be the event that no nodes influence $n^{0.05}$ or more nodes during $J$, i.e.,
\eqbn
E(J) = \{ \#T(i,J)<n^{0.05},\,\,\forall i\in[n] \}.
\eqen

\begin{lemma}
	There is a universal constant $c>0$ such that for any fixed interval $J$ of length $1$, $\P[E(J)^c] \leq \exp(-cn^{0.05})$.
	\label{prop37}
\end{lemma}
\begin{proof}
	We assume $J=[0,1]$ to simplify notation, but the general case can be done similarly. For any fixed $i\in\N$ define $\wh L=\#T(i,J)$. The random variable $\wh L$ is stochastically dominated by a Yule-Furry process with rate $2$ at time $1$, since $i$ initiates a communication at rate 1 and is contacted by another node at rate 1. (Note that $\wh L$ is not exactly equal in law to a Yule-Furry process since the set of nodes is finite and the rate at which a new node is added to $T(i,[0,t])$ is equal to $2\cdot \#T(i,[0,t])(n-\#T(i,[0,t]))/(n-1)$ (not $2\cdot \#T(i,[0,t])$).) By \cite[page 180]{karlin66},
	\eqbn
	\P[\wh L=\ell] = \exp(-2)(1-\exp( -2 ))^{\ell-1}.
	\eqen
	Integrating this,
	\eqb
	\P[ \wh L\geq n^{0.05} ] \leq \exp( -(n^{0.05}-2)e^{-2} ),
	\label{eq99}
	\eqe
	and by taking a union bound over all $i\in[n]$ we obtain the lemma.
\end{proof}

\begin{lemma}
	For a universal constant $c>0$, $\P[\tau_{\op{consensus}}<0.1\log n]<\exp(-cn^{0.1})$.
	\label{prop70}
\end{lemma}
\begin{proof}
	Let $Y$ denote the number of nodes $i$ that have never communicated (as neither initiator nor recipient) before time $0.1\log n$ and for which the initial bit $\frk b_i$ is different from the majority bit, i.e., $\frk b_i\neq \frk b$. Then we clearly have
	\eqbn
	\P[\tau_{\op{consensus}}<0.1\log n] \leq \P[Y=0].
	\eqen
	For each node $i$ for which $\frk b_i\neq\frk b$, the probability that $i$ has not communicated with anyone before time $0.1\log n$ is at least the following
		\eqbn
	\P\big[ (\cP_i\cap[0,0.1\log n])=\emptyset]\cdot\P[\frk r(j,t)\neq i\,\forall j\in[n],t\in\cP_j\cap[0,0.1\log n ]\big]\geq n^{-0.1}\cdot n^{-0.1}=n^{-0.2}.
	\eqen
	Since there are at least $\ep n$ nodes with the wrong initial bit, this gives $\E[Y]\geq \ep n^{0.8}$. 
	
	For $i\in[n]$ let $X_i$ denote the randomness associated with $\cP_i\cap[0,0.1\log n]$ and $\frk r(i,t)$ for $t\in \cP_i\cap [0,0.1\log n]$. Then $Y$ is a function of the independent random variables $X_i$. Let $E$ be the event that no Poisson clock rings more than $n^{0.1}$ times during $[0,0.1\log n]$. By a union bound and \cite[Theorem A.1.15]{alon2004probabilistic}, for all sufficiently large $n$,
	\eqbn
	\P[E^c] \leq n\P[ \#\cP_1\cap [0,0.1\log n]>n^{0.1} ] \leq \exp(-n^{0.1}).
	\eqen
	Changing one $X_i$ cannot change $Y$ by more than $n$, and 
	on the event $E$ changing one $X_i$ does not change $Y$ by more than $2n^{0.1}+1$. By a variant of McDiarmid's inequality when differences are bounded with high probability \cite[Theorem 3.9]{kutin-mcdiarmid}, for all sufficiently large $n$,
	\eqbn
	\P[Y=0] \leq \P[ |Y-\E[Y]|>n^{0.65} ] 
	\leq 4\exp\Big(-\frac{(n^{0.65})^2}{8(2n^{0.1}+1)^2 n} \Big)
	\leq 4\exp(-0.05n^{0.1}).
	\eqen
	The lemma follows by choosing $c$ sufficiently small.
\end{proof}

\begin{lemma}
	There is a constant $a\in(0,1)$ depending only on $p$ such that with probability at least $1-\exp(-an^{0.05})$ the following holds for $t=0,\dots,s$ and all $\sigma\in A(t)$
	\eqb
		\# \{ i\in[n]\,:\,\sigma(i,t)=\sigma  \} \geq a^{2^{t}}n.
		\label{eq104}
	\eqe
	\label{prop36}
\end{lemma}
Before presenting the proof we observe that the right side of \eqref{eq104} is greater than $n^{0.9}$ for $t\leq s$ and $c$ sufficiently small
\eqb
a^{2^{t}}n 
\geq a^{2^{\log\log n-c^{-1}}}n
= n^{1+(\log a )2^{-1/c} }.
\label{eq105}
\eqe
\begin{proof}[Proof of Lemma \ref{prop36}]
	For $k\in\N\cup\{0 \}$ and $a_0,a_1\in(0,1)$, let $E_{a_0,a_1}(k)$ denote the event that all states in $A(k)$ are well represented at time $k$. More precisely,
	\eqbn
	E_{a_0,a_1}(k) = \{  \# \{ i\in[n]\,:\,\sigma(i,k)=\sigma  \} \geq a_0^{2^k} a_1^{2^k-1} n,\,\forall \sigma\in A(k) \}.
	\eqen
	Let $c\in(0,1)$ be the constant in Lemma \ref{prop37}. We will prove that for all $k\in\N$ and for $a_0,a_1\in(0,1)$ depending only on $p$, the following holds for all sufficiently large $n$
	\eqb
		\P[ E_{a_0,a_1}(k)^c;\,E_{a_0,a_1}(k-1)]<4s\exp(-cn^{0.05}/2).
	\label{eq46}
	\eqe
	
	We will first explain why \eqref{eq46} implies the lemma. Observe that $E_{a_0,a_1}(0)$ occurs for $a_0<2\ep$. This and \eqref{eq46} imply the following for $a\leq 0.5a_0a_1$ and all sufficiently large $n$
	\eqbn
	\begin{split}
		\P[ \# \{ i\in[n]\,:&\,\sigma(i,k)=\sigma  \} \geq a^{2^{k}}n,\,\forall 
		\sigma\in A(k), k=0,\dots,s]\geq \P\bigg[ \bigcup_{k=0}^{s} E_{a_0,a_1}(k) \bigg]\\
		&\geq 1 - \sum_{k=1}^{s} \P[ E_{a_0,a_1}(k)^c;\,E_{a_0,a_1}(k-1)]\\
		&\geq 1-4s^2\exp(-cn^{0.05}/2).
	\end{split}
	\eqen
	By choosing $a$ sufficiently small this implies the lemma, since for sufficiently small $a$ right side is $\geq 1-\exp(-an^{0.05})$.
	
	We will now prove \eqref{eq46}. Assume $E(k-1)$ occurs and let $\sigma\in A(k)$. There are four cases (recall that $\Lambda_1$ and $\Lambda_2$ denote the coordinate functions of $\Lambda=(\Lambda_1,\Lambda_2)$): 
	(i) $\sigma\in A(k-1)$, 
	(ii) $\sigma\in \{ \Lambda_1(\sigma',\sigma'')\,:\,\sigma',\sigma''\in A(k-1),\sigma'\in\cS' \}$, 
	(iii) $\sigma\in \{ \Lambda_2(\sigma',\sigma'')\,:\,\sigma',\sigma''\in A(k-1),\sigma'\in\cS'  \}$, and
	(iv) $\sigma\in \{ \Lambda'(\sigma')\,:\, \sigma'\in A(k-1)\setminus\cS' \}$.
	We will only consider (i) and (ii), since (iii) and (iv) can be treated similarly.
	
	For any $i\in[n]$ the probability that the Poisson clock of $i$ does not ring during a given interval of length $1$ and that no one initiates a communication with $i$ during this interval is at least $e^{-2}$. Therefore we get the following in case (i) by choosing $a_0$ and $a_1$ sufficiently small
	\eqbn
	\begin{split}
		\E[\# \{ i\in[n]\,:\,\sigma(i,k)=\sigma  \}\,|\, E_{a_0,a_1}(k-1)] 
		&\geq \P[\sigma(i,k)=\sigma\,|\,\sigma(i,k-1)=\sigma]\cdot a_0^{2^{k-1}}a_1^{2^{k-1}-1} n\\
		&\geq e^{-2} \cdot a_0^{2^{k-1}}a_1^{2^{k-1}-1} n \geq a_0^{2^{k}}a_1^{2^{k}}n.
	\end{split}
	\eqen
	
	For any $i\in[n]$ the probability of the following event is at least $e^{-4}$ for $J=[k-1,k]$  
	\begin{itemize}
		\item the Poisson clock of $i$ rings exactly once during $J$ (probability $e^{-1}$), 
		\item no one initiates a communication with $i$ during $J$ (probability $\geq e^{-1}$),  
		\item if $i$ chooses to communicate when its Poisson clock rings then the node $j$ that it contacts has a Poisson clock which does not ring at all during $J$ (probability $e^{-1}$), and
		\item no one else than $i$ initiates a communication with $j$ during $J$ (probability $\geq e^{-1}$).
	\end{itemize}
	On this event $i$ will have state $\Lambda_1(\sigma(i,k-1),\sigma(j,k-1) )$ at time $k$.
 	We get the following in case (ii) by choosing $a_0$ and $a_1$ sufficiently small
	\eqbn
	\begin{split}
		\E[\# \{ i\in[n]\,:\,&\sigma(i,k)=\sigma  \}\,|\, E(k-1)]
		\geq \P[ \sigma(i,k)=\sigma \,|\,\sigma(i,k)=\sigma' ]
		\cdot a_0^{2^{k-1}}a_1^{2^{k-1}-1} n \\
		&\geq 
		e^{-4} \cdot a_0^{2^{k-1}}a_1^{2^{k-1}-1}\cdot a_0^{2^{k-1}}a_1^{2^{k-1}-1} n \cdot \frac{n-1}{n}
		\geq 
		a_0^{2^{k}}a_1^{2^{k}-1}n.
	\end{split}
	\eqen
	Note that the extra factor of $\frac{n-1}{n}$ in the second term may be needed if $\sigma=\Lambda(\sigma',\sigma')$ for $\sigma'\in A(k-1)$. 

	Concentration of the above random variables $\# \{ i\in[n]\,:\,\sigma(i,k)=\sigma \}$ follow from a version of McDiarmid's inequality when differences are bounded with high probability \cite[Theorem 3.9]{kutin-mcdiarmid}. We write out the details for case (ii), but the other cases are treated in the exact same way. For $i\in[n]$ let $X_i$ denote the randomness associated with $\cP_i\cap[k-1,k]$ and $\frk r(i,k)$ for $k\in \cP_i\cap [k-1,k]$. Note that the random variables $X_i$ are independent. Let $\cF$ denote the $\sigma$-algebra generated by the random variables $\sigma(i,k-1)$ for all $i\in[n]$. Conditioned on $\cF$, the random variable $Y:=\# \{ i\in[n]\,:\,\sigma(i,k)=\sigma \}$ is a function of the random variables $X_i$. Changing one $X_i$ cannot change $Y$ by more than $n$, and on the event $E$ of Lemma \ref{prop37} changing one $X_i$ changes $Y$ by at most $n^{0.05}$. By Lemma \ref{prop37} and McDiarmid's inequality for differences bounded with high probability, the following holds for all sufficiently large $n$,
	\eqbn
		\P[ Y<0.5a_0^{2^{k}}a_1^{2^{k}-1}n ]
		\leq \P[ |Y-\E[Y\,|\,\cF]|> \sqrt{c}n^{0.57} ]
		\leq 4\exp\left( -cn^{0.05}/2  \right).
	\eqen
	Taking a union bound over all $\sigma\in A(k)$ we obtain \eqref{eq46}.
\end{proof}

\begin{lemma}
	Under the assumptions of Theorem \ref{prop:lower-async}, there is a $n_0\in\N$ depending only on $p$ such that for $n\geq n_0$ the set $A(s)$ contains no aware states.
	\label{prop39}
\end{lemma}
\begin{proof}
	Let $\sigma\in A(s)$, and let $\frk b\in\{0,1 \}$ be such that the belief bit of a node with state $\sigma$ is $1-\frk b$. Choose initial states such that $\frk b$ is the majority bit. By Lemma \ref{prop36} there is an $n_0\in\N$ such that for $n\geq n_0$ it holds with probability greater than $1/2$ that we can find a node $i\in[n]$ such that $\sigma(i,s)=\sigma$. If $\sigma$ is an aware state, then $i$ would have belief bit $1-\frk b$ for all $t\geq s$, which is a contradiction to the assumption that the protocol reaches consensus with probability greater than $1/2$.
\end{proof}

In our proof of Theorem \ref{prop:lower-async} we need to ensure that a significant fraction of nodes cannot remain silent for long periods of time. If this was possible, we would not be able to link time elapsed with the number of communication events. To address this we introduce the notion of \emph{passive states}. 
\begin{definition}[Passive states]
	We say that a state $\sigma\in\cS$ is a \emph{passive} state if a node in this state will not initiate any communications until it has been contacted by another node. In other words, $\sigma$ is passive if $(\Lambda')^k(\sigma)\not\in\cS'$ for all $k\in\N$, where $\cS'$ and $\Lambda'$ are defined as in Section \ref{sec:model}.  Let $\cS^{\op{p}}\subset\cS$ denote the set of passive states.
\end{definition}

\begin{proof}[Proof of Theorem \ref{prop:lower-async}]
	The assertion about aware states is immediate by Lemma \ref{prop39}.
	
	To prove the assertion about the number of communications, consider two cases: 
	(i) $A(s)\cap\cS^{\op{p}}\neq\emptyset$ and 
	(ii) $A(s)\cap\cS^{\op{p}}=\emptyset$.
	
	(i) Let $\sigma_0\in  A(s)\cap\cS^{\op{p}}$ be an arbitrarily chosen passive state in $A(s)$. By Lemma \ref{prop36} (see \eqref{eq105}) and \eqref{eq45} it holds with probability converging to 1 as $n\rta\infty$ that
	\eqbn
		\#\cS_{\sigma_0}  \geq a^{2^s}n\geq n^{0.9},\qquad\text{where}\qquad
		\cS_{\sigma_0} :=  \{ i\in[n]\,:\,\sigma(i,s)=\sigma_0  \}.
	\eqen
	By symmetry in 0 and 1 we may assume without loss of generality that a node with state $\sigma_0$ estimate the majority bit to be 0. Since we require consensus to be reached both for $\frk b=0$ and $\frk b=1$, we may also assume that $\frk b=1$. Since nodes with state $\sigma_0$ are passive, before consensus can be reached, for each $j\in\cS_{\sigma_0}$ there must be some node who initiates a communication with $j$. Since $\#\cS_{\sigma_0}\geq n^{0.9}$, by a standard coupon collector argument, the number of communications needed to contact all these nodes stochastically dominates the sum of independent random variables $Y_1,\dots,Y_{\lceil n^{0.9}\rceil}$, where $Y_k$ has the law of a geometric random variable of parameter $k/n$. In particular, with $N_{\op{consensus}}=N(\tau_{\op{consensus}})$ as Section \ref{sec:model},
	\eqbn
	\P[N_{\op{consensus}}<n(\log n)^{0.9}] \leq 
	\P\left[ \left ( \sum_{k=1}^{\lceil n^{0.9}\rceil} Y_k \right ) < n(\log n)^{0.9} \right].
	\eqen
	Since
	\eqbn
	\E\left[ \sum_{k=1}^{\lceil n^{0.9}\rceil} Y_k \right] 
	=\sum_{k=1}^{\lceil n^{0.9}\rceil}\frac{n}{k}
	\geq 0.9n\log n
	\qquad\text{and}\qquad
	\op{Var}\left[ \sum_{k=1}^{\lceil n^{0.9}\rceil} Y_k \right] 
	=\sum_{k=1}^{\lceil n^{0.9}\rceil} \frac{n(n-k)}{k^2}
	\leq 2n^2,
	\eqen
	we obtain the desired bound by applying
	Chebyshev's inequality.\footnote{We could have obtained a better bound for the probability by evaluating $\E\big[ \exp\big( c\sum_{k=1}^{\lceil n^{0.9}\rceil} Y_k \big) \big]$ for an appropriate constant $c$ and applying Markov's inequality. However, the estimate we find here is sufficient for our purpose.}

	(ii)  By Lemma \ref{prop35} and since no states in $A(s)$ are passive, all nodes communicate at least every $s$ clock ring (either as an initiator or a recipient of the communication), so 
	\eqbn
	N(t) \geq 0.5 \sum_{i\in[n]} \lfloor s^{-1}\cdot\#([0,t]\cap\cP_i) \rfloor.
	\eqen
	The random variables in the sum on the right side are i.i.d. For $t=0.3\log n$ and $s\leq\log\log n-c^{-1}$ there is a universal constant $C$ such that 
	\eqbn
	\E[ \lfloor s^{-1}\cdot\#([0,t]\cap\cP_i) \rfloor] \geq t/s-1.
	\qquad\text{and}\qquad
	\op{Var}[\lfloor s^{-1}\cdot\#([0,t]\cap\cP_i)\rfloor ] \leq Ct/s^2.
	\eqen
	Using these estimates, Chebyshev's inequality gives that $N(0.1\log n)>0.01s^{-1}n\log n$ w.h.p.
	Using this and Lemma \ref{prop70} we get that the right side of the following inequality converges to 0 as $n\rta\infty$.
	\eqbn
	\P[N(\tau_{\op{consensus}})<0.01s^{-1}n\log n] 
	\leq
	\P[ N(0.1\log n)<0.01s^{-1}n\log n ]+\P[\tau_{\op{consensus}}<0.1\log n].
	\eqen
\end{proof}

\subsection{Synchronous model}
\label{sec:lower-sync}
We will not provide all details of the proof of Theorem \ref{prop:lower} since it is rather similar to the proof in the asynchronous case. Instead we will describe in which ways we need to change the argument.

Our argument will again make use of sets $A(t)\subset\cS$ for $t\in\N\cup\{0 \}$, but the definition and basic properties of the model are somewhat different as compared to the asynchronous case. Let $A(0)$ be the set on two elements defined exactly as in the asynchronous case. The inductive definition of $A(k)$ in terms of $A(k-1)$ is as follows
\eqbn
A(k) =
	\{ \Lambda_1(\sigma,\sigma'),\Lambda_2(\sigma,\sigma')\,:\,\sigma,\sigma'\in A(k-1),\sigma\in\cS' \}
	\cup\{ \Lambda'(\sigma)\,:\, \sigma\in A(k-1) \}.
\eqen
The following lemma can be easily proved by induction on $k$.
\begin{lemma}
	For any $\sigma\in\cS$ and $t\in\N\cup\{0 \}$ we have $\sigma\in A(k)$ if and only if 
	$\P[ \exists i\in[n] \,:\,\sigma(i,k)=\sigma]>0$.
\end{lemma}
Note that unlike in the asynchronous case, the sets $A(k)$ are not necessarily increasing in the synchronous case. Define
\eqbn
\wh A =  \bigcup_{m=0}^{\log\log n-c^{-1}} A(m).
\eqen
The following variant Lemma \ref{prop36} still holds.
\begin{lemma}
	There is a constant $a\in(0,1)$ depending only on $p$ such that with probability at least $1-\exp(-an^{0.05})$ the following holds for $k=0,\dots,\log\log n-c^{-1}$ and $\sigma\in A(t)$
	\eqbn
	\# \{ i\in[n]\,:\,\sigma(i,k)=\sigma  \} \geq a^{2^{k}}n.
	\eqen
	\label{prop36s}
\end{lemma}
The proof is as in the asynchronous case and is therefore omitted. In fact, in the setting of Lemma \ref{prop36s} it is easier to argue concentration, since it is deterministically the case that no node influence more than two nodes (including itself) in one time step, so we do not need to prove Lemma \ref{prop37} and we can apply the standard version of McDiarmid's inequality for deterministically bounded differences.

\begin{lemma}
	In the setting of Theorem \ref{prop:lower} and for $s\leq \log\log\log n-c^{-1}$, there is a $n_0$ depending only on $p$ and $\eps$ such that for $n\geq n_0$ the set $\wh A$ contains no aware states.
	\label{prop39s}
\end{lemma}
\begin{proof}
	Notice that since $A(k)$ is defined in terms $A(k-1)$, the sequence of sets $(A(k))_{k\in\N}$ is eventually periodic. Furthermore, the number of possible values of $A(k)$ is at most $2^s$, which implies that the period is at most $2^s$. We deduce from these properties that if $2^s\leq \log\log n-c^{-1}$ then $\wh A$ is equal to the union of $A(k)$ for all $k\in\N\cup\{0 \}$, so for any $i\in[n]$ and $t\in\R_+$ we have $\sigma(i,t)\in\wh A$. By Lemma \ref{prop36s} and since $a^{2^k}n>1$ we know that all states in $\wh A$ can be found in the model w.h.p.\ at some time $t\leq \log\log n-c^{-1}$, no matter what is the majority bit. We conclude by a similar argument as in the proof of Lemma \ref{prop39} that none of these states can be aware.
\end{proof}

The following lemma holds since the number of communications needed to reach agreement is $\Omega(n)$ for any $s$. 
\begin{lemma} 
	 Theorem \ref{prop:lower} holds for $s\geq 0.1\log\log n$.
	\label{prop29}
\end{lemma}

\begin{proof}[Proof of Theorem \ref{prop:lower}]
	The assertion about aware states follows from Lemma \ref{prop39s}.
	
	By Lemma \ref{prop29}, to prove the bound on the number of communications it is sufficient to prove Theorem \ref{prop:lower} for $s\leq 0.1\log\log n$. Similarly as in the asynchronous case, we consider two cases separately: 	
	(i) $\wh A\cap\cS^{\op{p}}\neq\emptyset$ and 
	(ii) $\wh A\cap\cS^{\op{p}}=\emptyset$. 
	
	Case (i) is treated similarly as in the asynchronous case by applying a coupon collector argument.
	
	In case (ii) first observe that by Lemma \ref{prop36s}, for at least one $\frk b\in\{0,1\}$,  consensus cannot have been reached at time $0.5\log\log n<\log\log n-c^{-1}$ w.h.p., since for all $\sigma\in A(\lceil 0.5\log\log n\rceil)$ there is a node $i\in[n]$ such that $\sigma(i,\lceil 0.5\log\log n\rceil)=\sigma$, and the set $A( \lceil 0.5\log\log n\rceil )$ does not depend on $\frk b$. Since no states in $\wh A$ are passive by assumption, all nodes communicate (as initiator or recipient) at least every $s$ time step. Therefore the number of communications before time $0.5\log\log n$ is at least $\lfloor 0.2s^{-1}n\log\log n-n\rfloor=\Theta(s^{-1}n\log\log n)$, so this is a lower bound for the number of communications needed to reach consensus.
\end{proof}

\section{Synchronous upper bound for $s=C (\log \log n)^2$}
\label{app:uppersync}
In this section we first describe precisely the protocol introduced in Section \ref{sec:intro-sync}, and then we give a detailed analysis of the protocol, which proves Theorem \ref{prop:uppersync}.

\subsection{The protocol}
\label{sec:uppersync-protocol}
 Recall that at any point in time a node is exactly one of the following six types: aspirant, expert, expert candidate, regular, informed, or terminal. 

Define 
$$
M=\lceil 2\log\log n\rceil,\qquad K=\lceil \CCC\log\log n \rceil.
$$
For any node $i\in[n]$ and $t\geq 0$ the state $\sigma(i,t)\in\cS$ of $i$ at time $t\geq 0$ is a tuple of integers such that the first element of the tuple indicates the type $\sigma_1(i,t)\in [6]$ of $i$ at time $t$. The remaining elements of the tuple are as follows for nodes of the various types.
\begin{itemize}
	\item aspirant: $\sigma(i,t)=(1,\eta,d,d',b',b)$, where 
	$\eta\in\{-1,0,\dots,6 \}$ is the future type, 
	$d\in[300\eps^{-2}K]$ is the time counter, 
	$d'\in[K]\cup\{0 \}$ is the trial counter, 
	$b'\in\{-1,0,1 \}$ is the test bit, and 
	$b$ is the belief bit.
	\item expert: $\sigma(i,t)=(2,m,d,b)$, where $m\in[M]$ is the level, $d\in[2K+3]$ is the time counter, and $b$ is the belief bit.
	\item regular: $\sigma(i,t)=(3,d,b)$, where $d\in[3MK]$ is the time counter and $b\in\{0,1 \}$ is the belief bit.
	\item terminal: $\sigma(i,t)=(4,b)$, where $b\in\{0,1 \}$ is the belief bit.
	\item expert candidate: $\sigma(i,t)=(5,b',b)$, where $b'\in\{-1,0,1 \}$ is the  test bit and $b\in\{0,1 \}$ is the belief bit.
	\item informed: $\sigma(i,t)=(6,b)$, where $b$ is the belief bit.
\end{itemize}
Note that if $t\in\N$ then $\sigma(i,t)$ refers to the state of $i$ \emph{after} all updates in time step $t$ are complete. In particular, the function $t\mapsto\sigma(i,t)$ is right-continuous. The state of $i$ immediately before time $t$ is denoted by $\sigma(i,t^-)$, i.e., $\sigma(i,t^-)=\lim_{t' \uparrow t}\sigma(i,t')$.

The protocol is divided into the following phases.

\paragraph{Expert selection phase}
This phase lasts for time $\lceil 300\eps^{-2}K\rceil$ and all nodes are aspirants throughout the phase.  At time $t=0$ each node $i\in[n]$ is an aspirant with state $\sigma(i,0)=(0,-1,1,1,-1,\frk b_i)$, where $\frk b_i$ is the initial bit.  

At each time step throughout this phase all nodes increase their value of $d$ by 1. This allows the nodes to know when the expert selection phase ends and the estimation phase begins.
In the first four time steps, each node's test bit will be set in such a way that the probability of obtaining a test bit $b'=1$ is equal to the probability of obtaining a test bit $b'=0$. 
These test bits will subsequently be used to select experts. 
Many nodes will end up in a third category, with test bit $b'=-1$; this test bit is effectively ignored during the expert selection phase.

At time $t=1$ each aspirant $i$ for which $b=0$ initiates a communication with a uniformly chosen node $j$. Two scenarios can occur: (i) The communication is rejected (since $j$ also initiates a communication or since someone else communicates with $j$ instead), or (ii) $i$ and $j$ communicate. In case (i) (resp.\ (ii)) $i$ sets $b'$ equal to 0 (resp.\ 1). 

At time $t=2$ each aspirant for which $b=0$ initiates a communication with a uniformly chosen node $j$, and again the two scenarios (i) and (ii) can occur. If $b'=1$ and (i) occurs, or if $b'=0$ and (ii) occurs, then the value of $b'$ is left unchanged. Otherwise $b'$ is set to $-1$.

Time steps $t=3,4$ are exactly as time steps $t=1,2$, except that the aspirants for which $b=1$ communicate instead.

In the remainder of this phase, each node will uniformly sample one node at each time step and count how many $b'=0$ bits it observes before encountering a node with $b'=1$. 
If a node encounters $K$ test bits 0 before the first test bit 1, then the node is labelled an expert. Otherwise, it becomes a regular node.
More precisely, at the remaining \emph{even} time steps $t=6,8,\dots,\lceil 300\eps^{-2}K\rceil-2$ of the expert selection phase each node $i\in[n]$ with state $\sigma(i,t^-)=(1,\eta,d,d',b'',b)$ does the following, in addition to increasing $d$ by 1.
\begin{itemize}
	\item If $\eta=-1$, $d'<K$, and $b=0$ then $i$ initiates a communication with some node $j$. Let $b'$ denote the test bit of $j$. If $b'=0$ then $d'$ increases by 1. If $b'=1$ then $\eta$ is set to $3$. 
	\item If $\eta=-1$, $d'=K$, and $b=0$ then $i$ sets $\eta=2$.
	\item Otherwise $i$ does not initiate a communication or update its state (except for increasing $d$ by 1).
\end{itemize}
Nodes which communicate because they were contacted by another node do not update their state.

At the remaining \emph{odd} times $5,7,\dots,\lceil 300\eps^{-2}K\rceil-1$ the same happens, but with the roles of $b=0$ and $b=1$ swapped.

At the end of the expert selection phase (i.e., at the time when all nodes have time counter $d=\lceil 300\eps^{-2}K\rceil$) the following happens for a node $i$ with state  $\sigma(i,t^-)=(1,\eta,\lceil 300\eps^{-2}K\rceil,d',b'',b)$:
\begin{itemize}
	\item Nodes $i$ for which $\eta=2$ become level 0 experts with state $(2,0,1,b)$.
	\item Nodes $i$ for which $\eta\in\{-1,3 \}$ become regular nodes with state $(3,1,b)$.
\end{itemize}

\paragraph{Estimation phase} The estimation phase is divided into $M$ rounds, where each round lasts for time $2K+3$. In round $m\in[M]$ the following happen, where the times $t$ indicates the time relative to the start of the round, so time $t$ means time $\lceil 300\eps^{-2}K\rceil+(m-1)(2K+3)+t$ for the protocol.

At $t=1$ each level $m-1$ expert $i$ initiates a communication with a uniformly chosen node $j$. The node $j$ becomes a level $m$ expert candidate with state $(5,b',b)$, where $b=\wh\sigma(i,t)$ is the belief bit of $i$ and $b'=-1$.

At $t=2$ each level $m-1$ expert $i$ initiates a communication with a uniformly chosen node $j$. If $j$ is an expert candidate then $j$ sets its test bit $b'$ equal to the belief bit $\wh\sigma(i,t)$ of node $i$. If $j$ is not an expert candidate (which means that $j$ is either a regular node or a level $m-1$ expert) then $j$ does not update its state. Expert candidates which are not contacted in this time step become regular nodes.

At $t=3$ each level $m-1$ expert $i$ initiates a communication with a uniformly chosen node $j$. If $j$ is an expert candidate with state $(5,b',b)$ for which $b'\neq -1$ then $j$ becomes a level $m$ expert with state $(2,m,1,b'')$, where $b''$ is the majority bit in $\{b,b',\wh\sigma(i,t) \}$. All nodes $i$ which do not become level $m$ experts in this time step become regular nodes with state $(3,1,b)$ (where $b$ is the belief bit of $i$ immediately before time $t$), so at the end of this time step all nodes are either level $m$ experts or regular nodes.

At $t=4,\dots,2K+3$ the following happen. Each level $m$ expert $i$ initiates a communication with a uniformly chosen node $j$. The node $j$ becomes a level $m$ expert with the same state as $i$, i.e., $\sigma(j,t)=\sigma(i,t)$. A round $m$ expert increases its counter $d$ by $1$ in each time step. This allows it to know when one round ends and the next round starts. 

Regular nodes also increase their counter $d$ by 1 at each time step. Since the total duration of the estimation phase is $M(2K+3)$, the regular nodes' counter will not reach its maximal value $3MK$ in the estimation phase.

At the end of round $M$ all level $M$ experts become informed with state $(6,b)$, where $b$ is the belief bit of the expert.

\paragraph{Pushing phase} At each integer time $t$ each informed node $i$ initiates a communication with some node $j$. If the connection is established and $j$ is \emph{not} a terminal node, then $j$ also becomes informed with the same state as $i$, while the state of $i$ is unchanged, i.e., $\sigma(i,t)=\sigma(i,t^-)$ and $\sigma(j,t)=\sigma(i,t)$. If the communication is rejected (since $j$ initiates a communication and/or communicates with another node instead) or $j$ is a terminal node, then $i$ becomes a terminal node with state $(4,\wh\sigma(i,t^-))$. Terminal nodes never initiate communications and never change their state when being contacted by other nodes, so they satisfy the criteria for terminal states as defined in Section \ref{sec:model}.

Let $\frk T_t\subset[n]$ (resp.\ $\frk I_t\subset[n]$) denote the set of terminal nodes (resp.\ informed nodes) at time $t$ of the pushing phase. Define $T$ to be the first time (relative to the beginning of the pushing phase) at which either the informed nodes or the terminal nodes reach a frequency of 0.1, i.e.,
\eqb
T = \inf\{ t\geq 0\,:\, \#\frk T_t\geq 0.1n \text{\,\,or\,\,} \#\frk I_t\geq 0.1n \}.
\label{eq121}
\eqe
We define the pushing phase to end immediately after time step $T+1$. Note that since the end of the pushing phase is defined by the occurrence of a global event, the nodes will not know when this phase ends.

Regular nodes behave just as in the estimation phase, i.e., they increase their counter $d$ by 1 (modulo $3MK$) at each time step. We will show (Lemma \ref{prop46}) that w.h.p.\ the pushing phase is sufficiently short such that no regular nodes will have a counter which reaches its maximal value $3MK$ during this phase.

\paragraph{Pulling phase} All nodes execute the exact same actions as in the pushing phase (in fact, this has to be the case since the nodes do not know when the pushing phase ends and the pulling phase starts). When the counter $d$ of a regular node $i$ reaches $3MK$ the node $i$ will initiate a communication with a uniformly chosen node $j$. If $j$ is a terminal node then $i$ also becomes a terminal node with the same state as $i$, i.e., $\sigma(i,t)=\sigma(j,t)$. 
If $j$ is not a terminal node then $i$ sets its counter $d$ to 1. In other words, a regular node $i$ initiates a communication every $3MK$ time step until it encounters a terminal node, upon which it becomes a terminal node as well.

\subsection{Analysis}
In this section we will prove Theorem \ref{prop:uppersync} by analyzing the protocol defined in Section \ref{sec:intro-sync}. The theorem is an immediate consequence of Lemmas \ref{prop:memory2}, \ref{prop:correct2}, and \ref{prop:comm2}.

\begin{lemma} 
	For the protocol described in Section \ref{sec:uppersync-protocol} there is a constant $C$ depending only on $\eps$ such that it is sufficient with $\lceil C(\log\log n)^2\rceil$ states of memory per node.
	\label{prop:memory2}
\end{lemma}
\begin{proof}
	The lemma is immediate by considering the memory usage for each of the six types of nodes. Note that aspirants, experts, and regular nodes all require order $(\log\log n)^2$ states of memory. (See the proof of Lemma \ref{prop42} for a similar more detailed calculation.)
\end{proof}

For $k\in[2K]\cup\{0 \}$ and $m\in[M]$ let $\cE^k(m)$ denote the set of experts immediately after step $k+3$ of round $m$. In particular, $i\in\cE^0(m)$ if and only if $i$ became a level $m$ expert upon being contacted by three level $m-1$ experts in the beginning of the $m$th round. Furthermore, a node $i$ communicates in the first three time steps of round $m$ if and only if $i\in\cE^{2K}(m-1)$. Let $\cE^{2K}(0)=\cE(0)$ denote the set of level $0$ experts at the end of the expert selection phase. Define $\alpha^k= n\cdot0.5^{3K}\cdot 2^{k}$ for $k\in[2K]\cup\{0 \}$. For $k=2K$ and $m=0$ set $\alpha=\alpha^{2K}$. Recall that $p\in[1/2+\ep,1-\ep]$ and $p-\frac{1}{2}$ denotes the initial advantage of the majority bit $\frk b\in\{0,1 \}$.
\begin{lemma} 
	W.h.p.\ $|\#\cE(0)-\alpha|<\alpha (\log n)^{-\CCC}$. Furthermore, w.h.p.,
	\eqb
	\begin{split}
		|\#\{i\in\cE(0)\,:\, \frk b_i=\frk b \} - n0.5^Kp|&<0.5\alpha (\log n)^{-\CCC},\\
		|\#\{i\in\cE(0)\,:\, \frk b_i\neq\frk b \} - n0.5^K(1-p)|&<0.5\alpha (\log n)^{-\CCC}.
	\end{split}
	\label{eq137}
	\eqe
	\label{prop53}
\end{lemma}
\begin{proof}
	It is sufficient to prove \eqref{eq137} since the first assertion of the lemma follows from these two estimates.
	Recall that in the expert selection phase (after the first four initial steps) each node $i$ repeatedly asks a uniformly chosen other node for their test bit $b'$. If $i$ observes $K$ test bits $b'=0$ before the first test bit $b'=1$, then it becomes a level 0 expert at the end of the expert selection phase, while if $i$ observes a test bit $b'=1$ before $K$ test bits $b'=0$ then $i$ becomes a regular node. If none of these two events occur before the end of the expert selection phase, then $i$ becomes a regular node. Let $\wh\cE(0)$ denote the set of experts we would have obtained if the expert selection phase lasted for infinitely many time steps, such that a node becomes an expert if and only if it observes $K$ test bits $b'=0$ before the first test bit $b'=1$. 
	This set can differ from $\cE(0)$ because during the expert selection phase, which lasts only $\lceil 300\eps^{-2}K\rceil$ time steps, a node may encounter so many nodes with $b'=-1$ that it neither reaches $K$ nodes with $b'=0$ nor any node with $b'=1$ before the phase is over.
	For $b=0,1$ define
	\eqbn
	\wh\cE^b(0) = \{ i\in\wh\cE(0)\,:\,\frk b_i=b  \}.
	\eqen
	
	When the test bits are defined at the very beginning of the expert selection phase, the following two events are equally likely: 
	(i) happens first and then (ii) (this gives $b'=0$), and 
	(ii) happens first and then (i) (this gives $b'=1$). It follows that if $b'$ is the test bit of a uniformly sampled node, then $\P[b'=0\,|\,b'\in\{0,1 \} ]=\P[b'=1\,|\,b'\in\{0,1 \} ]=1/2$.
	The $b'$ are not completely independent for different nodes. However, an application of McDiarmid's inequality implies that if $V^b_0$ (resp.\ $V^b_1$) for $b=0,1$ is the number of nodes for which $b'=0$ (resp.\ $b'=1$) and $\frk b_i=b$ then w.h.p.\ $|V^b_1/(V^b_0+V^b_1)-0.5|<n^{-0.45}$ for $b=0,1$; denote this event by $E$. Let $\cF$ denote the $\sigma$-algebra generated by $V^b_0$ and $V^b_1$ for $b=0,1$. For any $i\in[n]$ with $\frk b_i=b$ the following holds for $X$ a geometric random variable of parameter $V^b_1/(V^b_0+V^b_1)$ 
	\eqbn
	\P[i\in\wh\cE(0)\,|\,\cF ]
	=\P[X>K\,|\,\cF]
	=\Big(\frac{V^b_0}{V^b_0+V^b_1} \Big)^{K}.
	\eqen
	On the event $E$ the right side deviates from $0.5^K$ by at most $0.5^K\cdot 3n^{-0.45}K$, so 
	$$
	|\E[ \#\wh\cE^{\frk b}(0)]-p\alpha|<0.5^K\cdot 3n^{0.55}K,\qquad
	\E[ \#\wh\cE^{1-\frk b}(0)]-(1-p)\alpha|<0.5^K\cdot 3n^{0.55}K.
	$$ 
	By McDiarmid's inequality applied conditionally on $\cF$, w.h.p.
	\eqb
	|\#\wh\cE^{\frk b}(0)-p\alpha|<n^{0.6},\qquad
	|\#\wh\cE^{1-\frk b}(0)-(1-p)\alpha|<n^{0.6}.
	\label{eq47}
	\eqe
	
	In order to conclude the proof we need to bound $\#(\cE^b(0)\setminus\wh\cE^b(0))$ for $b=0,1$. First we will bound from below the number of nodes for which the test bit $b'=0$ after the first four time steps and for which the initial bit $\frk b_i=0$. By symmetry, the exact same bound holds for $\frk b_i=1$, and by the result of the preceding paragraph, the same bound also holds for test bit $b'=1$. Consider the very first step $t=1$ of the expert selection phase, and let $i\in[n]$ be a node with $\frk b_i=0$ which initiates a communication with a uniformly sampled node $j\in[n]$ at time $t=1$. 
	Recall the two scenarios (i) and (ii) described in the specification of the expert selection phase, and let $\gamma\in\{p,1-p\}\in[\eps,1-\ep]$ denote the fraction of the nodes with initial bit 0. The event (i) occurs if $\frk b_j=0$, which has probability $(\gamma n-1)/(n-1)$. 
		The event (ii) occurs if $\frk b_j=1$ (which has probability $n(1-\gamma)/(n-1)$) and if no other node initiates a communication with $j$ at time $t=1$ (which has probability $(1-(n-1)^{-1})^{\gamma n-1}$).	Combining these estimates and using that $\ep<1/4$ we see that for all sufficiently large $n$, the probability that $b'=0$ (conditioned on the event $\frk b_i=0$) is
	\eqbn
	\geq (\gamma n-1)/(n-1) \cdot n(1-\gamma)/(n-1) \cdot (1-(n-1)^{-1})^{\gamma n-1}
	\geq 0.3\eps.
	\eqen
	
	Every time a node $i$ with $\frk b_i=0$ initiates a communication with a uniformly chosen node $j$ it will set $\eta=3$ if (a) $\frk b_j=1$, (b) no one else initiates a communication with $j$ at the same time, and (c) the test bit of $j$ equals 1. The conditions in (a) and (b) ensure that the two nodes communicate. The event in (a) has probability $>\alpha$ (with $\alpha$ as in the previous paragraph).
	The event in (b) has probability $\geq(1-(n-1)^{-1})^{\alpha n}$ conditioned on the occurrence of the event in (a).
	The event in (c) has probability $>0.3\eps$ conditioned on (a) and (b). Combining these estimates, for sufficiently large $n$ the probability that $i$ will set $\eta=3$ is at least $0.03\eps^2$. Therefore $\P[i\in \wh\cE(0)\setminus\cE(0)]\leq \P[ Y>\lceil 300\eps^{-2}K\rceil-5)/2 ]$, where $Y$ is a geometric random variable with success probability $0.03\eps^2$. This estimate also holds for $\frk b_i=1$ by a similar argument. For all sufficiently large $n$,
	$$
	\P[ Y>(\lceil 300\eps^{-2}K\rceil-5)/2 ]
	=
	(1-0.03\eps^2)^{(\lceil300\eps^{-2}K\rceil-5)/2} 
	< 
	\exp(-0.03\eps^2 \cdot(\lceil 300\eps^{-2}K\rceil-5)/2 )
	<
	0.5^{4K}.
	$$ 
	By Markov's inequality,
	\eqb
	\P[ \#( \wh\cE^{\frk b}(0)\setminus\cE^{\frk b}(0) )>n0.5^{3K} ] 
	\leq 
	\frac{pn \P[i\in \wh\cE^{\frk b}(0)\setminus\cE^{\frk b}(0)] }{n0.5^{3K}}
	< \frac{pn0.5^{4K}}{n0.5^{3K}}
	=p0.5^{K}.
	\label{eq48}
	\eqe
	By a union bound, \eqref{eq47}, and \eqref{eq48},
	\eqbn	
	\begin{split}
		\P[ |\#\cE^{\frk b}(0)-p\alpha|\geq 0.5\alpha (\log n)^{-\CCC} ]
		&\leq
		\P[ |\#\cE^{\frk b}(0)-\#\wh\cE^{\frk b}(0)|\geq n0.5^{3K} ]
		+\P[ |\#\wh\cE^{\frk b}(0)-p\alpha|\geq n^{0.6} ]\\
		&\leq p 0.5^K + \exp(-2n^{0.2}).
	\end{split}
	\eqen
	This implies the first estimate of \eqref{eq137}. The second estimate of \eqref{eq137} follows by a similar argument.
\end{proof}

\begin{remark}	
	No regular node will initiate a communication before time $3MK+\lceil 300\eps^{-2}K\rceil-1$. In particular, if $T<3MK+\lceil 300\eps^{-2}K\rceil-2$, where $T$ is defined as in \eqref{eq121}, then the only phase during which regular nodes will initiate communications is the pulling phase.	
	\label{prop47}
\end{remark}

For all $m\in[M]$ and $k\in[2K]\cup\{0 \}$ define the event $E_m^k$ by 
\eqbn
	E_m^k=\{ |\#\cE^k(m)-\alpha^k|<\alpha^k\beta_m^k \}, \qquad\text{where}\qquad
	\beta_m^k= 5^m (\log n)^{-\CCC} + (k+1) n^{-0.3}.
\eqen
\begin{lemma} 
	W.h.p., the event $E_m^k$ occurs for all $m\in[M]$ and $k\in[2K]\cup\{0 \}$.
\label{prop51b}
\end{lemma}
\begin{proof}
	For all $m\in[M]$ and $k\in[2K]\cup\{0 \}$ let $\cF_m^k$ denote the $\sigma$-algebra generated by all events occurring at or before time step $\lceil 300\eps^{-2}K \rceil+(m-1)(2K+3)+k+3$ of the estimation phase. In particular, if we order these $\sigma$-algebras by increasing information, $\cF_{m-1}^{2K}$ is the first $\sigma$-algebra containing information about the set of experts $\cE^{2K}(m-1)$ which will initiate a communication in steps 1, 2, and 3 of round $m$, while $\cF_m^0$ is the first $\sigma$-algebra containing information about the set of nodes which become level $m$ experts upon receiving a bit in all three steps.
	 
	We will prove that for $m\in[M]$ and $k\in[2K]$,
	\eqb
	\P[ (E_m^k)^c\,|\,\cF_m^{k-1} ]\1_{E_m^{k-1}} \leq 
	\exp\left( -\frac{n^{1.2}}{2\cdot\#\cE^{2K}(m-1)} \right)
	\label{eq110}
	\eqe
	and that for $m\in[M]$,
	\eqb
	\P[ (E_m^0)^c\,|\,\cF_{m-1}^{2K} ]\1_{E^{2K}(m-1)} \leq 
	\exp\left( -\frac{n^{1.2}}{6\cdot\#\cE^{2K}(m-1)} \right).
	\label{eq119}
	\eqe
	These two inequalities combined with Lemma \ref{prop53} and a union bound immediately imply our lemma.
	
	First we prove \eqref{eq110}. We condition on $\cF_m^{k-1}$ throughout the argument. Observe that in step $k+3$ of round $m$ there are $\#\cE^{k-1}(m)$ nodes which are initiating a communication with another node. If all the contacted nodes were regular nodes and no nodes were contacted by two experts, we would have $2\cdot\#\cE^{k-1}(m)$ experts immediately after this time step. However, fewer new experts may be created due to experts contacting nodes which are already experts or several experts contacting the same node. We prove \eqref{eq110} by bounding the number of such collisions.
	
	For $i\in \cE^{k-1}(m)$ let $\wh E(i)$ be the event that the node $j=\frk r(i,\lceil 300\eps^{-2}K\rceil+(m-1)(2K+3)+k+3)$ contacted by $i$ at time $\lceil 300\eps^{-2}K \rceil+(m-1)(2K+3)+k+3$ is not in $\cE^{k-1}(m)$ (so it does not initiate a communication itself at this time) and that no on else contacts $j$ simultaneously as $i$, i.e., 
	\eqbn
	\begin{split}
		\wh E(i)=\,\,& \Big\{ \frk r(i,\lceil 300\eps^{-2}K\rceil+(m-1)(2K+3)+k+3)\\
		&\neq \frk r(j,\lceil 300\eps^{-2}K\rceil+(m-1)(2K+3)+k+3),\,\forall j\in [n] \Big\} \\
		&\cup
		\big\{ \frk r(i,\lceil 300\eps^{-2}K\rceil+(m-1)(2K+3)+k+3)\not\in \cE^{k-1}(m) \big\}.
	\end{split}
	\eqen
	On the event $\wh E(i)$, the node $\frk r(i,\lceil 300\eps^{-2}K \rceil+(m-1)(2K+3)+k+3)$ becomes a new expert due to the communication initiated by $i$. Let $i_1<\dots<i_{\#\cE^{k-1}(m) }$ denote the elements of $\cE^{k-1}(m)$ in increasing order. 
	Conditioned on $\cF_{m}^{k-1}$ and on whether the events $\wh E(i_1),\dots,\wh E(i_{\ell})$ occur for $\ell<\#\cE^{k-1}(m)$, the event $\wh E(i_{\ell+1})$ occurs with probability at least  $1-2n^{-1}\cdot\#\cE^{k-1}(m)$. Therefore the following process is a submartingale
	\eqbn
	\ell \mapsto \sum_{i\in \{i_1,\dots,i_\ell \}} (\1_{\wh E(i)}-(1-2n^{-1}\cdot\#\cE^{k-1}(m))).
	\eqen
	By Azuma's inequality,
	\eqb
	\P\left[ \sum_{i\in \cE^{k-1}(m)} \1_{\wh E(i)} < 
		\#\cE^{k-1}(m)( 1-2n^{-1}\cdot\#\cE^{k-1}(m) )- n^{0.6}\,\Big|\, \cF_m^{k-1} \right] 
		\leq \exp\left( - \frac{n^{1.2}}{2\cdot\# \cE^{k-1}(m)} \right).
		\label{eq111}
	\eqe
	On the \emph{complement} of the event in \eqref{eq111} and on $E_m^{k-1}$,
	\eqb
	\begin{split}
		\#\cE^k(m) 
		&\geq \#\cE^{k-1}(m)+\#\cE^{k-1}(m)\cdot( 1-2n^{-1}\cdot\#\cE^{k-1}(m) )- n^{0.6} \\
		&\geq 2\alpha^{k-1}(1-\beta_m^{k-1})( 1-n^{-1}\alpha^{k-1}(1-\beta_m^{k-1}) ) -n^{0.6} 
		\geq \alpha^k (1-\beta_m^k).
	\end{split}
	\label{eq112}
	\eqe
	Furthermore, on the event $E_m^{k-1}$ it is deterministically the case that
	\eqb
	\#\cE^k(m) 
	\leq 2\cdot\#\cE^{k-1}(m)
	\leq 2\cdot\alpha^{k-1}(1+\beta_m^{k-1})
	\leq \alpha^k(1+\beta_m^k).
	\label{eq113}
	\eqe
	Combining \eqref{eq111}, \eqref{eq112}, and \eqref{eq113} we get \eqref{eq110}.
	
	Now we will prove \eqref{eq119}. Conditioned on $\cF_{m-1}^{2K}$, for any fixed $i\in[n]\setminus \cE^{2K}(m-1)$ the probability that $i$ is \emph{not} contacted by an expert in the first (resp., second, third) time step of round $m$ is $(1-(n-1)^{-1})^{\#\cE^{2K}(m-1)}$. Therefore, if $n^{-1}\cdot\#\cE^{2K}(m-1)<0.01$, for $n$ sufficiently large and on the event $E_m^{k-1}$ there is a $\delta_0\in\R$ satisfying $|\delta_0|\leq  0.6 n^{-1}\cdot\#\cE^{2K}(m-1)$ such that
	\eqbn
	\P[ i\in\cE^0(m)\,|\,\cF_{m-1}^{2K} ] 
	= 
	\Big(1-(1-(n-1)^{-1})^{\#\cE^{{2K}}(m-1)}\Big)^3
	=
	(n^{-1}\cdot\#\cE^{2K}(m-1))^3\cdot(1 + \delta_0).
	\eqen
	Using this and that $n^{-2}(\alpha^{2K})^3=\alpha^0$ it follows that on the event $E_{m-1}^{2K}$ and with $\delta\in[-\beta_{m-1}^{2K},\beta_{m-1}^{2K}]$ chosen such that
	$\#\cE^{2K}(m-1)=\alpha^{2K}(1+\delta)$,
	\eqbn
	\begin{split}
		\E[ \#\cE^0(m) \,|\,\cF_{m-1}^{2K} ]
		&= (n - \#\cE^{2K}(m-1))\big( n^{-1}\alpha^{2K}(1+\delta) \big)^3\cdot(1+\delta_0)\\
		&= \alpha^0 \Big(1+3\delta+\delta_0-n^{-1}\cdot \#\cE^{2K}(m-1
		+ O\big( (n^{-1}\cdot \#\cE^{2K}(m-1))^2 + (\beta_{m-1}^{2K})^2 \big) \Big).
	\end{split}
	\eqen
	Conditioned on $\cF_{m-1}^{2K}$, the random variable $\#\cE^0(m)$ can be written as a function of the random variables $\frk r(i,t)$ for $i\in \cE^{2K}(m-1)$ and $t$  the three first time steps of round $m$. Changing one of these random variables $\frk r(i,t)$ changes $\#\cE^0(m)$ by at most 1. Therefore McDiarmid's inequality gives
	\eqb
	\P\big[ |\#\cE^0(m)-\E[\#\cE^0(m)\,|\,\cF^{2K}_{m-1}] | > n^{0.6} \,|\, \cF^{2K}_{m-1} \big]
	\leq \exp\left( -\frac{n^{1.2}}{6\cdot\#\cE^{2K}(m-1)} \right).
	\label{eq114}
	\eqe
	On $E_{m-1}^{2K}=\{|\delta|<\beta_{m-1}^{2K} \}$  
	and the event in \eqref{eq114}, the event $E_m^0$ occurs for $n$ sufficiently large, since
	\eqbn
	|\#\cE^0(m)-\alpha^0| 
	\leq n^{0.6}+\alpha^0 \Big(3|\delta|+|\delta_0|+n^{-1}\cdot \#\cE^{2K}(m-1)
	+ O\big( (n^{-1}\cdot \#\cE^{2K}(m-1))^2 + (\beta_{m-1}^{2K})^2 \big) \Big) 
	\leq \alpha^0\beta_m^0.
	\eqen
\end{proof}

A node $i\in[n]$ is defined to be \emph{incorrect} at time $t\geq 0$ if its belief bit is different from the majority bit, i.e., $\wh\sigma(i,t)\neq\frk b$. For $m\in[M]$ let $\ell_m$ be the number of experts which are incorrect at the end of round $m$, i.e., if $t=\lceil 300\eps^{-2}K\rceil+m({2K}+3)$ is the time at which round $m$ ends then
\eqbn
\ell_m = \#\{i\in\cE^{2K}(m-1)\,:\, \wh\sigma(i,t)\neq\frk b \}.
\eqen
Similarly, let $\ell_0$ denote the number of experts which are incorrect at the end of the expert selection phase, i.e., if $t=\lceil 300\eps^{-2}K\rceil$ is the time at which the expert selection phase ends then
\eqbn
\ell_0 = \#\{i\in\cE^{2K}(0)\,:\, \wh\sigma(i,t)\neq\frk b \}.
\eqen
For $m\in[M]\cup\{0 \}$ let $\delta_{m}$ be the fraction of level $m$ experts that have an incorrect estimate for $\frk b$ at the end of round $m$ (or, for $m=0$, at the end of the expert selection phase), i.e., 
$$
\delta_m = \frac{\ell_m}{\#\cE^{2K}(m)}.
$$
\begin{lemma} 
	For $\ell_{m-1}\geq n^{0.25}$,
	\eqb
	\P[\ell_{m} \geq 2^{2K}(3\ell_{m-1}^2/n + \ell_{m-1}^{0.55})\,|\, \cF_{m-1}^{2K}] 
	\leq \exp \Big(-\frac{\ell_{m-1}^{0.1}}{6}\Big).
	\label{eq109}
	\eqe
	For $\delta_{m-1}\geq n^{-0.25}$,
	\eqb
	\P[\delta_{m} \geq (3 \delta_{m-1}^2 - 2\delta_{m-1}^{3})(1+ (\beta_m^{2K})^{0.5} ) \1_{E_M^{2K}} \,|\, \cF_{m-1}^{2K}]\1_{E_{m-1}^{2K}} \leq \exp\Big( -\frac{n^{1.1}}{6\cdot\#\cE^{2K}(m-1) } \Big).
	\label{eq108}
	\eqe
	\label{prop:error}
\end{lemma}
\begin{proof}
	First we prove \eqref{eq109}. For any fixed node $i$ which is not a level $m-1$ expert the probability that this node is contacted by an incorrect expert in the first (resp., second, third) time step of round $m$ is equal to $1-(1-(n-1)^{-1})^{\ell_{m-1}}\leq \ell_{m-1}/(n-1)$. 
	Therefore the probability that $i$ is contacted by an incorrect expert in at least two of these three time steps is
	\eqbn
	3\ell_{m-1}^2/(n-1)^2 - \ell_{m-1}^3/(n-1)^3 \leq 3\ell_{m-1}^2/n^2.
	\eqen
	Let $\ell'_{m}$ be the number of nodes which are contacted by an expert with an incorrect belief bit in at least two of the three time steps, and note that 
	$\E[\ell'_{m}\,|\,\cF_{m-1}^{2K} ]\leq 3\ell_{m-1}^2/n$. Observe that $\ell'_m$ is an upper bound for the number of incorrect round $m$ experts at the end of step 3 of round $m$. Conditioned on $\cF_{m-1}^{2K}$, the random variable $\ell'_{m}$ is a function of $\frk r(i,t)$ for $i\in \cE^{2K}(m-1)$ an incorrect expert and $t$ equal to each of the first three steps of round $m$. Changing one of the $3\ell_{m-1}$ random variables $\frk r(i,t)$ changes $\ell'_{m}$ by at most 1. Therefore McDiarmid's inequality gives
	\eqb
	\P[ \ell'_{m} - \E[\ell'_{m}\,|\,\cF_{m-1}^{2K} ]\geq\ell_{m-1}^{0.55} \,|\,\cF_{m-1}^{2K} ]
	\leq \exp \Big(-\frac{\ell_{m-1}^{1.1}}{6\ell_{m-1}}\Big)
	= \exp \Big(-\frac{\ell_{m-1}^{0.1}}{6}\Big).
	\label{eq115}
	\eqe
	The number of incorrect experts at most doubles in each time step from 4 to ${2K}+3$ of round $M$. Therefore the following holds on the \emph{complement} of the event in \eqref{eq115} 
	\eqbn
	\ell_{m}< 2^{2K}\ell'_{m-1} < 2^{2K}(3\ell_{m-1}^2/n + \ell_{m-1}^{0.55}),
	\eqen
	which implies \eqref{eq109}.
	
	Now we will prove \eqref{eq108}. Let $\ell''_{m}$ denote the number of incorrect experts at the end of time step 3. (Note that $\ell''_m\leq\ell'_m$ in general, since in order to become an incorrect expert a node needs to receive two incorrect bits, and in addition a third (correct or incorrect) bit, while $\ell'_m$ counts the number of nodes satisfying the first of these two requirements.) Then $\ell''_{m}$ is a function of $\frk r(i,t)$ for $i\in\#\cE^{2K}(m)$ and $t$ equal to each of the first three steps of round $m$. Since changing one of these $3\cdot\#\cE^{2K}(m-1)$ random variables $\frk r(i,t)$ changes $\ell''_{m}$ by at most 1, McDiarmid's inequality gives
	\eqb
	\P[ \ell''_m -\E[ \ell''_m\,|\,\cF_{m-1}^{2K} ] > n^{0.55} \,|\,\cF_{m-1}^{2K} ]
	\leq
	\exp\Big( -\frac{n^{1.1}}{6\cdot \#\cE^{2K}(m-1)} \Big).
	\label{eq49}
	\eqe
	An expert in $\cE^0(m)$ obtains its belief bit by taking the majority bit among the three bits received from level $m-1$ experts. 
	The expert is incorrect 
	if it receives exactly two incorrect bits or 
	if it receives three incorrect bits. Since the three bits are independent, letting $\delta=\#\cE^{2K}(m-1)/n$ we have
	\eqbn
	\begin{split}
		\E[ \ell''_{m}&\,\,|\,\cF_{m-1}^{2K} ] \\
		=\,\,&
		3\cdot(n-\#\cE^{2K}(m-1))\cdot\Big( 1- (1-(n-1)^{-1})^{\ell_{m-1}} \Big)^2\cdot
		\Big( 1- (1-(n-1)^{-1})^{\#\cE^{2K}(m-1)-\ell_{m-1}} \Big)\\
		&+
		(n-\#\cE^{2K}(m-1))\cdot\Big( 1- (1-(n-1)^{-1})^{\ell_{m-1}} \Big)^3\\
		=\,\,& 3n(1-\delta)
		\cdot
		(\delta\delta_{m-1})^2(1+O(\delta\delta_{m-1}))
		\cdot
		\delta(1-\delta_{m-1})(1+O(\delta)) \\
		&+ n(1-\delta)\cdot(\delta\delta_{m-1})^3(1+O(\delta\delta_{m-1}))\\
		=\,\,&\delta^3n(3\delta_{m-1}^2-2\delta_{m-1}^3)(1+O(\delta)).
	\end{split}
	\label{eq50}
	\eqen
	It follows that on the \emph{complement} of the event in \eqref{eq49}, and if the events $\delta_{m-1}\geq n^{-0.25}$, $E_{m-1}^{2K}$, and $E_m^{2K}$ also occur, for all sufficiently large $n$,
	\eqbn
	\begin{split}
		\delta_{m} 
		=    \frac{\ell_{m}}{\#\cE^{2K}(m)} 
		\leq \frac{\ell''_{m}2^{2K}}{\alpha^{2K}(1-\beta^{2K}_{m})}
		&\leq (3\delta_{m-1}^2-2\delta_{m-1}^3)\big(1+O(\delta)+\beta_m^{2K}+2(\beta_m^{2K})^2\big)\\
		&<(3\delta_{m-1}^2-2\delta_{m-1}^3)\big(1+(\beta_m^{2K})^{0.5}\big),
	\end{split}
	\eqen
	Using \eqref{eq49} this implies \eqref{eq108}.
\end{proof}

\begin{lemma} 
	W.h.p.\ all level $M$ experts have a belief bit equal to $\frk b$.
	\label{prop52}
\end{lemma}
\begin{proof}
	For any $\delta\in[0,1]$  
we define the stopping time $T(\delta)$ by
	\eqbn
	T(\delta) = \inf\{m\in[M]\cup\{0 \}\,:\, \delta_m \leq \delta  \},
	\eqen
	where we let the infimum of an empty set be $\infty$. We will argue that w.h.p., for $n$ sufficiently large,
	\eqbn
	\begin{split}
		&(i)\,\, T(0.1)<0.01\log\log n,\qquad\qquad\qquad\,\,\,
		(ii)\,\, T(n^{-0.25})-T(0.1)<1.98\log\log n,\\
		&(iii)\,\, T(n^{-0.75})-T(n^{-0.25})\leq 3,\qquad\qquad
		(iv)\,\, T(0)-T(n^{-0.75})\leq 1.
	\end{split}
	\eqen
	Combining these four bounds and using that $M=\lceil 2\log\log n \rceil$ we immediately get the lemma.
	
	By a union bound and Lemmas \ref{prop51b} and \ref{prop:error}, we may assume the \emph{complements} of the events considered in \eqref{eq109} and \eqref{eq108} occur for all $m$, and that $E_m^k$ occurs for all $m\in[M]$ and $k\in[{2K}]\cup\{0 \}$. Define $f_n:[0,1]\to\R$ by
	\eqbn
	f_n(\delta) = (3 \delta^2 - 2\delta^{3})(1+(\beta_m^{2K})^{0.5}).
	\eqen
	Note that $f_n(\delta_{m-1})$ appears on the right side of the inequality defining the event in \eqref{eq108}, so we can bound the fraction of incorrect experts by recursively applying $f_n$.
	Also note that the function $f_n$ depends on $n$ since $\beta_{m}^{2K}$ depends on $n$.
	
	(i) Lemma \ref{prop53} gives that $\delta_0<1/2-\ep/2$ w.h.p. We condition on this high probability event throughout the proof of (i). We have $T(0.1)=0$ if $\delta_0\leq 0.1$, so we assume $\delta_0>0.1$. For some $n_0$,
	\eqbn
	v 
	:= 
	\min_{n>n_0,\delta \in [0.1,1/2-\ep/2)} \delta - f_n(\delta) 
	= 
	\min_{n>n_0} \min \left \{1/2-\eps/2-f_n(1/2-\eps/2),\, 0.1 - f_n(0.1) \right \}
	> 0.
	\eqen
	The equality holds because $f_n$ is strictly increasing and convex in $(0,1/2)$, so the minimum is attained at one of the endpoints of the interval.
	Note that $v$ represents the smallest decrease in $\delta_n$ we can obtain from one round of the protocol, i.e., one application of the function $f_n$.
	Therefore $T(0.1)\leq (1/2-\eps/2-0.1)/v<0.01\log\log n$ for $n$ large enough. 
	
	(ii) Observe that $f_n(\delta)<\delta^{1.5}$ for $\delta<0.1$ and $n$ sufficiently large. This completes the proof since $0.1^{1.5^{1.98\log\log n}}<n^{-0.25}$ for all sufficiently large $n$.
	
	(iii) If the event in \eqref{eq109} occurs for all $m$ and  $\ell_m\in[n^{0.25},n^{0.75}]$, then $\ell_{m+1}<2^{2K}(3 (n^{0.75})^2/n +(n^{0.75})^{0.55} ) < n^{0.5+0.0001}$, and further 
	$\ell_{m+2}<2^{2K}(3 (n^{0.5+0.0001})^2/n +(n^{0.5+0.0001})^{0.55} ) < n^{0.28}$ and 
	$\ell_{m+3}<2^{2K}(3 (n^{0.28})^2/n +(n^{0.28})^{0.55} ) < n^{0.25}$, so $T(n^{-0.75})-T(n^{-0.25})\leq 3$.
	
	(iv) If the probability that a uniformly sampled expert is incorrect is smaller than $n^{-0.75}$, then there are at most $n^{-0.75} \cdot \# \cE^k(m)<n^{-0.75}\cdot n=n^{0.25}$ incorrect experts. 
	The probability that any given node receives a bit from two such experts is at most $3\cdot n^{-0.75}\cdot n^{-0.75}(1+o(1))$. By a union bound, w.h.p.\ there will be no nodes which receive two such bits.
\end{proof}

Let $T_0:=\lceil 300\eps^{-2}K\rceil+(2K+3)M$ be the time at which the estimation phase ends and the pushing phase begins. Recall the time $T$ defined in \eqref{eq121}.
\begin{lemma} 
	W.h.p., $T-T_0\leq 10\cdot\CCC\log\log n$. Furthermore, there is a constant $q>0$ depending only on $\eps$ such that w.h.p.,
	\begin{itemize}
		\item[(i)] there are at least $qn$ terminal nodes at time $T+1$, and
		\item[(ii)] the number of communications during $(T_0,T+1]$ is at most $n$.
	\end{itemize}
	\label{prop46}
\end{lemma}
\begin{proof}
	Let $\cF_t$ be the $\sigma$-algebra generated by all events and states before or at time $t$. When an informed node $i$ initiates a communication with a node $j$, either a new informed node will be created (if $j$ is \emph{not} terminal and the connection is established) or $i$ becomes a terminal node (if $j$ is terminal, or if the connection is not established since $j$ is communicating with someone else and/or initiate a communication itself). 
 Assume $t<T\wedge (3MK+\lceil 300\ep^{-2}K \rceil)$ and recall Remark \ref{prop47}. The probability that $j$ is terminal is $\leq 0.1$, while the probability that $j$ initiates a communication and/or communicates with someone else is $\leq 2\frk I_{t-1}/n \leq 0.2$. 
 	The probability that $i$ communicates with a node $j$ that is not terminal and the connection is established (thereby adding a new informed node) is therefore at least 0.7, and the complement has probability at most 0.3, which causes $i$ to be removed from the set of informed nodes. 
	Therefore the expected number of  \emph{additional} informed nodes created due to the communication initiated by $i$ is $\geq 0.7-0.3=0.4$. 
	Therefore, on the event $t<T\wedge (3MK-\lceil 300\ep^{-2}K \rceil)$,
	\eqbn
	\E[\#\frk I_t\,|\,\cF_{t-1} ] \geq 1.4\cdot\#\frk I_{t-1}.
	\eqen
	Given $\cF_{t-1}$ the random variable $\#\frk I_t$ is determined by the random variables $\frk r(i,t)$ for $i\in\frk I_{t-1}$. Since changing one $\frk r(i,t)$ can change $\#\frk I_t$ by at most 2, McDiarmid's inequality gives that
	\eqb
	\P[ \#\frk I_t-\E[\#\frk I_t\,|\,\cF_{t-1} ]<-0.1\cdot\#\frk I_{t-1} \,|\,\cF_{t-1}  ]
	\leq \exp\Big( -\frac{2\cdot(0.1\cdot\#\frk I_{t-1})^2}{ 2^2 \cdot\#\frk I_{t-1} } \Big)
	= \exp ( - 0.005 \cdot\#\frk I_{t-1} ).
	\label{eq120}
	\eqe
	By Lemma \ref{prop51b}, w.h.p., $\#\frk I_{T_0}\geq \alpha^{2K}(1-\beta_M^{2K})\geq n2^{-{K}-1}$. On this event and on the complement of the event in \eqref{eq120} for $t=T_0+1,\dots,T$ we have $\#\frk I_{t-1}>n^{0.99}$ 
	for all such $t$. By using this, \eqref{eq120}, and that
	$$
	n2^{-{K}-1}\cdot 1.3^{\lfloor 10\cdot\CCC\log\log n\rfloor-1}>n,
	$$
	we see that 
	\eqbn 
	 \P[ T-T_0>\lfloor 10\cdot\CCC\log\log n\rfloor-1 ] 
	 \leq
	 (\lfloor 10\cdot\CCC\log\log n\rfloor-1) \exp ( - 0.005 n^{0.99} )
	 +\P[ \#\frk I_0<\alpha^{2K}(1-\beta_M^{2K}) ].
	\eqen
	In particular, $T-T_0\leq \lfloor 10\cdot\CCC\log\log n\rfloor-1$ w.h.p. By Remark \ref{prop47}, regular nodes do not initiate any communications in the pushing phase on this event.
	
	By the definition of $T$ we know that at least one of the following two events must occur: (a) $\#\frk I_T\geq 0.1n$ and (b) $\#\frk T_T\geq 0.1n$. In case (b) it is immediate that assertion (i) of the lemma is satisfied. In case (a) the expected number of terminal nodes at time $T+1$ is at least $\exp(-0.09) \#\frk I_T(\#\frk I_T-1)/(n-1)>0.005n$ for sufficiently large $n$, since an informed node becomes a terminal node if it initiates a communication with an informed node (probability $(\#\frk I_T-1)/(n-1)$) and if no other informed nodes initiate a communication with the same informed node (probability $(1-(n-1)^{-1})^{\#\frk I_T-2}\geq \exp(-0.09)$ for sufficiently large $n$). By McDiarmid's inequality it follows that (i) holds w.h.p.\ also in case (a).
	
	To prove (ii) notice that at each time step $T_0+1,...,T_0+T-1$ the number of informed nodes $\#\frk I_t$ grows by at least a factor of 1.3 w.h.p., so $\#\frk I_{T-t}\leq 0.1n\cdot 1.3^{-t+1}$. Furthermore, $\#\frk I_T\leq 2\cdot\#\frk I_{T-1}<0.2n$. The number of communications initiated at time $t+1$ is equal to $\#\frk I_{t}$ for all $t$. Therefore the number of communications initiated by informed nodes at or before time $T+1$ is at most
	\eqbn
	0.2n+\sum_{j=0}^{\infty} 0.1n\cdot 1.3^{-j}\leq n.
	\eqen
	Since only informed nodes initiate communications in the pushing phase w.h.p., this proves (ii).
\end{proof}

\begin{lemma} 
	The protocol reaches terminal consensus in finite time w.h.p., i.e., $\tau_{\op{terminal}}<\infty$ w.h.p.
	\label{prop:correct2}
\end{lemma}
\begin{proof}
	By Lemma \ref{prop51b}, w.h.p.\ there will be at least one informed node at the beginning of the pushing phase. On this event the protocol terminates a.s., in the  sense that all nodes eventually become terminal nodes. The belief bit of all the terminal nodes originate from a level $M$ expert. Therefore, by Lemma \ref{prop52}, w.h.p.\ all terminal nodes will have belief bit $\frk b$. Combining the above we get $\tau_{\op{terminal}}<\infty$ w.h.p.
\end{proof}

\begin{lemma} 
	There exists a constant $C$ depending only on $\eps$ such that w.h.p.\ the communication cost before terminal consensus is smaller than $Cn$, i.e., $N_{\op{terminal}}<Cn$.
	\label{prop:comm2}
\end{lemma}
\begin{proof}
	We estimate the cost in each phase separately.
	
	In the first four steps of the expert selection phase there are $\Theta(n)$ communications. Let $i\in[n]$ and let $b'\in\{-1,0,1 \}$ be the value of the test bit of $i$ after the first four time steps. Then $\P[b'=1]=(1-p)p\exp(-p)(1+O(1/n))$ (resp.\ $\P[b'=1]=(1-p)p\exp(-(1-p))(1+O(1/n))$) if $\frk b_i$ is equal (resp.\ is not equal) to the majority bit and $p$ is the fraction of nodes for which $\frk b_{i} = \frk b$. 
	 
	 In either case $\P[b'=1]>0.18\ep$. By McDiarmid's inequality the number of nodes $i$ for which the test bit $b'=1$ after the first four time steps is at least $0.17\ep n$ w.h.p. On this event, in the remainder of the expert selection phase the number of communications of each fixed node $i$ is stochastically dominated by a geometric random variable with parameter $0.17\ep-1/n$, independently for each $i$, since each node becomes a terminal node with probability at least $0.17\ep-1/n$ every time it initiates a communication. By concentration of the sum of independent geometric random variables, it follows that the number of communications in the expert selection phase is $O(n)$ w.h.p.

	By Remark \ref{prop47}, w.h.p.\ all communications in the estimation phase are initiated by experts. W.h.p.\ there are $O( n(\log n)^{-\CCC} )$ experts at any given time. Since the duration of the estimation phase is $\Theta((\log\log n)^2)$, we get $o(n)$ communications. 
	
	By Lemma \ref{prop46}(ii) there are at most $2n$ communications in the pushing phase w.h.p.
	
	By Lemma \ref{prop46}(i), w.h.p., every time a node $i$ (informed or regular) initiates a communication with a uniformly chosen node $j$ in the pulling phase the probability that $j$ is a terminal node is at least $q$. Therefore the number of communications initiated by $i$ is stochastically dominated by a geometric random variable with success probability $q$. By concentration of the sum of independent geometric random variables, we get that the number of communications in the protocol after time $T$ of the pushing phase is at most $2q^{-1}n$ w.h.p.
\end{proof}

\section{Asynchronous upper bound for $s = C (\log \log n)^3$}
\label{app:upperasync2}

In this section we first describe precisely the protocol introduced in Section \ref{sec:intro-async2}, and then we give a detailed analysis of the protocol, which proves Theorem \ref{prop:upperasync2}.

\subsection{The protocol}\label{ssec:upper1}
We define the protocol by specifying the behavior of the nodes of the various types: aspirant, expert, regular node, terminal node, expert candidate, and informed. For each type we describe (i) what the node does when its clock rings (i.e., whether it initiates a communication with another node and how it updates its state), and (ii) what the node does when another node contacts it (i.e., how the node updates its state when this happens). In Section \ref{sec:intro-async2} we give a more intuitive but less complete description of the protocol than here, and we strongly encourage to read that section before this one.

For each node $i$ we write the state $\sigma(i,t)\in\cS$ of $i$ at time $t\geq 0$ as a tuple of integers such that 
the first element of the tuple indicates the type of $i$ at time $t$, the second to last element of the tuple is the initial bit $\frk b_i\in\{0,1 \}$, and
the last element of the tuple is the belief bit $\wh\sigma(i,t)\in\{0,1 \}$.
Otherwise the form of the tuple depends on the type. We let $\sigma_1(i,t)\in[6]$ denote the type of $i$ at time $t$. Let $\sigma(i,t^-)=\lim_{t' \uparrow t}\sigma(i,t')$ denote the state of $i$ infinitesimally before time $t$. Define
\eqbn
\begin{split}
	&M = \lceil 2\log\log n\rceil,\qquad \qquad\qquad\qquad\qquad\quad
	K=\lceil 6\log\log n\rceil,\\ 
	&\Tasp=\lceil 5000\eps^{-1}\log\log n\rceil,\\
	&t_m^\init = 6\Tasp + 7K(m-1) + K,\qquad\qquad
	t_m = 6\Tasp + 7Km.
\end{split}
\eqen
First we describe the form of the state of a node for each of the different types of nodes.
\begin{itemize}
\item If $i\in[n]$ is an \emph{aspirant} at time $t\geq 0$ then the state of $i$ is of the form $\sigma(i,t)=(1,d,\xi,\chi,b',b'',b''',\frk b_i,b)$, where 
$d\in[\Tasp]$ is the counter, 
$\xi\in[4]\cup\{0 \}$ is the expert type,
$\chi\in[3]$ is the phase,
$b'\in\{-1,0,1 \}$ is the first test bit,  
$b''\in\{-1,0,1 \}$ is the second test bit, and
$b'''\in\{-1,0,1 \}$ is the third test bit.
\item If $i\in[n]$ is an \emph{expert} at time $t\geq 0$ then the state of $i$ is of the form $\sigma(i,t)=(2,m,d,\xi,\frk b_i,b)$, where  
	$m\in [M]\cup\{0 \}$ is the level number,
	$d\in[2K+7]$ is the time counter, and
	$\xi\in[4]$ is the expert type. 
\item If $i\in[n]$ is a \emph{regular node} at time $t\geq 0$ then the state of $i$ is of the form $\sigma(i,t)=(3,d,\xi,\psi,\frk b_i,b)$, where $d\in[\lceil(\log\log n)^2\rceil]$ is the time counter, $\xi\in[4]$ is the expert type, and $\psi\in\{0,1 \}$ is the expert indicator. 
\item If $i\in[n]$ is a \emph{terminal node} at time $t\geq 0$ then the state of $i$ is of the form $\sigma(i,t)=(4,\frk b_i,b)$. 
\item If $i\in[n]$ is an \emph{expert candidate} at time $t\geq 0$ then the state of $i$ is of the form $\sigma(i,t)=(5,m,d,\xi,b^1,b^2,b^3,\frk b_i,b)$, where $m\in[M]$ is the level number, $d\in[2t_M]$ is the time counter, $\xi\in[4]$ is the expert type, and $b^1,b^2,b^3\in\{-1,0,1 \}$ are the test bits. 
\item If $i\in[n]$ is an \emph{informed node} at time $t\geq 0$ then the state of $i$ is of the form $\sigma(i,t)=(6,\frk b_i,b)$.
\end{itemize}

The \emph{initial data} of the protocol are as follows. At time 0 each node $i\in[n]$ is an aspirant with state of the form
$\phi(i,0)=(1,1,0  ,1   ,-1,-1 ,-1  ,\frk b_i,\frk b_i)$, where $\frk b_i$ is the initial bit assigned to $i$.

First we describe the behavior of \emph{aspirants}. An aspirant first determines the value of its expert type $\xi$ (phase $\chi=1$), then it is determined whether the aspirant will become a level 0 expert or not (phase $\chi=2$), and then, on the event that the aspirant will become a level 0 expert, the aspirant waits for an additional $\Tasp$ clock rings (phase $\chi=3$). In phase $\chi=1$, the aspirant repeatedly collects a tuple of four bits by asking four randomly chosen nodes for their initial bit. The first time it gets a tuple with exactly one or three bits 1, the expert type is determined by considering the order of 0's and 1's in the tuple. In phase $\chi=2$ the aspirant repeatedly collects pairs of bits by asking two randomly chosen nodes for their initial bit. If the aspirant collects $K$ pairs $(0,1)$ before the first pair $(1,0)$ then it will become a level 0 expert; otherwise it turns into a regular node. The following is a more precise description of the behavior of aspirants. We assume $i\in[n]$ is an aspirant at time $t\geq 0$ with state $(1,d,\xi,\chi,b',b'',b''',\frk b_i,b)$. 
\begin{itemize}
	\item[(i)] When the clock of $i$ rings and the phase $\chi=1$ then $i$ initiates a communication with a uniformly chosen node $j$. The following describes how $i$ updates its state based on the state of $j$. 
	\begin{itemize}
		\item If $b'=-1$ then $i$ sets $b'=\frk b_j$.
		\item If $b'\neq -1$ and $b''=-1$ then $i$ sets $b''=\frk b_j$.
		\item If $b'\neq -1$, $b''\neq -1$, and $b'''=-1$ then $i$ sets $b'''=\frk b_j$.
		\item If $b', b'',b''' \neq -1$ then consider the tuple $(b',b'',b''',\frk b_j)$. If this tuple has three bits 0 and one bit 1, or three bits 1 and one bit 0, then let $v\in[4]$ be the position of the bit which is different from the other bits. Set the expert type $\xi=v$ and the phase $\chi=2$. If the tuple $(b',b'',b''',\frk b_j)$ does not satisfy the mentioned condition, set $b',b'',b'''$ all equal to $-1$.
	\end{itemize}
	\item[(i')] When the clock of $i$ rings and the phase $\chi=2$ then $i$ initiates a communication with a uniformly chosen node $j$. If $d\leq K$ then $i$ updates its state as follows.
	\begin{itemize}
		\item If $b'=\frk b_j$ then $i$ sets $b'=-1$.
		\item If $b'=-1$ then $i$ sets $b'=\frk b_j$.
		\item If $b'=0$ and $\frk b_j=1$ then $d$ increases by 1 and $i$ sets $b'=-1$.
		\item If $b'=1$ and $\frk b_j=0$ then $i$ becomes a regular node with state $(3,1,\xi,0,\frk b_i,\frk b_i)$, where $\xi$ is the expert type of $i$ immediately before the clock was ringing.
	\end{itemize}
Note that $d$ counts the number of times the event described in the third item above happens. If $d=K+1$ then $i$ sets $\chi=3$ and $d=1$. 
	\item[(i'')] When the clock of $i$ rings then the following happens if the phase $\chi=3$. If $d<\Tasp$ then $d$ increases by 1. If $d=\Tasp$ then $i$ becomes an expert with state $(2,0,1,\xi,\frk b_i,\frk b_i)$, where $\xi$ is the expert type of $i$ immediately before the Poisson clock was ringing. 
	\item[(ii)] If another node $j$ initiates a communication with $i$ then $i$ will not change its state.  
\end{itemize}
If $i\in[n]$ is an \emph{expert} at time $t\geq 0$ then the following holds. 
\begin{itemize}
	\item[(i)] When the clock of $i$ rings, $i$ initiates a communication with another node $j$. If the time counter $d$ of $i$ equals $2K+7$ and the level $m<M$, then $i$ will transform into a regular node with state $(3,1,\xi,1,\frk b_i,b)$ immediately after the communication, where $\xi$ (resp.\ $b$) is the expert type (resp.\ belief bit) of $i$ immediately before the communication. If $d=2K+7$ and $m=M$ then $i$ transforms into an informed node immediately after the communication. 
	Notice that $m$'s growth is governed by the rules for regular nodes.
	If $d\neq 2K+7$ then the time counter $d$ will increase by 1 and the other elements of the tuple describing the state remain unchanged.		
	\item[(ii)] If another node $j$ initiates a communication with $i$ then $i$ will not change its state.
\end{itemize}
If $i\in[n]$ is a \emph{regular node} at time $t\geq 0$ then the following holds.
\begin{itemize}
	\item[(i)] When the clock of $i$ rings then $i$ initiates a communication with a uniformly sampled node $j$ if and only if the time counter $d=\lceil(\log\log n)^2\rceil$. If $j$ is \emph{not} a terminal node then $i$ sets its time counter $d=1$. If $j$ is a terminal node then $i$ becomes a terminal node with the same state as $j$, i.e., $\sigma(i,t)=\sigma(j,t)$. If $i$ does not initiate a communication with another node (since $d\neq \lceil(\log\log n)^2\rceil$) then $i$ increases its time counter $d$ by $1$. In other words, a regular node will initiate a communication every $\lceil(\log\log n)^2\rceil$ clock rings until it encounters a terminal node, upon which it will also become a terminal node.
	\item[(ii)] If a node $j$ contacts $i$ then the state of $i$ is updated as follows:
	\begin{itemize}
		\item If $j$ is an expert with state $\sigma(j,t^-)=(2,m,d,\xi',\frk b_j,b)$  for $d\neq 2K+7$, and if the expert indicator $\psi$ of $i$ satisfies $\psi=0$, then $i$ becomes an expert with state $(2,m,d+1,\xi,\frk b_i,b)$, where $\xi$ is the expert type of $i$ immediately before the communication. In other words, if $j$ is an expert and $i$ has not previously been an expert (since $\psi=0$) then $i$ will become an expert of the same level and with the same belief bit as $j$, but with counter $d$ increased by 1 (assuming the counter satisfies $d\neq 2K+7$). 
		\item If $j$ is an expert with state $(2,m,2K+7,\xi',\frk b_j,b)$ for $m<M$ and $\xi\in[3]$, 
		and if the expert indicator $\psi$ of $i$ satisfies $\psi=0$,
		then $i$ becomes an expert candidate with state $(5,m+1,1,\xi,b^1,b^2,b^3,\frk b_i,b)$, where $b^{\xi'}=b$, $b^{\xi''}=-1$ for $\xi''=[3]\setminus\{\xi' \}$, and $\xi$ is the expert type of $i$ immediately before the communication. In particular, $i$ becomes an expert candidate of level $m+1$ upon being contacted by a level $m$ expert with counter $d=2K+7$ and expert type $\xi\in[3]$. 
		\item If $j$ has state $(2,M,2K+7,\xi,\frk b_j,b)$ then $i$ becomes an informed node with state $(6,\frk b_i,b)$. 
		\item If $j$ has state $(6,\frk b_j,b)$ then $i$ becomes an informed node with state $(6,\frk b_i,b)$. 
	\end{itemize}
	Note that in all four cases $i$ adopts the belief bit $b$ of $j$ (except in the first two cases for $\psi=1$).
\end{itemize} 
If $i\in[n]$ is an \emph{expert candidate} at time $t\geq 0$ then the following holds.
\begin{itemize}
	\item[(i)] When the clock of $i$ rings $i$ will not initiate a communication. If the time counter $d<2t_M$ then $d$ is increased by 1. If $d=2t_M$ then $i$ becomes a regular node with state $(3,1,\xi,1,\frk b_i,b)$, where $b$ (resp.\ $\xi$) is the belief bit (resp.\ expert type) of $i$ immediately before the clock was ringing. In other words, an expert candidate will remain an expert candidate for at most $2t_M$ clock rings, and if its clock rings $2t_M$ times before it has turned into an expert (see (ii) right below) then it will turn into a regular node.
	\item[(ii)] If $i$ has state $(5,m,d,\xi',b^1,b^2,b^3,\frk b_i,b)$ and is contacted by a node $j$ then the state of $i$ updates as follows:
	\begin{itemize}
		\item If $j$ is an expert with state $(2,m-1,2K+7,\xi,\frk b_j,b')$ for $\xi\in[3]$, and if $b^\xi=-1$, then $i$ sets $b^\xi=b'$. If $b^1,b^2,b^3$ are all different from $-1$ after this update then $i$ becomes a level $m$ expert with state $(2,m,1,\xi',\frk b_i,b_0)$, where $b_0$ is the majority bit in $\{b^1,b^2,b^3\}$. In other words, an expert candidate $i$ becomes a level $m$ expert if it has received bits from three level $m-1$ experts with counter $2K+7$ and expert type 1,2,3, respectively, and the belief bit of $i$ will be the majority belief bit among these three level $m-1$ experts.
		\item Otherwise the state of $i$ does not change.
	\end{itemize} 
\end{itemize}
If $i\in[n]$ is an \emph{informed node} at time $t\geq 0$ then the following hold.
\begin{itemize}
	\item[(i)] When the clock of $i$ rings it initiates a communication with a uniformly chosen node $j$. If $j$ is informed or terminal then $i$ becomes a terminal node with unchanged belief bit. Otherwise $i$ does not update its state. 
	\item[(ii)] If another node $j$ initiates a communication with $i$ then $i$ will not change its state.
\end{itemize}
If $i\in[n]$ is a \emph{terminal node} then it does not initiate communications and it does not update its state if it is contacted by other nodes. Notice that a terminal node is in a terminal state as defined in Section \ref{sec:model}, i.e., a terminal node has a state in the set $\cS_\infty$ defined in that section.

\subsection{Analysis}

\begin{lemma}
		For the protocol described in Section \ref{ssec:upper1} there is a constant $C$ depending only on $\eps$ such that $C\lceil\log\log n\rceil^3$ states of memory per node suffice.
\end{lemma}
\begin{proof}
	This is immediately verified by calculating the memory need for each of the six types of nodes. We notice that the expert candidates require the most memory. Namely, by multiplying the number of states allowed in each element of the tuple specified above, we see that an expert candidate must be able to store the following number of states 
	\eqbn
	6 \cdot
	M\cdot
	2t_M\cdot
	4\cdot
	3^3\cdot
	2\cdot
	2=\Theta( (\log\log n)^3 ). 
	\eqen
\end{proof}

For any $m\in[M]\cup\{0 \}$ let $\cE(m)$ denote the set of level $m$ experts, i.e.,
\eqbn
\begin{split}
	\cE(m)=
	\{ i\in[n]\,:\,\exists t\geq 0 \text{\,\,such\,\,that\,\,}  \sigma(i,t)=(2,m,2K+7,\xi,\frk b_i,b),\, b\in\{0,1\},\xi\in[4] \}.  
\end{split}
\eqen
\begin{remark}
	Notice that the sets $\cE(m)$ for $m\in[M]$ are disjoint, since a node can only be an expert once, which follows by using that when an expert converts into a regular node the regular node will have expert indicator $\psi=1$, and a regular node with expert indicator $\psi=1$ cannot become an expert or expert candidate.
	\label{rmk2}
\end{remark}
For $m\in[M]$ let $\cE^\init(m)\subset\cE(m)$ denote the set of level $m$ experts which became experts upon receiving three bits from level $m-1$ experts, i.e.,
\eqbn
\begin{split}
	\cE^\init(m)=&\,\,
	\{ i\in[n]\,:\,\exists t\geq 0 \text{\,\,such\,\,that\,\,}\sigma(i,t)=(2,m,1,\xi,\frk b_i,b),\, b\in\{0,1\},\xi\in[4] \}. 
\end{split}
\eqen

We will now define what it means that an expert is \emph{premature}. We will let $\frk P\subset\bigcup_{m\in[M]\cup\{0 \}}\cE(m)$ denote the set of premature experts. A level 0 expert $i$ is premature if it became an expert before time $\Tasp/2$, i.e.,  
\eqbn
\frk P\cap \cE(0) = \{i\in\cE(0)\,:\,\exists t<\Tasp/2
\,\,\text{such that}\,\, \sigma_1(i,t)=2 \}.
\eqen
We now define $\frk P\cap \cE(m)$ inductively. Given $\frk P\cap \cE(m-1)$ we first define $ \frk P\cap \cE^\init(m)$ and then we define $\frk P\cap (\cE(m)\setminus\cE^\init(m))$. Let $i\in\cE^\init(m)$. Then the belief bit of $i$ was determined by taking the majority bit among the bits received from three level $m-1$ experts $\frk i_1(i),\frk i_2(i),\frk i_3(i)\in\cE(m-1)$. We say that $i$ is premature if at least one of these three nodes is premature, i.e., 
\eqbn
\frk P\cap \cE^\init(m) = \{i\in\cE^1(m)\,:\, \{\frk i_1(i),\frk i_2(i),\frk i_3(i) \}\cap\frk P \neq\emptyset  \}.
\eqen
Let $i\in\cE(m)\setminus\cE^\init(m)$. Then $i$ first became an expert upon being contacted by some node $\frk i_4(i)\in\cE(m)$. We say that $i$ is premature if $\frk i_4$ is premature, i.e.,
\eqb
\frk P\cap (\cE(m)\setminus\cE^\init(m)) = \{i\in\cE(m)\setminus\cE^1(m)\,:\, \frk i_4(i)\in\frk P \}.
\eqe
Note that $\frk P$ is well-defined since a node can only be an expert once, see Remark \ref{rmk2}. Recall that an expert $i$ spreads its majority bit $\frk b$ by repeatedly contacting other nodes $j$. If $i$ is premature the bit may be spread to fewer other nodes since it may be more likely that $j$ is an aspirant. Part of our analysis in this section involves bounding from above the number of premature nodes, which guarantees that sufficiently many experts are created.

Next we define the set $\cA\subset[n]$ of \emph{non-expiring} nodes. 
Recall that an expert candidate will only remain an expert candidate for $2t_M$ clock rings, and will turn into a regular node if it has not received three expert bits before this happens. In order to bound from below the number of experts, we bound from above the number expert candidates whose clock is ringing at least $2t_M$ times before time $t_M$. The purpose of introducing non-expiring nodes is to keep track of experts candidates for which the clock rings \emph{less than} $2t_M$ times before time $t_M$ and experts originating from such expert candidates, since it is easier to bound from below the number of such experts. 
We say that a node $i$ is non-expiring if at least one of the following three criteria are satisfied: (i) The Poisson clock of $i$ rings at most $2t_M$ times during $[0,t_M]$ and $i$ is an expert candidate at some point in time, i.e.,
\eqbn
\exists t\geq 0\text{\,\,such\,\,that\,\,}\sigma_1(i,t)=5\text{\qquad and\qquad }
\#(\cP_i\cap[0,t_M])\leq 2t_M,
\eqen
(ii) $i$ is a level $m$ expert for some $m \in [M]$ due to receiving a bit from a non-expiring level $m$ expert, i.e.,
$$
i\in\bigcup_{m\in[M]}\cE(m)\setminus\cE^\init(m)
\text{\qquad and\qquad }
\frk i_4(i)\in\cA,
$$ 
or (iii) $i\in\cE(0)$.
Note that if a non-expiring expert candidate receives bits from three level $m-1$ experts before time $t_M$ and is a regular node with $\psi=0$ when it receives the first of these three bits, then it will become a level $m$ expert. We will show that for an expert candidate $i$ the event in (i) happens w.h.p., which helps us to lower bound the number of non-expiring experts and expert candidates.

Let $\cE_-(m)\subset\cE(m)$ be the set of level $m$ experts created before time $t_m$ which are non-expiring and \emph{not} premature, i.e.,
\eqbn
\begin{split}
	\cE_-(m)=&\,\,
	\{ i\in[n]\,:\,\exists t\in[0,t_m] \text{\,\,such\,\,that\,\,}\sigma(i,t)=(2,m,2K+7,\xi,\frk b_i,b),\, b\in\{0,1\},\xi\in[4],\\
	&\qquad i\not\in\frk P,\text{\,\,and\,\,} i\in\cA \}. 
\end{split}
\eqen
Let $\cE^\init_-(m)\subset\cE_-(m)\cap \cE^1(m)$ denote the set of level $m$ experts which became experts upon receiving three bits from level $m-1$ experts, 
which became experts before time $t_m^\init$, and 
which are non-expiring and not premature, i.e.,
\eqbn
\begin{split}
	\cE_-^\init(m)=&\,\,
	\{ i\in[n]\,:\,\exists t\in[0,t_m^\init] \text{\,\,such\,\,that\,\,}\sigma(i,t)=(2,m,1,\xi,\frk b_i,b),\, b\in\{0,1\},\xi\in[4],\\
	&\qquad i\not\in\frk P,\text{\,\,and\,\,} i\in\cA. 
\end{split}
\eqen
Define 
\eqbn
L(m) = \#\cE(m),\quad 
L^\init(m) = \#\cE^\init(m),\quad
L_-(m) = \#\cE_-(m),\quad
L^\init_-(m) = \#\cE^\init_-(m)
\eqen
and
\eqbn
\alpha = n0.5^K,\qquad \beta_{m}=5^m(\log n)^{-\CCC},\qquad
\alpha^\init = n0.5^{3K},\qquad 
\beta^\init_{m}=3\cdot 5^{m-1}(\log n)^{-\CCC}+(\log n)^{-\CCC}.
\eqen

We will use the following estimates for a Poisson random variable $X\sim\op{Pois}(\lambda)$ with $\lambda>0$ multiple times throughout this section. See e.g.\ \cite[Theorem A.1.15]{alon2004probabilistic} for a proof.
\eqb
\P[ X\geq 2\lambda ]\leq (e/4)^\lambda <2^{-0.55\lambda},\qquad
\P[X\leq\lambda/2]\leq \exp(-\lambda/8)<2^{-0.18\lambda}. 
\label{eq:poisson}
\eqe
The following basic estimate for a geometric random variable $Y\sim\op{Geom}(\mu)$ will also be used multiple times.
\eqb
\P[Y\geq x] = \exp(-\mu x)<2^{-1.44\mu x},\qquad x>0.
\label{eq:geometric}
\eqe

The following lemma says that the number of level 0 experts is very close to $\alpha$. Furthermore, it says that the set $\cE_-(0)$ contains almost all level 0 experts. 
\begin{lemma}[Initial bound, number of experts] For all sufficiently large $n$,
	\eqbn
	\P[|L_-(0)-\alpha|>\beta_0\alpha ]\leq 3\exp(-2n^{0.2}),\qquad
	\P[|L(0)-\alpha|>\beta_0\alpha ]\leq 2\exp(-2n^{0.2}).
	\eqen
	\label{prop48}
\end{lemma}
\begin{proof}
	Recall that an aspirant $i$ with phase $\chi=2$ repeatedly collects pairs of bits, and that $i$ becomes an expert if and only if it observes $K$ bit pairs $(0,1)$ before the first bit pair $(1,0)$. Since $i$ is equally likely to observe a pair $(0,1)$ and a pair $(1,0)$ each time it collects a bit pair, it becomes a level 0 expert with probability exactly $0.5^K$. Furthermore, the events $\{i\in\cE(0) \}$ are independent for different $i$. 
By Hoeffding's inequality,
	\eqb
	\P[ |\#\cE(0)-\alpha|>n^{0.6} ] \leq 2\exp(-2n^{0.2}).
	\label{eq122}
	\eqe
	This implies the second inequality of the lemma. 
	
	To prove the first inequality of the lemma, we will bound from above $\#(\cE(0)\setminus\cE_-(0))$. For $i\in[n]$ and $v=1,2,3$ let $\tau_i^v\geq 0$ denote the time at which $i$ exits phase $\chi=v$ as an aspirant, i.e., letting $\chi(i,t)\in[3]$ denote the phase of $i$ at time $t$ on the event that $\sigma_1(i,t)=1$ (and setting $\chi(i,t)=0$ if $\sigma_1(i,t)\neq 1$), we have
	\eqbn
	\tau_i^v = \sup\{t\geq 0\,:\, \chi(i,t)=v \}.
	\eqen
	If $i\in\cE(0)\setminus\cE_-(0)$ then, by definition of $\cE_-(0)$, at least one of the following four events occur: 
	(i) $\#(\cP_i\cap[0,6\Tasp])\leq 3\Tasp$, 
	(ii) $\#(\cP_i\cap[0,\tau_i^1])\geq \Tasp$,    
	(iii) $\#(\cP_i\cap(\tau_i^1,\tau_i^2])\geq \Tasp$, and
	(iv) $i\in\frk P$.
	We will bound the probability of the events (i)-(iv) from above. 
	
	To bound the probability of the event in (i), we use \eqref{eq:poisson} to obtain the following
	\eqbn
	\P[ \#(\cP_i\cap[0,6\Tasp])\leq 3\Tasp ]\leq 2^{-0.18\cdot 6\Tasp}<(\log n)^{-20000}.
	\eqen
	
	To bound the probability of the event in (ii), recall that in phase $\chi=1$ of the aspirant phase a node collects bit quadruples $(b',b'',b''',b'''')$ repeatedly, and phase $\chi=1$ ends the first time that this quadruple contains exactly one or three bits 0. For each quadruple and sufficiently large $n$ this constraint is satisfied with probability at least $4\ep(1-\ep)^3>\ep/2$, so $\#(\cP_i\cap[0,\tau_i^1])$ is stochastically dominated by 4 times a geometric random variable $Y$ with success probability $\ep/2$. Therefore we get upon an application of \eqref{eq:poisson},
	\eqbn
	\P[\#(\cP_i\cap[0,\tau_i^1])\geq \Tasp]
	\leq \P[Y>\Tasp/4]
	\leq 2^{-1.44 \cdot \ep/2 \cdot \Tasp/4}
	\leq(\log n)^{-900},
	\eqen
	where we have used \eqref{eq:geometric} to get the second to last inequality.
	
	To bound the probability of the event in (iii), recall that in phase $\chi=2$ of the aspirant phase a node collects bit pairs $(b',b'')$ repeatedly, and phase $\chi=2$ ends at the latest at the first time that $(b',b'')=(1,0)$. For each pair and sufficiently large $n$, each time $i$ collects a bit pair we have $(b',b'')=(1,0)$ with probability at least $\ep(1-\ep)(1-1/n)>\ep/2$, so $\#(\cP_i\cap(\tau_i^1,\tau_i^2])$ is stochastically dominated by 2 times the geometric random variable $Y$ considered above. Therefore we get the following
	\eqbn
	\P[\#(\cP_i\cap(\tau_i^1,\tau_i^2])\geq \Tasp]
	\leq \P[Y>\Tasp/2]
	\leq 2^{-1.44 \cdot \ep/2 \cdot \Tasp/2}
	\leq(\log n)^{-1800}.
	\eqen
	
	To bound the probability of the event in (iv), recall that an aspirant which becomes a round 0 expert has a phase $\chi=3$ which consists of $\Tasp$ clock rings. In particular, in order for a round 0 expert to be in $\frk P$ its clock must ring at least $\Tasp$ times during the time interval $[0,\Tasp/2]$, so \eqref{eq:poisson} gives
	\eqbn
	\begin{split}
		\P[ i\in\cE(0)\cap\frk P ]
		&=\P[ i\in\cE(0)]\cdot\P[i\in\frk P \,|\, i\in\cE(0)]
		<0.5^K\cdot\P[\#(\cP_i\cap[0,\Tasp/2])>\Tasp]\\
		&<0.5^K\cdot 2^{-0.55\cdot \Tasp/2}<(\log n)^{-5500}. 
	\end{split}
	\eqen
	
	Combining the bounds for the events (i)-(iv) above, we get that for each $i\in[n]$,
	\eqbn
	\P[i\in\cE(0)\setminus\cE_-(0)]<2(\log n)^{-900}.
	\eqen
	Since the event $i\in\cE(0)\setminus\cE_-(0)$ happens independently for each $i$, Hoeffding's inequality gives that except on an event of probability $\exp(-(\log n)^{-1800}n)$ for sufficiently large $n$, we have $\#(\cE(0)\setminus\cE_-(0))<3(\log n)^{-900}n$. Combining this with \eqref{eq122} and using that $\cE_-(0)\subset\cE(0)$ gives the first inequality of the lemma.
\end{proof}

The following lemma says that there are few aspirants at time $\Tasp/2$. Recall that for $i\in[n]$ and $t\geq 0$ we have $\sigma_1(i,t)=1$ if and only if node $i$ is an aspirant at time $t$.
\begin{lemma}
	Define the event $D_1$ by 
	\eqbn
	D_1= \Big\{ \#\{ i\in[n]\,:\,\sigma_1(i,\Tasp/2)=1\} <2n(\log n)^{-6}\Big\}.
	\eqen
	Then $D_1$ happens w.h.p.
	\label{prop51}
\end{lemma}
\begin{proof} 
	An aspirant can be in three phases $\chi=1,2,3$. In phase $\chi=1$ the aspirant collects bit quadruples $(b',b'',b''',b'''')$, and if exactly one or exactly three of the bits in the quadruple are equal to 0 then the aspirant proceeds to phase $\chi=2$. Let $Y_1$ denote the number of quadruples collected by the aspirant in phase $\chi=1$, so $Y_1=\#(\cP_i\cap[0,\tau_i^1])/4$ in the notation of Lemma \ref{prop48}. In phase $\chi=2$ the aspirant collects bits pairs $(b',b'')$. It becomes a regular node the first time it observes a pair $(1,0)$, and it will eventually become a level 0 expert if it observes $K$ pairs $(0,1)$ before the first pair $(1,0)$. Let $Y_2$ denote the number of bit pairs collected in phase $\chi=2$, so $Y_2=\#(\cP_i\cap(\tau_i^1,\tau_i^2])/2$ in the notation of Lemma \ref{prop48}. We observed in the proof of Lemma \ref{prop48} that both $Y_1$ and $Y_2$ are stochastically dominated by a geometric random variable $Y$ with success probability $\ep/2$.
	
	Let $E(i)$ denote the event that $i$ becomes a regular node immediately after the aspirant phase (i.e., $i\not\in\cE(0)$), and that the second part of its aspirant phase ends after time $\Tasp/2$ (i.e., $\sigma_1(i,\Tasp/2)=1$). Due to the upper bound on the number of experts established in Lemma \ref{prop48}, in order to conclude the proof it is sufficient to show that w.h.p., there are at most $2n(\log n)^{-6}-\alpha-\beta_0\alpha$ nodes $i$ for which $E(i)$ occurs. 
	The clock of $i$ either rings at least $\Tasp/4$ times before time $\Tasp/2$, or this does not happen. On the former event, and if $E(i)$ occurs, then there are either more than $\Tasp/4$ clock rings during $[0,\tau^1_i]$ (so $Y_1>\Tasp/32$) or more than $\Tasp/4$ clock rings during $(\tau^1_i,\tau^2_i]$ (so $Y_2>\Tasp/16$). By a union bound,  \eqref{eq:poisson}, and \eqref{eq:geometric} 
	\eqbn
	\begin{split}
		\P[ E(i) ]
		&\leq \P[ \#(\cP_i\cap[0,\Tasp/2])\leq\Tasp/4 ]
		+ \P[ Y_1>\Tasp/32 ] + \P[Y_2>\Tasp/16]\\
		&\leq 
		2^{-0.18\cdot \Tasp/2 }+
		2^{-1.44\cdot \ep/2\cdot \Tasp/32 }+
		2^{-1.44\cdot \ep/2\cdot \Tasp/16 }
		<(\log n)^{-100}.
	\end{split}
	\eqen
	Since $E(i)$ occurs independently for each node, this estimate and Hoeffding's inequality gives that w.h.p., there are at most $2(\log n)^{-200}n<2n(\log n)^{-6}-\alpha-\beta_0\alpha$ nodes $i$ for which $E(i)$ occurs, which concludes the proof of the lemma.
\end{proof}

The following lemma says that there are few premature level 0 experts, i.e., there are few level 0 experts which are created before time $\Tasp/2$.
\begin{lemma}
	\eqbn
	D_6 = \{ \#(\cE(0)\cap\frk P ) < \alpha(\log n)^{-6} \}. 
	\eqen
	Then $D_6$ happens w.h.p.
\end{lemma}
\begin{proof}
	Since the clock of an aspirant must ring at least $\Tasp$ times in order for the node to become a level 0 expert, and since $\cP_i$ is independent of the event $\{ i\in\cE(0) \}$ for each $i$, for all sufficiently large $n$,
	\eqbn
	\begin{split}
		\P[ i\in \cE(0)\cap\frk P ]
		&= \P[ i\in \cE(0) ]\cdot \P[i\in\frk P\,|\,i\in \cE(0) ]
		\leq \alpha/n \cdot \P[ \#(\cP_i\cap[0,\Tasp/2])>\Tasp ]\\
		&\leq \alpha/n \cdot 2^{-0.55\cdot \Tasp/2}
		\leq \alpha/n\cdot (\log n)^{-5000}.
	\end{split}
	\eqen
	where we use \eqref{eq:poisson} to bound $\P[ \#(\cP_i\cap[0,\Tasp/2])>\Tasp ]$. Since the events $\{i\in\cE(0)\cap\frk P \}$ are independent for different $i$, the lemma follows by Hoeffding's inequality.
\end{proof}

The following lemma says that the expert type of each node $i\in[n]$ is uniformly distributed on $[4]$, conditional on all information until the time at which the expert initiates a communication which may lead to the creation of a new expert candidate.
\begin{lemma}
	Let $t$ be a time at which an expert $i$ with counter $2K+7$ initiates a communication. Let $\cF$ be the $\sigma$-algebra generated by $\frk r(i,t)$ and all information before time $t$, except for the expert type $\xi$ of $i$. Then $\xi$ is independent of $\cF$ and satisfies 
	\eqbn
	\P[\xi=1\,|\,\cF]=\P[\xi=2\,|\,\cF]=\P[\xi=3\,|\,\cF]=\P[\xi=4\,|\,\cF]=1/4.
	\eqen
	\label{prop50}
\end{lemma}
\begin{proof}
	In the aspirant phase the value of $\xi$ is determined by letting the aspirant repeatedly sample bit quadruples $(b^1,b^2,b^3,b^4)$, and setting $\xi=1$ (resp.\ $\xi=2,3,4$) if this quadruple equals $(1,0,0,0)$ or $(0,1,1,1)$ (resp.\ 
	$(0,1,0,0)$ or $(1,0,1,1)$;
	$(0,0,1,0)$ or $(1,1,0,1)$;
	$(0,0,0,1)$ or $(1,1,1,0)$). Notice that the probability of sampling each quadruple is the same (Von Neumann unbiasing). Therefore 
	\eqbn
	\P[\xi=1]=\P[\xi=2]=\P[\xi=3]=\P[\xi=4]
	=1/4.
	\eqen
	
	We have $\P[\xi=u]=\P[\xi=u\,|\,\cF]$ for $u=1,2,3,4$ since the value of $\xi$ does not influence the state of any nodes other than $i$ until (at the earliest) at time $t$. Here we use that $i$ can only be an expert once, see Remark \ref{rmk2}.
\end{proof}

The following lemma upper bounds the number of experts.
\begin{lemma}[Upper bound, number of experts]
	Let $D_2$ denote the event that for all $m\in[M]$, we have $L(m)<\alpha(1+\beta_m)$. Then $D_2$ happens w.h.p.
	\label{prop41}
\end{lemma}
\begin{proof}
	For $\xi=1,2,3,4$ let $\cE_\xi(m)$ denote the set of level $m$ experts with expert type $\xi$, i.e.,
	\eqbn
	\begin{split}
		\cE_\xi(m)=
		\{ i\in[n]\,:\,\exists t\geq 0 \text{\,\,such\,\,that\,\,}  \sigma(i,t)=(2,m,2K+7,\xi,\frk b_i,b),\, b\in\{0,1\} \},
	\end{split}
	\eqen
	and set $L_\xi(m) = \#\cE_\xi(m)$. 
	
	For $m\in[M]\cup\{0 \}$ let $\wh E(m)$ denote the following event
	$$
	\wh E(m)=\{ L_\xi(m)<\alpha(1+\beta_m)/4,\,\xi\in[4] \}.
	$$ 
	Notice that $\wh E(m)\subset\{L(m)<\alpha(1+\beta_m) \}$. Let $\cG_{m-1}$ denote the $\sigma$-algebra generated by $\cE_\xi(m-1)$ for $\xi\in[4]$. We will show the following for $m=1,\dots,M$ and sufficiently large $n$, which immediately implies the lemma by Lemma \ref{prop48} and a union bound
	\eqbn
	\P[ \wh E(m)^c\,|\,\cG_{m-1}]\1_{\wh E(m-1)} 
	<\exp(-2n^{0.05}).
	\eqen
	Fix $m\in\{1,\dots,M \}$ and condition on $\cG_{m-1}$. For each $i\in\cE(m-1)$ let $s_i$ denote the time that $i$ initiates a communication and has counter $2K+7$ immediately prior to initiating the communication. For each $i\in[n]$ let $\wh E_i$ denote the event that $i$ receives bits from three level $m-1$ experts with counter $2K+7$ and with expert type 1, 2, and 3, respectively, i.e., with $\xi(j)$ denoting the expert type of node $j\in[n]$ (and $\xi(j)=0$ if $j$ does not have an expert type),
	\eqbn
	\begin{split}
		\wh E_i = \{ \exists i_1,i_2,i_3\in\cE(m-1)\,:&\,s_{i_1}<s_{i_2}<s_{i_3},\, i=\frk r(i_1,s_{i_1})=\frk r(i_2,s_{i_2})=\frk r(i_3,s_{i_3}) \\
		& \{ \xi(i_1),\xi(i_2),\xi(i_3) \}=\{1,2,3 \}  \}.
	\end{split}
	\eqen
 	Then for sufficiently large $n$ and $i\not\in \cE(m-1)$, on the event $\wh E(m-1)$,
	\eqbn
	\begin{split}
		\P[ \wh E_i&\,|\,\cG_{m-1} ] \\
		&= 
		\Big( 1- (1-(n-1)^{-1})^{L_1(m-1)} \Big)
		\Big( 1- (1-(n-1)^{-1})^{L_2(m-1)} \Big)
		\Big( 1- (1-(n-1)^{-1})^{L_3(m-1)} \Big)\\
		&\leq \frac{L_1(m-1)L_2(m-1)L_3(m-1)}{(n-1)^3}
		<0.5^{3K+6}(1+3.5\beta_{m-1}).
	\end{split}
	\eqen
	Let $\wh L=\#\{ i\in[n]\,:\, \wh E_i \text{\,\,occurs} \}$. 
	Observe that $\wh L$ is a function of the random variables $\frk r(i,s_i)$ for $i\in\cE(m-1)$, and that changing one of these $\#\cE^\init(m)$ random variables changes $\wh L$ by at most 1. Therefore McDiarmid's inequality gives
	\eqbn
	\begin{split}
		 \P\Big[ \wh L> 0.5^{3K+6}&(1+3.5\beta_{m-1})n+n^{0.55}\,|\,\cG_{m-1} \Big]\1_{\wh E(m-1)} \\
		&\leq \exp\left( -\frac{2n^{1.1}}{L(m-1)} \right)\1_{\wh E(m-1)}
		\leq 
		\exp\left(-\frac{2n^{1.1}}{\alpha(1+\beta_{m-1})} \right).
	\end{split}
\label{eq118}
	\eqen
	Observe that $\cE^\init(m)\subset\{i\in[n]:\wh E_i\text{ occurs} \}$. These sets are not identical e.g.\ because $i$ may not be eligible to become an expert, or because a node can only remain an expert candidate for a bounded number of clock rings. Furthermore, recall that each $i\in \cE^\init(m)$ gives at most $2^{2K+6}$ level $m$ experts, so $L(m)<2^{2K+6}\cdot \#\cE^\init(m)\leq 2^{2K+6}\wh L$. Therefore, on the complement of the event in \eqref{eq118} and for sufficiently large $n$,
	$$
	L(m)
	\leq 2^{2K+6}\big( 0.5^{3K+6}(1+3.5\beta_{m-1})n+n^{0.55} \big)
	\leq \alpha(1 + 3.6\beta_{m-1}).
	$$
	By Lemma \ref{prop50}, conditioned on $\cE(m)$ the random variables $\xi(i)$ for $i\in\cE(m)$ are independent and uniformly distributed on $[4]$. The lemma now follows from \eqref{eq119} and Hoeffding's inequality.
\end{proof}

Recall that each type of node is associated with a number (aspirants correspond to 1, experts correspond to 2, etc.).
For $t\geq 0$ and $u\in[6]$ let $N_u(t)$ denote the number of nodes of type $u$ at time $t$, i.e.,
\eqbn
N_u(t) = \# \{i\in [n]\,:\, \sigma_1(i,t)=j \}.
\eqen
\begin{remark}
	Since each expert creates at most one expert candidate, on the event $D_2$ we have the following bounds for the number of experts and expert candidates, respectively, for any $t\geq 0$ 
	\eqbn
	N_2(t) \leq (M+1)\alpha(1+\beta_m),\qquad
	N_5(t) \leq (M+1)\alpha(1+\beta_m).
	\eqen
	\label{rmk1}
\end{remark}

We say that $j\in[n]$ is a \emph{potential level $(m,d)$ expert} if we can find a sequence of 
pairs $(i_1,t_1),\dots,(i_{d},t_{d})$ such that $t_k\in\cP_{i_k}$ for all $k$, $i_1\in \cE^\init_-(m)$, $i_{d}=j$, and for each $k\in[d-1]$ we have 
$i_{k+1}\in\{i_k,\frk r(i_k,t_k) \}$. If there are $v\in\N$ such sequences we say that $j$ is a potential level $(m,d)$ expert of multiplicity $v$. Observe that the number of potential level $(m,d)$ experts (counted with multiplicities) is exactly $2^{d-1}L^\init_-(m)$. Also observe that 
all level $m$ experts are potential level $(m,2K+7)$ experts while the opposite is not necessarily true. We say that $j$ is a \emph{late} potential level 
$(m,d)$ expert if $\sum_{k=1}^{d-1} (t_{k+1}-t_k)>6K$.

The next lemma bounds from above the number of late potential level $(m,d)$ experts. Eventually we want to bound the number of late experts (rather than the number of \emph{potential} late experts) but the estimate is easier for potential experts since the random variables $t_{k+1}-t_k$ have the law of independent unit rate exponential random variables in this case.
\begin{lemma} 
	For $m\in[M]$ and $d=[2K+7]$ let $D(m,d)$ be the event that the number of late potential level $(m,d)$ experts (counted with multiplicity) is smaller than $(\log n)^{-6}2^{d-1}L^1_-(m)$. Then the event $D_3:=\bigcap_{m=1}^M \bigcap_{d=1}^{2K+7} D(m,d)$ occurs w.h.p.
\end{lemma}
\begin{proof}
	Fix $m\in[M]$ and $d\in[2K+7]$. Let $j\in[n]$ be a randomly chosen potential level $(m,d)$ expert such that the probability of sampling a given $j$ is proportional to its multiplicity as a potential level $(m,d)$ expert. Let $(i_1,t_1),\dots,(i_{d},t_{d})$ be as described above. Then the random variables $t_{k+1}-t_k$ have the law of independent unit rate exponential random variables. 
	Let $E_j$ denote the event that $j$ is a potential level $(m,d)$ expert, and let $\wh E_j$ denote the event that $j$ is a \emph{late} potential level $(m,d)$ expert. Then the following holds by \eqref{eq:poisson} for $X_1,\dots,X_{d-1}$ independent unit rate exponential random variables and $Y$ a Poisson random variable with expectation $6K$ 
	\eqbn
	\P[ \wh E_j\,|\,E_j ] 
	= \P\left[ \sum_{k=1}^{d-1} X_k>6K \right]
	= \P\left[ Y<d-1 \right]
	<  \P\left[ Y<3K \right]
	\leq (\log n)^{-6.4}.
	\eqen
	Markov's inequality gives
	\eqbn
	\P[D(m,d)^c] 
	\leq \frac{(\log n)^{-6.4}}{(\log n)^{-6}}
	=(\log n)^{-0.4}.
	\eqen
	The lemma follows by taking a union bound over all  $m\in[M]$ and $d\in[2K+7]$.
\end{proof}

Let $\Tinf$ be the first time at which at least $n(\log n)^{-6}$ nodes are informed, i.e.,
\eqbn
\Tinf = \inf\{t\geq 0\,:\,N_6(t)\geq n(\log n)^{-6} \}.
\eqen
Let $\Tinff$ be the first time at which at least a fraction $0.01$ of the nodes are either informed or terminal, i.e.,
\eqbn
\Tinff = \inf\{t\geq 0\,:\,N_6(t)\vee N_4(t)\geq 0.01n \}.
\eqen
The following lemma bounds from above the number of terminal nodes which are created before time $\Tinf$, and also bounds the number of communications initiated by informed nodes before this time.
\begin{lemma} 
	Let $D_4$ denote the event that for all $t\leq\Tinf\wedge t_M$ we have $N_4(t)<n(\log n)^{-6}$. Then $D_4$ occurs w.h.p. Furthermore, the number of communications initiated by informed nodes before time $\Tinf\wedge (\log n)$ is smaller than $2n$ w.h.p.
	\label{prop:T1}
\end{lemma}
\begin{proof}
	Let $a$ be a constant independent of $n$ such that $a(\log\log n)^2>2t_M$. Let $\Tterm=\inf\{t\geq 0\,:\,N_4(t)\geq n(\log n)^{-6} \}$, and define $T=\Tinf\wedge t_M\wedge\Tterm$. To prove the first assertion of the lemma it is sufficient to show that $\Tterm>T$ w.h.p. Let $\frk k$ be the number of communications initiated by regular nodes at or before time $T$, i.e., if $T_k$ is the $k$th time a regular node initiates a communications we have $\frk k=\sup\{ k\,:\,T_k\leq T \}$. Note that each terminal node was either informed or regular immediately before becoming a terminal node. We have $N_4(t)=N_4^{\op{r}}(t)+N_4^{\op{i}}(t)$, where $N_4^{\op{i}}(t)$ (resp.\ $N_4^{\op{r}}(t)$) is the number of nodes $i$ which are terminal nodes at time $t$ and which were informed (resp.\ regular) immediately before they became terminal, i.e., if $t_0=\inf\{ t'\geq 0\,:\,\sigma_1(i,t')=4 \}$ then $\sigma_1(i,t_0^-)=6$ (resp.\ $\sigma_1(i,t_0^-)=3$). 
	
	We will now define a martingale $(M_k)_{k\in\N\cup\{0 \}}$ inductively. First set $M_0=0$. If $k>\frk k$ set $M_{k}=M_{k-1}$. Otherwise define $M_{k}=M_{k-1}+A_k-\E[A_k\,|\,\cF_{k-1}]$, where the random variable $A_k$ and the $\sigma$-algebra $\cF_k$ are defined as follows. Set $A_k=1$ if a new terminal node is created at time $T_k$, and set $A_k=0$ otherwise. Observe that for $k<\frk k$, $N_4^{\op r}(T_k)=\sum_{k'=1}^{k}A_k$. For $k\leq\frk k$ let $\cF_{k}$ be the $\sigma$-algebra which contains all information about the protocol for times $t<T_{k+1}$. Then $(M_k)_{k\in\N\cup\{0 \}}$ is a martingale for the filtration $\cF_k$ (if $\cF_k=\cF_{\frk k}$ for $k>\frk k$) with increments bounded by 1. 
	
	Azuma's inequality gives that w.h.p., $M_{\frk k\wedge an} \leq (an)^{0.55}$. By Hoeffding's inequality, w.h.p.\ the total number of clock rings before time $t_M$ is bounded by $1.1t_Mn$, i.e.,
	\eqbn
	\#\{ t\in[0,t_M]\,:\,\exists i\in[n]\text{\,\,such\,\,that\,\,}t\in\cP_i  \} \leq 1.1t_Mn.
	\eqen
	On this event, by the definition of $a$ and since regular nodes communicate every $\lceil \log\log n \rceil^2$ clock rings, we have $\frk k\leq 1.1t_Mn/\lceil \log\log n \rceil^2<an$. 
	
	Since informed nodes communicate at unit rate, the total number of times that some informed node initiates a communication before time $T\leq t_M$ is bounded by $2t_M\cdot \sup_{t\in[0,T]} N_6(t)$ w.h.p. Furthermore, an informed node becomes a terminal node with probability $(n-1)^{-1}( N_4(t^-)+N_6(t^-)-1 )<n^{-1}( N_4(t^-)+N_6(t^-))$ if it initiates a communication at some time $t$. Using these two observations we get that w.h.p.,
	\eqbn
	\begin{split}
	N_4^{\op{i}}( T )
	&\leq 
	2\cdot2t_M\cdot \sup_{t\in[0,T]} N_6( t ) \cdot 
	\sup_{t\in[0,T]} n^{-1}( N_4(t^-)+N_6(t^-))\\
	&\leq 4t_M\cdot n(\log n)^{-6}\cdot 2(\log n)^{-6}
	= 8t_Mn(\log n)^{-12}.
	\end{split}
	\eqen
	
	We condition on the high probability events discussed in the last two paragraphs in the remainder of the proof. Then
	\eqbn
	\begin{split}
	N_4^{\op{r}}(T)
	&=
	N_4^{\op{r}}(T_{\frk k})
	=
	\sum_{k=1}^{\frk k} A_{k}
	=
	M_{\frk k} + \sum_{k=1}^{\frk k}\E[A_k\,|\,\cF_{k-1}] 
	= (an)^{0.55} + \sum_{k=1}^{\frk k}\frac{N_4( (T_k)^- )}{n-1}\\
	&\leq (an)^{0.55} + \sum_{k=1}^{\frk k}\frac{N^{\op{r}}_4( (T_k)^- ) + 8t_Mn(\log n)^{-12}}{n-1}.
	\end{split}
	\eqen
	Defining $f(k):=N_4^{\op{r}}(T_{k})=N_4^{\op{r}}((T_{k+1})^- )$ for $k=1,\dots,\frk k$, $\eps_1=(an)^{0.55}$ and $\eps_2=8t_M(\log n)^{-12}n/(n-1)$, we have
	\eqb
	f(k) \leq \eps_1+\eps_2k+(n-1)^{-1}\sum_{k'=1}^{k-1} f(k').
	\label{eq135}
	\eqe
	By induction on $k$ we get from \eqref{eq135} that
	\eqbn
	f(k) \leq  \eps_2n(e^{2k/(n-1)}-1) + \eps_1 e^{2k/(n-1)}.
	\eqen
	Inserting $k=\frk k<an$ and using that $f$ is monotone, we get the following for sufficiently large $n$
	\eqbn
	N_4^{\op{r}}(T_{\frk k}) = f(\frk k) < f(an) \leq \eps_2 n ( e^{3a}-1 ) + \eps_1 e^{3a},
	\eqen
	so 
	\eqbn
	N_4(T)=N^{\op{i}}_4(T)+N^{\op{r}}_4(T_{\frk k})\leq 
	8t_Mn(\log n)^{-12}
	+\eps_2 n ( e^{3a}-1 ) + \eps_1 e^{3a}
	<n(\log n)^{-11}, 
	\eqen
	which implies $T<\Tterm$ as desired.
	
	To complete the proof of the lemma we need to bound the number of communications initiated by informed nodes before time $\Tinf\wedge (\log n)$. Since informed nodes communicate at unit rate, the total number of times that some informed node initiates a communication before time $\Tinf\wedge (\log n)$ is bounded by $2(\Tinf\wedge (\log n))\cdot \sup_{t\in[0,\Tinf]} N_6(t)\leq 2n$ w.h.p.
\end{proof}

Define
\eqbn
\begin{split}
	E_m &= \{ L_-(m)>\alpha(1-\beta_m) \} \cup \{  \Tinf<t_m \},\\
	E_m^\init &= \{ L^\init_-(m)>\alpha^\init(1-\beta^\init_m) \} \cup \{ \Tinf<t_m^\init  \}.
\end{split}
\eqen
The following lemma lower bounds the total number of level $m$ experts, given a lower bound for the number of nodes which become level $m$ experts upon receiving a bit from three level $m-1$ experts. 
More precisely, we will prove these results for the experts in $\cE_-^1(m)$ and $\cE_-(m)$, respectively. 
\begin{lemma}[Lower bound, number of experts]
	For all sufficiently large $n$,
	\eqb
	\P[(E_m)^c ; E^\init_{m}; D_1; D_2; D_3; D_4] < \exp(-2n^{0.1}).
	\label{eq116}
	\eqe
	\label{prop59}
\end{lemma}
\begin{proof}
	For $k\in\N$ and $d\in[2K+7]$ let $T^d_k$ be the $k$th time that a level $m$ expert $i\in\cE_-(m)$ with counter $d$ initiates a communication, where $T^d_k=\infty$ if this time is not well-defined. Observe that $T^d_k>\Tasp/2$ by the definition of $\cE_-(m)$. We will now define a martingale $M^d_k$ inductively. First set $M^d_0=0$. If $T^d_k=\infty$ or $T^d_k>\Tinf$ then set $M^d_k=M^d_{k-1}$. Otherwise define $M^d_k=M^d_{k-1}+A^d_k-\E[A^d_k\,|\,\cF^d_{k-1}]$, where the random variable $A^k_d$ and the $\sigma$-algebra $\cF^d_k$ are defined as follows. Let $A^d_k=2$ if some node $i\in[n]$ (other than the node initiating the communication at time $T^d_k$) becomes a new level $m$ expert at time $T^d_k$ and set $A^d_k =1$ otherwise. For $T_k^d\leq \Tinf$ 
let $\cF^d_{k-1}$ be the $\sigma$-algebra which contains information about the protocol for times $t<T^d_{k}$. 
	For $T_k^d> \Tinf$ let $\cF^d_{k-1}$ be the $\sigma$-algebra which contains information about the protocol for times $t\leq \Tinf$.  
	Then $M^d_k$ is a martingale for the filtration $\cF^d_k$. The increments of $M_k^d$ are bounded by 1. For $d=2,\dots,2K+7$ let $L^d_-(m)$ be the number of level $m$ experts with counter $d$ which are created before time $t_m$, i.e.,
	\eqbn
	L^d_-(m) = \#\{  i\in[n]\,:\,\exists t\in[0,t_m] \text{\,\,such\,\,that\,\,}\sigma(i,t)=(2,m,d,\xi,\frk b_i,b),\, b\in\{0,1\},\xi\in[4] \}.
	\eqen
	Also define
	$$
	k_d=\sup\{k\in[ \lceil (1 - 4(\log n)^{-6}  )\cdot L^{d-1}_-(m)\rceil ]\,:\, 
	T_k^d<t_m	 \},
	$$
	with $k_d=0$ if the considered set is empty.
	By Azuma's inequality and since $k_d\leq n$, $\P[ M^d_{k_d}<-n^{0.55}]\leq \exp(-2n^{0.1})$. Let $E=E_m^\init\cap D_1\cap D_2\cap D_3\cap D_4\cap\{\Tinf>t_m \}\cap\{M^d_{k_d}\geq -n^{0.55}  \}$. Recall Remark \ref{rmk1}.
	On $E$ we have $T^d_{k_d}<t_m<\Tinf$, so for all $k\leq k_d$ and letting $t=(T^d_k)^-$ be infinitesimally smaller than $T^d_k$,
	\eqbn
	\begin{split}
		\E[A^d_k\,|\,\cF^d_k] 
		&= 1+(n-1)^{-1}N_3(t)
		\geq 1+n^{-1}(n-N_1(t)-N_2(t)-N_4(t)-N_5(t)-N_6(t))\\
		&\geq 1+n^{-1}\Big(n- 2n(\log n)^{-6}-(M+1)\alpha(1+\beta_M)-n(\log n)^{-6}\\
		&\qquad-(M+1)\alpha(1+\beta_M)-n(\log n)^{-6}\Big)\\
		&\geq 2 - 2.5n^{-1}M\alpha.
	\end{split}
	\eqen	
	It follows that on $E$,
	\eqb
	\begin{split}
		L_-^{d}(m)
		&=\sum_{k\,:\,T_k^d\leq t_m}A^d_{k}
		\geq
		\sum_{k=1}^{k_d}A^d_{k} 
		= M_{k_d}+\sum_{k=1}^{k_d}\E[A^d_k\,|\,\cF^d_k] 
		\geq -n^{0.55}+(2 - 2.5n^{-1}M\alpha )k_d.
	\end{split}
	\label{eq132}
	\eqe
	We will argue by induction on $d$ that 
	\eqb
	k_d=(1 - 4(\log n)^{-6})\cdot L^{d-1}_-(m).
	\label{eq134} 
	\eqe
	For $d=2$, \eqref{eq134} is immediate since the number of late potential level $(m,2)$ experts is smaller than $(\log n)^{-6}L^1_-(m)$. Assuming that \eqref{eq134} holds for $1,\dots,d$, we get by \eqref{eq132} that
	\eqb
		L_-^{d}(m)\geq 2(1-2n^{-1}M\alpha)\cdot L_-^{d-1}(m).
		\label{eq133}
	\eqe
	Iterating this gives 
	$L_-^{d}(m)
	\geq (1-2n^{-1}M\alpha)^{d-1}\cdot 2^{d-1}\cdot L_-^{1}(m)
	\geq 0.9\cdot 2^{d-1}\cdot L_-^{1}(m)$.
	By first using the definition of $D_3$ and then using this bound, we get that the number of late potential level $(m,d+1)$ experts is smaller than $(\log n)^{-6}2^{d}L^1_-(m)<4(\log n)^{-6}L_-^{d}(m)$, which implies by the definition of $k_{d+1}$ that we have $k_{d+1}=(1 - 4(\log n)^{-6})\cdot L^{d}_-(m)$. 

	Iterating \eqref{eq133} and using that $L_-^1(m)\geq \alpha^1(1-\beta_{m}^1)$ on $E_m^1$,
	\eqbn
	L_-^{2K+7}(m) \geq 2^{2K+6}(1- 2n^{-1}M\alpha )^{2K+6} L_-^\init(m)
	\geq 2^{2K+6} (1- 3n^{-1}M\alpha\cdot (2K+6)-\beta_{m}^\init ) \alpha^1
	\geq \alpha(1-\beta_{m}),
	\eqen
	which implies the occurrence of $E_m$.
\end{proof}

The following lemma lower bounds the number of nodes which become level $m$ experts upon receiving bits from three level $m-1$ experts, given a lower bound on the number of level $m-1$ experts. More precisely, we will prove these results for the experts in $\cE_-^1(m)$ and $\cE_-(m-1)$, respectively.
\begin{lemma}[Lower bound, number of experts]
	For all sufficiently large $n$,
	\eqb
	\P[  (E^\init_m)^c ; E_{m-1}; D_1; D_2; D_4; D_6]< \exp(-n^{0.09})
	\label{eq117}
	\eqe
	\label{prop60}
\end{lemma}
\begin{proof}	
	We say that an expert candidate with state $(5,m,d,\xi,b^1,b^2,b^3,\frk b_i,b)$ has one (resp.\ two) bits if exactly one (resp.\ two) of the variables $b^1,b^2,b^3$ is different from $-1$. For $v=1,2,3$ we say that the expert candidate \emph{has bit $b^v$} if $b^v\neq -1$.
	
	Let $T_k$ be the $k$th time that an expert $i\in\cE_-(m-1)$ with counter $d=2K+7$ initiates a communication, where $T_k=\infty$ if this time is not well-defined (so $T_k=\infty$ if and only if $k>\#\cE_-(m-1)$). For $k\in\N$ and $u=1,2,3$ define $M^u_k$ inductively as follows with $M^u_0=0$. If $T_k=\infty$ or $T_k>\Tinf$ set $M^u_k=M^u_{k-1}$. Otherwise set $M_k^u=M_{k-1}^u+A_k^u-\E[A_k^u\,|\,\cF_{k-1}]$, where the random variable $A^u_k$ and the sigma-algebra $\cF_k$ are defined as follows for such $k$. Let $\cF_{k}$ be the $\sigma$-algebra containing information about the protocol for times $t<T_{k+1}$, in addition to $\cP_i\cap [0,t_m]$ for all nodes $i$ which are level $m$ expert candidates at time $T_k$. For $u=1,2$ set $A^u_k=1$ iff a new expert candidate $i$ with $u$ bits is created at time $T_k$, and if $i\in\cA$. Recall that $i\in\cA$ means that the Poisson clock of $i$ rings less than $2t_M$ times during $[0,t_{M}]$, and note that if $i$ is the node contacted at time $T_k$ then $i\in\cA$ is measurable with respect to $\cF_{k}$. Observe that $i\in\cA$ occurs except on an event of the following probability for $Y$ a Poisson random variable of parameter $t_M$: $\P[Y\geq 2t_M]<2^{-0.55t_M}<2^{-40(\log\log n)^2}$ (see \eqref{eq:poisson}). Set $A^3_k=1$ if a new level $m$ expert is created at time $T_k$. Observe that $L_-^\init(m)\geq \sum_{k\,:\,T_k\leq t_m\wedge\Tinf} A^3_k$ (with equality if $t_m<\Tinf$). 
	
	For $u=1,2,3$ the process $M_k^u$ is a martingale with increments bounded by 1. By Azuma's inequality, for $k=1,\dots,k_0:=\lceil 1-2(\log n)^{-6}\rceil L_-(m-1)$,
	\eqb
		\P[ |M^u_{k}|>n^{0.55} ] \leq 2\exp(-n^{1.1}/(2k)).
		\label{eq126}
	\eqe
	Let $L$ be the number of experts $i\in\cE_-(m-1)$ with counter $2K+7$ which initiate a communication before time $t_m^\init$. For $i\in\cE_-(m-1)$ let $\frk s_i\in\cP_i$ denote the time such that the counter of $i$ is set equal to $2K+7$ at time $\frk s_i$. For a uniformly sampled $i\in\cE_-(m-1)$ and $Y$ a unit rate exponential random variable, the probability that $i$ does not initiate a communication during the interval $(\frk s_i,t_m^\init)$ is at most $\P[Y>K]=\exp(-K)\leq (\log n)^{-6}$. Therefore, by Hoeffding's inequality,
 	\eqbn 
	\P[ L\geq k_0\,|\, \cE^\init_-(m) ]
	\1_{E_{m-1}}
	<\exp(-2(\log n)^{-12}\cdot \#\cE^\init_-(m))\1_{E_{m-1}}.
	\eqen
	In order for $A_k^1=1$ three criteria must be satisfied: (i) The node $j$ which is contacted at time $T_k$ is regular, (ii) the expert $i$ which initiates the communication must have an expert type $\xi\in\{1,2,3 \}$, and (iii) $j\in\cA$. The two events (i) and (ii) are independent by Lemma \ref{prop50}, and they have probability $(n-1)^{-1}N_3\big((T_k^1)^-\big)>n^{-1}N_3\big((T_k^1)^-\big)$ and $3/4$, respectively. The event (iii) has probability $\P[\wh Y>2t_m]$ for $\wh Y$ a Poisson random variable with parameter $t_M$.
	
	Let $E=D_1\cap D_2\cap D_4\cap E_{m-1}\cap \{\Tinf>t_m^\init \}\cap \{ L> k_0 \}$. On $E$ we have $T_{k_0}<t_m^1<\Tinf$, so for $k=1,\dots,k_0$ and $t=(T_k^\init)^-$ infinitesimally smaller than $T_k^\init$,
	\eqbn
	\begin{split}
	\E[A_k^1&\,|\,\cF_{k-1}] > \frac{3}{4}\cdot n^{-1} ( N_3(t) )-\P[\wh Y>2t_M]\\
	&= \frac{3}{4}\cdot n^{-1}(n-N_1(t)-N_2(t)-N_4(t)-N_5(t)-N_6(t))-2^{-0.55\lambda}\\
	&\geq \frac{3}{4}\cdot n^{-1}\Big(n- 2n(\log n)^{-6}-(M+1)\alpha(1+\beta_M)-n(\log n)^{-6}\\
	&\quad-(M+1)\alpha(1+\beta_M) - n(\log n)^{-6} \Big)-2^{-0.55\lambda}\\
	&\geq \frac 34(1 - 2.5n^{-1}M\alpha).
	\end{split}
	\eqen
	and
	\eqbn
	\begin{split}
		\E[A_k^1&\,|\,\cF_{k-1}] \leq \frac{3}{4}\cdot n^{-1} \cdot N_3(t) \leq \frac 34.
	\end{split}
	\eqen
	
	Define $\wh A_k^u=\sum_{k'\leq k} A_{k'}^u$ for all $u,k$. We have $A_k^2=1$ if and only if the following two criteria are satisfied: (i) The node which is contacted at time $T_k$ has exactly one bit $b^\xi$, and (ii) the expert which initiates the communication has an expert type in $\{1,2,3 \}\setminus\{\xi \}$. These two events (i) and (ii) are independent by Lemma \ref{prop50}, and they have probability $\wt n^{-1} (\wh A_{k-1}^1-\wh A_{k-1}^2)$ and $1/2$, respectively, where $\wt n=n-1$. Therefore
	\eqbn
	\E[A_k^2\,|\,\cF_{k-1}] = \frac{1}{2}\wt n^{-1}  (\wh A_{k-1}^1-\wh A_{k-1}^2).
	\eqen
	By a similar argument,
	\eqbn
	\E[A_k^3\,|\,\cF_{k-1}] 
	= \frac{1}{4} \wt n^{-1} (\wh A_{k-1}^2-\wh A_{k-1}^3).
	\eqen
	Define $a=1-2.5n^{-1}M\alpha$ to simplify notation. On the intersection of the \emph{complement} of the event in \eqref{eq126} for $u=1,2,3$ and $E$,
	\eqbn
	\begin{split}
		\wh A_k^1&=\sum_{j\leq k} A^1_j = M_k^1 + \sum_{j\leq k} \E[A^1_j\,|\,\cF_{j-1}]\geq -n^{0.55}+\frac 34 ak,\\
		\wh A_k^1&=\sum_{j\leq k} A^1_j = M_k^1 + \sum_{j\leq k} \E[A^1_j\,|\,\cF_{j-1}]\leq n^{0.55}+\frac 34 k,\\
		\wh A_k^2&=\sum_{j\leq k} A^2_j
		= M_k^2 + \sum_{j\leq k} \E[A^2_j\,|\,\cF_{j-1}]
		\leq n^{0.55}+\sum_{j\leq k} \frac 12 \wt n^{-1} \wh A^1_{j-1}\\
		&\leq n^{0.55}+\sum_{j\leq k} \frac 12 \wt n^{-1} \Big(\frac 34 (j-1) + n^{0.55}\Big) \leq \frac {3}{16} \wt n^{-1}k^2 + \Big(\frac{1}{2}\wt n^{-1}k +1\Big)n^{0.55}.
	\end{split}
	\eqen
	Using the lower bound for $\wh A_k^1$ and the upper bound for $\wh A_k^2$, 
	\eqbn
	\begin{split}
		\wh A_k^2&=\sum_{j\leq k} A^2_j
		= M_k^2 + \sum_{j\leq k} \E[A^2_j\,|\,\cF_{j-1}]
		\geq -n^{0.55}+\sum_{j\leq k} \frac 12 \wt n^{-1} (\wh A^1_{j-1}-\wh A^2_{j-1})\\
		&\geq -n^{0.55}+\sum_{j\leq k} \frac 12 \wt n^{-1} \Big(\frac 34 a(j-1) - n^{0.55} -\frac {3}{16} \wt n^{-1}(j-1)^2 - \Big(\frac{1}{2}\wt n^{-1}(j-1) +1\Big)n^{0.55} \Big)\\
		&\geq \frac {3}{16} a\wt n^{-1}(k-1)^2 
		- \frac{1}{2^5}\wt n^{-2}k^3
		-\Big(\frac{1}{2} \wt n^{-1}k +\frac {1}{8} \wt n^{-2}k^2 + \frac{1}{2}\wt n^{-1}k+1 \Big)n^{0.55},\\
		\wh A_k^3&=\sum_{j\leq k} A^3_j 
		= M_k^3 + \sum_{j\leq k} \E[ A^3_j\,|\,\cF_{j-1}]
		\leq M_k^3 + \sum_{j\leq k}
		\frac{1}{4}\wt n^{-1} \wh A_{j-1}^2\\
		& \leq n^{0.55}+\sum_{j\leq k} \frac 14 \wt n^{-1} 
		\Big( \frac {3}{16} \wt n^{-1}j^2 + \Big(\frac{1}{2}\wt n^{-1}j +1\Big)n^{0.55} \Big)\\
		& \leq \frac{1}{2^6}\wt n^{-2}k^3 +
		\Big( \frac {1}{16} \wt n^{-2}k^2 + \frac{1}{4}\wt n^{-1}k +1\Big)n^{0.55}.
		\end{split}
	\eqen
	Using the lower bound for $\wh A_k^2$ and the upper bound for $\wh A_k^3$, 
	\eqbn
	\begin{split}
		\wh A_k^3&=\sum_{j\leq k} A^3_j 
		= M_k^3 + \sum_{j\leq k} \E[ A^3_j\,|\,\cF_{j-1}]
		\geq M_k^3 + \sum_{j\leq k}
		\frac{1}{4}\wt n^{-1} (\wh A_{j-1}^2-\wh A_{j-1}^3)\\
		& \geq -n^{0.55}+\sum_{j\leq k} \frac 14 \wt n^{-1} 
		\Big( \frac {3}{16} a\wt n^{-1}(j-2)^2 - \Big(\frac{1}{2}\wt n^{-1}(j-1) +1\Big)n^{0.55} \Big)\\
		&\qquad - \sum_{j\leq k} \frac 14 \wt n^{-1} 
		\Big(\frac{1}{2^6}\wt n^{-2}j^3 +
		\Big( \frac {1}{16} \wt n^{-2}j^2 + \frac{1}{4}\wt n^{-1}j +1\Big)n^{0.55}\Big)\\
		& \geq \frac{1}{2^6}a\wt n^{-2}(k-2)^3 
		-\frac{1}{2^{10}} \wt n^{-3}k^4
		- \frac 14 \wt n^{-1} 
		\Big( \frac {1}{16\cdot 3} \wt n^{-2}k^3 + \frac{1}{8}\wt n^{-1}k^2 +k\Big) n^{0.55}\\
		&\qquad-\Big( \frac {1}{16} \wt n^{-2}(k-1)^2 + \frac{1}{4}\wt n^{-1}k +1\Big)n^{0.55}.
	\end{split}
	\eqen
	Inserting $k=k_0$, on the event $E$,
	\eqbn
	L_-(m)\geq
	\wh A_{k_0}^3 \geq \alpha^\init(1-\beta_{m}^\init),
	\eqen
	so $E_m^\init$ occurs.
\end{proof}

\begin{coro}
	W.h.p., at least one of the following two holds: (i) $\Tinf<t_M$ or (ii) $|L_-(M)-\alpha|<\alpha\beta_M$. 
	\label{prop49}
\end{coro}
\begin{proof}
	Combining the lemmas above, the event $D_1\cap D_2\cap D_4\cap D_6$ occurs w.h.p. It follows from Lemmas \ref{prop48}, \ref{prop59}, and \ref{prop60} that the events $E_m^1$ and $E_m$ occur w.h.p.\ for all $m$. We now obtain the corollary, since the event in the statement of the corollary (i.e., the union of (i) and (ii)) is equivalent to occurrence of $E_M$. 
\end{proof}

The following lemma implies that at time $\Tinf$ there are less than $0.1n$ nodes which are still aspirants. We will use this to argue that informed nodes spread their bit efficiently after time $\Tinf$, since w.h.p., every time they initiate a communication they reach a regular node with constant order probability.
\begin{lemma}
	Let $D_8$ be the event 
	\eqbn
	D_8 = \{\Tinf>200\eps^{-1} \}\cap\{N_1(200\eps^{-1})<0.1n \}.
	\eqen 
	Then $D_8$ happens w.h.p.
\end{lemma}
\begin{proof}
	We will prove separately that the two conditions of $D_8$ are satisfied w.h.p. First we consider the condition $\Tinf>200\eps^{-1}$. Let $\cE\subset[n]$ denote the set of nodes which become level 0 experts before time $200\ep^{-1}$, i.e.,
	\eqbn
	\cE = \{ i\in[n]\,:\,i\in\cE(0) \text{\,\,and\,\,}\exists t\in[0,200\ep^{-1}] \text{\,\,such\,\,that\,\,}\sigma_1(i,t)=2 \}.
	\eqen
	Let $\cI\subset[n]$ denote the set of nodes which become informed before time $200\ep^{-1}$, i.e.,
	\eqbn
	\cI = \{ i\in[n]\,:\exists t\in[0,200\ep^{-1}] \text{\,\,such\,\,that\,\,}\sigma_1(i,t)=6 \}.
	\eqen
	Recall the notion of \emph{influence} from Definition \ref{def:influence}. Let $T(i) \subset[n]$ denote the set of nodes influenced by $i$ during $[0,200\eps^{-1}]$. By the definition of the protocol, any node which becomes informed before time $200\eps^{-1}$ is influenced by a level 0 expert during $[0,200\eps^{-1}]$. Therefore
	\eqb
	\cI\subset \bigcup_{i\in\cE} T(i).
	\label{eq123}
	\eqe
	
	If $i\in[n]$ is a level 0 expert, then the clock of $i$ must ring at least $\Tasp$ times before $i$ becomes an expert since this is the duration of phase $\chi=3$ of the aspirant phase. Therefore, for any fixed $i\in[n]$ and with $Y$ a Poisson random variable with parameter $\Tasp/2$, \eqref{eq:poisson} gives
	\eqbn
	\P[i\in\cE] \leq \P[ \#(\cP_i\cap[0,200\ep^{-1}])\geq\Tasp ]
	\leq \P[Y\geq\Tasp] \leq 2^{-0.55\Tasp/2} < (\log n)^{-5000}.
	\eqen
	Markov's inequality gives 
	\eqb
	\P[\#\cE\geq n(\log n)^{-4000}] \leq (\log n)^{-1000}. 
	\eqe
	
	The random variable $\#T(i)$ is stochastically dominated by a Yule-Furry process with rate 2 at time $200\ep^{-1}$, see the proof of Lemma \ref{prop37}. It follows from the explicit formula for the distribution of a Yule-Furry process \cite[page 122]{karlin66-2} that $C_{\op{YF}}:=\E[\#T(i)]<\infty$ for a constant $C_{\op{YF}}$ depending only on $t$. Note that the random variables $T(i)$ and $\1_{i\in\cE}$ are independent. Therefore, by Markov's inequality and \eqref{eq123},
	\eqbn
	\P[ \Tinf>200\eps^{-1}\,|\,\cE ]
	\leq \P[\#\cI\geq n(\log n)^{-6}\,|\,\cE]
	\leq \frac{\E[\#\cI\,|\,\cE]}{n(\log n)^{-6}}
	\leq \frac{\#\cE \cdot C_{\op{YF}}}{n(\log n)^{-6}}.
	\eqen
 	By taking a union bound and applying this estimate and \eqref{eq112},
 	\eqbn
 	\begin{split}
 		\P[ \Tinf>200\eps^{-1}]
 		&\leq \P[ \Tinf>200\eps^{-1};\#\cE< n(\log n)^{-4000}]+\P[ \#\cE\geq n(\log n)^{-4000}]\\
 		&\leq \frac{n(\log n)^{-4000} \cdot C_{\op{YF}}}{n(\log n)^{-6}} 
 		+(\log n)^{-1000},
 	\end{split}
 	\eqen
 	which converges to 0 as $n\rta\infty$. We conclude that the first of the two events defining $D_8$ occurs w.h.p.
	
	Now we consider the second of the two events defining $D_8$, and we will prove that $N_1(200\eps^{-1})<0.1n$ w.h.p. Fix $i\in[n]$, and let $\tau_i^1$, $\tau_i^2$, $Y_1$, and $Y_2$ be as in the proof of Lemmas \ref{prop48} and \ref{prop51}. By a union bound, \eqref{eq:poisson}, and \eqref{eq:geometric}, 
	\eqbn
	\begin{split}
		\P[&\sigma_1(i,200\eps^{-1})=1]\\
		&\leq \P[\#(\cP_i\cap [0,200\ep^{-1}])\leq 100\ep^{-1}]
		+ \P[ \#(\cP_i\cap [0,\tau_i^1])\geq 50\ep^{-1} ]
		+ \P[ \#(\cP_i\cap (\tau_i^1,\tau_i^2])\geq 50\ep^{-1} ]\\
		&\leq 2^{-0.18\cdot 200\ep^{-1}}+\P[Y_1\geq 50\ep^{-1}/4]+\P[Y_2\geq 50\ep^{-1}/2]\\
		&\leq 2^{-0.18\cdot 200\cdot 4} + 2^{-1.44\cdot 50\ep^{-1}/4\cdot \ep/2} + 2^{-1.44\cdot 50\ep^{-1}/2\cdot \ep/2}<0.05.
	\end{split}
	\eqen
	Since the events $\{\sigma_1(i,200\eps^{-1})=1 \}$ are independent for different $i$, Hoeffding's inequality gives that w.h.p.,
	\eqbn
	N_1(200\ep^{-1})=\sum_{i\in[n]} \1_{ \sigma_1(i,200\eps^{-1})=1 }<0.1n.
	\eqen
\end{proof} 

An informed node $i$ spreads its estimate for the majority bit $\frk b$ by contacting a uniformly chosen node $j$ every time its clock rings. If the node $j$ is a regular node infinitesimally before the communication it will also become informed with the same majority bit estimate as $i$. The next lemma shows that this spreading is rather fast after time $\Tinf$. More precisely, the lemma shows that it typically takes at most time $8\log\log n$ from $n(\log n)^{-6}$ nodes are informed to a constant fraction of the nodes are informed. It also shows that the number of communications initiated by informed nodes during this time interval is at most $n$ w.h.p.
\begin{lemma} 
	Let $D_6$ be the event that 
	\begin{itemize}
		\item[(i)] during the interval $[\Tinf,\Tinff]$ the number of communications initiated by informed nodes is smaller than $n$, and 
		\item[(ii)] $\Tinff-\Tinf<12\log\log n$.
	\end{itemize}
	Then $D_6$ occurs w.h.p. 
	\label{prop:T1T2}
\end{lemma}
\begin{proof}
	Let $T_k$ be the $k$th time after $\Tinf$ at which either an informed node or a level $M$ expert with counter $2K+7$ initiates a communication (for $k$ sufficiently large such that this time is not well-defined, set $T_k=\infty$). Let $\frk k=\sup\{k\,:\,T_k\leq\Tinff \}$. Define a martingale $(M_k)_{k\in\N\cup\{0 \}}$ by $M_0=0$ and $M_k=M_{k-1}+A_{k}-\E[A_k\,|\,\cF_{k-1}]$ for $k\in\N$, where the random variable $A_k$ and the $\sigma$-algebra $\cF_{k-1}$ are defined as follows. For $k\leq \frk k$ let $A_k=N_6(T_k)-N_6( (T_k)^- )$, and for $k>\frk k$ set $A_k=0$. For $k\leq\frk k$ let $\cF_{k-1}$ be the $\sigma$-algebra containing all information until infinitesimally before $T_k$, and for $k>\frk k$ let  $\cF_{k-1}$ be the $\sigma$-algebra containing all information until  $\Tinff$.
	The martingale $M_k$ has increments bounded by 1. By Azuma's inequality,
	\eqb
	\P[M_{\frk k\wedge n} \leq -n^{0.55}] \leq \exp(-2n^{0.1}).
	\label{eq127}
	\eqe
	On the \emph{complement} of the event in \eqref{eq127},
	\eqb
	\begin{split}
		N_6(T_{\frk k\wedge n})
		&=N_6(\Tinf)+\sum_{k'=1}^{\frk k\wedge n} A_{k'}
		=N_6(\Tinf)+M_{\frk k\wedge n}+\sum_{k'=1}^{\frk k\wedge n} \E[A_{k'}\,|\,\cF_{k'-1}]\\
		&>\frac{n}{(\log n)^6}-n^{0.55}+\sum_{k'=1}^{\frk k\wedge n} \E[A_{k'}\,|\,\cF_{k'-1}].
	\end{split}
	\label{eq131}
	\eqe
	Let $E=D_2\cap D_4\cap D_8\cap \{ M_{\frk k\wedge n} > -n^{0.55} \}$. Then $E$ occurs with high probability. We assume in the remainder of the proof that $E$ occurs. 
	
	Now we will argue that for $k<\frk k$ we have
	\eqb
	\E[A_{k}\,|\,\cF_{k-1}] \geq 0.6.
	\label{eq128}
	\eqe
	When the clock of an informed node or a level $M$ expert $i$ with counter $2K+7$ rings, a new informed node is created if and only if $i$ contacts a regular node $j$. By occurrence of the events $D_8$ and $D_2$, and by definition of $\Tinff$, for $t=(T_i)^-$ infinitesimally smaller than $T_i$, this event has probability 
	\eqb
	\begin{split}
		(n-1)^{-1}N_3(t) &\geq
		n^{-1}N_3(t) \geq n^{-1}(n-N_1(t)-N_2(t)-N_4(t)-N_5(t)-N_6(t)) \\
		&\geq n^{-1}(n-0.1n-(M+1)\alpha(1+\beta_M)-0.01n-(M+1)\alpha(1+\beta_M)-0.01n)
		\geq 0.8
	\end{split}
	\label{eq129}
	\eqe
	for $n$ large enough.
	When the clock of an informed node or a level $M$ expert $i$ with counter $2K+7$ rings, the number of informed nodes decreases by 1 if and only if $i$ is an informed node which contacts either an informed node or a terminal node. This event has probability at most
	\eqb
	(n-1)^{-1}(N_4(t)+N_6(t)-1)<
	n^{-1}(N_4(t)+N_6(t)) \leq 0.2.
	\label{eq130}
	\eqe 
	Combining \eqref{eq129} and \eqref{eq130} gives \eqref{eq128}.
	
	Now consider two different cases: (i) $\frk k\geq n$ and (ii) $\frk k< n$. In case (i) we get a contradiction to \eqref{eq131} since the left side is smaller than $0.1n$ and the right side is bigger than $n\cdot 0.6-o(n)$.	
	In case (ii) we get that the number of communications initiated by informed nodes during $[\Tinf,\Tinff]$ is smaller than $\frk k<n$. This proves that the first requirement in the definition of $D_6$ is fulfilled w.h.p.
	
	To prove that the second requirement of $D_6$ is fulfilled w.h.p., notice that the law of $T_{k+1}-T_k$ given all information before or at time $T_k$ is stochastically dominated by an exponential random variable with rate $N_6(T_k)$, which has expectation $1/N_6(T_k)$. 
	By \eqref{eq127}, \eqref{eq131}, and \eqref{eq128}, w.h.p.\ we have $N_6(T_k)>0.5n(\log n)^{-6}+0.6k$ for all $k\leq\frk k$.
	Using this and that $\frk k<n$ w.h.p., follows that w.h.p. $T_{\frk k\wedge n}-\Tinf$ is stochastically dominated by the sum of $n$ independent geometric random variables $X_k$ such that $X_k$ has parameter $0.5n(\log n)^{-6}+0.6k$. Defining $X=\sum_{k=1}^{\frk k\wedge n} X_k$ we have for sufficiently large $n$,
	\eqbn
	\begin{split}
		\E\left[ X \right] &= \sum_{k=1}^n\frac{1}{0.5n(\log n)^{-6}+0.6k} \leq 11\log\log n,\\
		\op{Var}\left[ X \right]
		&=\sum_{k=1}^n \frac{1}{(0.5n(\log n)^{-6}+0.6k)^2} 
		\leq \frac{4(\log n)^6}{n}.
	\end{split}
	\eqen
	We conclude that $\Tinff-\Tinf<12\log\log n$ w.h.p.\ since
	\eqbn
	\P[X>12\log\log n]
	\leq\P\big[ | X-\E[X] |^2>(\log\log n)^2 \big]
	\leq \frac{4n^{-1}(\log n)^6}{(\log\log n)^2}\rta 0\qquad\text{as\,\,$n\rta\infty$}.
	\eqen
\end{proof}

\begin{lemma}
	W.h.p.\ $\Tinff<t_M+12\log\log n+1$.
	\label{prop:T2small}
\end{lemma}
\begin{proof}
	By Lemma \ref{prop:T1T2} it is sufficient to argue that $\Tinf\leq t_M+1$ w.h.p. By Corollary \ref{prop49}, w.h.p.\ at least one of the properties (i) and (ii) in Corollary \ref{prop49} is satisfied. In case (i) we are done. Assume the event in case (ii) occurs, and that the event in case (i) does not occur. For each $i\in\cE_-(M)$ let $T_i\geq 0$ be such that $i$ initiates a communication at time $T_i$ and is a level $M$ expert with counter $2K+7$ infinitesimally before time  $T_i$. By the definition of $\cE_-(M)$, and since the time interval between two clock rings has the law of a unit rate exponential random variable, it holds with probability at least $1-e^{-1}$ that $T_i<t_M+1$, independently for each $i\in\cE_-(M)$. Let $j\in[n]$ denote the node which is contacted by $i$ at time $T_i$. Assume the events $D_2$, $D_4$, and $D_8$ occur. Then the probability that $j$ is a regular node is the following for $t=(T_i)^-$ infinitesimally smaller than $T_i$
	\eqbn
	\begin{split}
		(n-1)^{-1}N_3(t) &= (n-1)^{-1}( n - N_1(t) - N_2(t) - N_4(t) - N_5(t) - N_6(t) )\\
		&\geq (n-1)^{-1}
		\Big(n- 0.1n-(M+1)\alpha(1+\beta_M)-n(\log n)^{-6}\\
		&\qquad-(M+1)\alpha(1+\beta_M)-n(\log n)^{-6}\Big) \geq 0.85.
	\end{split}
	\eqen
	On the event that $j$ is a regular node, $j$ will become an informed node. By a Chernoff bound, w.h.p.\ at least $0.8(1-e^{-1})n$ informed nodes will be created before time $t_M+1$, i.e., w.h.p.,
	\eqb
	\#\{ i\in[n]\,:\,\exists t\leq t_M+1\text{\,\,such\,\,that\,\,}\sigma_1(i,t)=6 \}\geq 0.8(1-e^{-1})n.
	\label{eq124}
	\eqe
	Assume in the remainder of the proof that the event in \eqref{eq124} occurs. Recall that if an informed node changes type it will become a terminal node, and that a terminal node never changes type. Therefore the quantity on the left side of \eqref{eq124} is bounded above by $N_4(t_M+1)+N_6(t_M+1)$. If $t_M+1<\Tinf$ then we have $N_6(t_M+1)<n(\log n)^{-6}$ by occurrence of $D_4$, so $N_4(t_M+1)>0.8(1-e^{-1})n-n(\log n)^{-6}$ by \eqref{eq124}, which contradicts $t_M+1<\Tinf$. We conclude that we must have $\Tinf\leq t_M+1$.
\end{proof}

\begin{lemma} 
	W.h.p.\ all level $M$ experts have belief bit equal to $\frk b$.
	\label{prop55}
\end{lemma}
\begin{proof}
	The analogue of Lemma \ref{prop:error} still holds in the setting of the asynchronous model, by a similar proof as before. The proof of Lemma \ref{prop52} also carries through, which implies the current lemma.
\end{proof}

\begin{lemma} 
	The protocol reaches terminal consensus in finite time w.h.p. In other words, $\tau_{\op{terminal}}<\infty$ w.h.p.
	\label{prop:correct3}
\end{lemma}
\begin{proof}
	By Corollary \ref{prop49}, w.h.p.\ at least one informed node will be created. On this event the protocol terminates a.s., in the sense that all nodes eventually become terminal nodes. The belief bit of all the terminal nodes originate from a level $M$ expert. Therefore, by Lemma \ref{prop55}, w.h.p.\ all terminal nodes will have belief bit $\frk b$. Combining the above we get $\tau_{\op{terminal}}<\infty$ w.h.p.
\end{proof}

\begin{lemma} 
	There is a constant $C>0$ depending only on $\eps$ such that w.h.p.\ the number of communications until the protocol reaches terminal consensus is at most $Cn$, i.e., $N_{\op{terminal}}<Cn$ w.h.p.
\end{lemma}
\begin{proof}
	We consider separately the contribution to the number of communications coming from nodes of the following types: aspirant, expert, regular node, and informed node. This is sufficient to complete the proof since expert candidates and terminal nodes do not initiate communications.
	
	First we consider aspirants in phase $\chi=1$. An aspirant for which $\chi=1$ repeatedly collects quadruples of bits $(b',b'',b''',b'''')$. As observed in the proof of Lemma \ref{prop48}, the number of quadruples is stochastically dominated by a geometric random variable with success probability $\ep/2$, independently for each node. Therefore the expected number of communications initiated by aspirants in phase $\chi=1$ is bounded above by $4n/(\ep/2)=8n\ep^{-1}$. Furthermore, by concentration of the sum of independent geometric random variables, we get that this number is smaller than $10n\ep^{-1}$ w.h.p. By a similar argument we get that the number of communications initiated by experts in phase $\chi=2$ is bounded above by $5n\ep^{-1}$ w.h.p.
		
	By Lemma \ref{prop41} (see also Remark \ref{rmk1}), w.h.p.\ the number of nodes which are experts at some point in time is bounded above by $(M+1)\alpha(1+\beta_M)$. Each expert initiates at most $2K+7$ communications. Therefore the total number of communications ever initiated by an expert is at most $(M+1)\alpha(1+\beta_M)(2K+7)\ll n$.
	
	It remains to bound the number of communications initiated by regular nodes and informed nodes. We consider communications before and after time $\Tinff+1$ separately. By Lemma \ref{prop:T2small} and since regular nodes communicate every $\lceil(\log\log n)^2\rceil$ clock ring, we see that the total number of communications initiated by regular nodes before time $\Tinff$ is smaller than $2n(t_M+12\log\log n+1)/\lceil(\log\log n)^2\rceil=\Theta(n)$ w.h.p. By Lemmas \ref{prop:T1} and \ref{prop:T1T2} we get that the number of communications initiated by informed nodes before time $\Tinff$ is smaller than $2n$ w.h.p. Furthermore, the total number of communications in the protocol between $\Tinff$ and $\Tinff+1$ is smaller than $2n$ w.h.p.
	
	To bound the number of communications after $\Tinff+1$, we will first argue that w.h.p.\ a positive fraction of the nodes are terminal nodes at time $\Tinff+1$. More precisely, we show
	\eqb
	\lim_{n\rta\infty} \P[N_4(\Tinff+1)>0.001n]=1.
	\label{eq125}
	\eqe 
	By the definition of $\Tinff$, at least one of the following holds 
	(i) $N_6(\Tinff)>0.1n$ or (ii) $N_4(\Tinff)>0.1n$. In case (ii) and since terminal nodes never change type, the event in \eqref{eq125} occurs. In case (i) condition on the state of all nodes at time $\Tinff$. Let $T_k$ be the $k$th time after time $\Tinff$ at which the clock of an informed node is ringing. The only way for an informed node to change type is if it contacts a node which is either informed or terminal. Therefore we have $N_4(T_k)\geq 0.1n-k$. Since each node initiates a communication during $[\Tinff,\Tinff+1]$ with probability at least $(1-e^{-1})$, w.h.p.\ we have $T_{\frk k}<\Tinff+1$ for $\frk k:=\lceil0.9\cdot(1-e^{-1})\cdot 0.1n\rceil$. It follows that for $k=1,\dots,\frk k$, conditioned on the state of all nodes infinitesimally before $T_k$, with probability at least $(0.1n-\frk k)/n>0.05$ the node which initiates a communication at time $T_k$ contacts an informed node. Every time this event happens a terminal node is created. By concentration of the sum of Bernoulli random variables, w.h.p.\ there are at least $0.9\cdot0.05\cdot \frk k>0.001n$ terminal nodes at time $T_{\frk k}$. Since $T_{\frk k}<\Tinff+1$ w.h.p., this implies \eqref{eq125}.
	
	On the event in \eqref{eq125}, every time a regular node or an informed node initiates a communication after time $\Tinff+1$ it reaches a terminal node with probability at least $0.001$. When a regular node reaches a terminal node it also becomes terminal. Therefore, on the event in \eqref{eq125}, the number of communications initiated a regular node or an informed node after time $\Tinff+1$ is stochastically dominated by an exponential random variable with parameter $0.001$ (which has expectation $1000$), independently for each node. By concentration of the sum of independent exponential random variables, we get that w.h.p.\ the total number of communications from regular nodes or informed nodes after time $\Tinff+1$ is at most $1001n$.
\end{proof}

\subsection*{Acknowledgments}
We thank Rati Gelashvili for providing useful references on related work. 
Part of this work was completed when G.F.\ and N.H.\ visited Microsoft Research in Redmond, and G.R. was at Microsoft Research full time, and they want to thank Microsoft for the hospitality. Part of this work was done when G.R.\ was visiting the Simons Institute for the Theory of Computing. N.H.\ acknowledges support from a Microsoft Research internship, Dr.\ Max R\"ossler, the Walter Haefner Foundation, and the ETH Z\"urich Foundation.
G.F. acknowledges support from the Distributed Technologies Research Foundation, NSF grant CIF-1705007, ARO grant W911NF1810332, and Input Output Hong Kong.

\bibliographystyle{hmralphaabbrv}
\bibliography{references}

\newcommand{\etalchar}[1]{$^{#1}$}
\begin{thebibliography}{BKKO18}

\bibitem[AAD{\etalchar{+}}06]{angluin2006computation}
D.~Angluin, J.~Aspnes, Z.~Diamadi, M.~J. Fischer, and R.~Peralta.
\newblock Computation in networks of passively mobile finite-state sensors.
\newblock {\em Distributed computing}, 18(4):235--253, 2006.

\bibitem[AAE08]{angluin2008simple}
D.~Angluin, J.~Aspnes, and D.~Eisenstat.
\newblock A simple population protocol for fast robust approximate majority.
\newblock {\em Distributed Computing}, 21(2):87--102, 2008.

\bibitem[AAE{\etalchar{+}}17]{alistarh2017time}
D.~Alistarh, J.~Aspnes, D.~Eisenstat, R.~Gelashvili, and R.~L. Rivest.
\newblock Time-space tradeoffs in population protocols.
\newblock In {\em Proceedings of the twenty-eighth annual ACM-SIAM symposium on
  discrete algorithms}, pages 2560--2579. SIAM, 2017.

\bibitem[AAG18]{alistarh2018space}
D.~Alistarh, J.~Aspnes, and R.~Gelashvili.
\newblock Space-optimal majority in population protocols.
\newblock In {\em Proceedings of the Twenty-Ninth Annual ACM-SIAM Symposium on
  Discrete Algorithms}, pages 2221--2239. SIAM, 2018.

\bibitem[AD15]{abdullah2015global}
M.~A. Abdullah and M.~Draief.
\newblock Global majority consensus by local majority polling on graphs of a
  given degree sequence.
\newblock {\em Discrete Applied Mathematics}, 180:1--10, 2015.

\bibitem[AF02]{aldous2002reversible}
D.~Aldous and J.~Fill.
\newblock Reversible {M}arkov chains and random walks on graphs, 2002.

\bibitem[AG15]{AG15}
D.~Alistarh and R.~Gelashvili.
\newblock Polylogarithmic-time leader election in population protocols.
\newblock In {\em International Colloquium on Automata, Languages, and
  Programming}, pages 479--491. Springer, 2015.

\bibitem[AGV15]{alistarh2015fast}
D.~Alistarh, R.~Gelashvili, and M.~Vojnovi{\'c}.
\newblock Fast and exact majority in population protocols.
\newblock In {\em Proceedings of the 2015 ACM Symposium on Principles of
  Distributed Computing}, pages 47--56. ACM, 2015.

\bibitem[AS04]{alon2004probabilistic}
N.~Alon and J.~H. Spencer.
\newblock {\em The probabilistic method}.
\newblock John Wiley \& Sons, 2004.

\bibitem[BCER17]{bilke2017population}
A.~Bilke, C.~Cooper, R.~Elsaesser, and T.~Radzik.
\newblock Population protocols for leader election and exact majority with o
  (log\^{} 2 n) states and o (log\^{} 2 n) convergence time.
\newblock {\em arXiv preprint arXiv:1705.01146}, 2017.

\bibitem[BCN{\etalchar{+}}15]{becchetti2015plurality}
L.~Becchetti, A.~Clementi, E.~Natale, F.~Pasquale, and R.~Silvestri.
\newblock Plurality consensus in the gossip model.
\newblock In {\em Proceedings of the twenty-sixth annual ACM-SIAM symposium on
  Discrete algorithms}, pages 371--390. Society for Industrial and Applied
  Mathematics, 2015.

\bibitem[BFGK16]{berenbrink2016efficient}
P.~Berenbrink, T.~Friedetzky, G.~Giakkoupis, and P.~Kling.
\newblock Efficient plurality consensus, or: the benefits of cleaning up from
  time to time.
\newblock Schloss Dagstuhl, Leibniz-Zentrum f{\"u}r Informatik, 2016.

\bibitem[BKKO18]{berenbrink2018simple}
P.~Berenbrink, D.~Kaaser, P.~Kling, and L.~Otterbach.
\newblock Simple and efficient leader election.
\newblock In {\em OASIcs-OpenAccess Series in Informatics}, volume~61. Schloss
  Dagstuhl-Leibniz-Zentrum fuer Informatik, 2018.

\bibitem[BTV09]{benezit2009interval}
F.~B{\'e}n{\'e}zit, P.~Thiran, and M.~Vetterli.
\newblock Interval consensus: from quantized gossip to voting.
\newblock In {\em Acoustics, Speech and Signal Processing, 2009. ICASSP 2009.
  IEEE International Conference on}, pages 3661--3664. IEEE, 2009.

\bibitem[BTV11]{benezit2011distributed}
F.~B{\'e}n{\'e}zit, P.~Thiran, and M.~Vetterli.
\newblock The distributed multiple voting problem.
\newblock {\em IEEE Journal of Selected Topics in Signal Processing},
  5(4):791--804, 2011.

\bibitem[CCDS14]{ccds15}
H.-L. Chen, R.~Cummings, D.~Doty, and D.~Soloveichik.
\newblock Speed faults in computation by chemical reaction networks.
\newblock In {\em International Symposium on Distributed Computing}, pages
  16--30. Springer, 2014.

\bibitem[CDFR16]{cooper2016discordant}
C.~Cooper, M.~Dyer, A.~Frieze, and N.~Rivera.
\newblock Discordant voting processes on finite graphs.
\newblock {\em arXiv preprint arXiv:1604.06884}, 2016.

\bibitem[CFR09]{cooper2009multiple}
C.~Cooper, A.~Frieze, and T.~Radzik.
\newblock Multiple random walks in random regular graphs.
\newblock {\em SIAM Journal on Discrete Mathematics}, 23(4):1738--1761, 2009.

\bibitem[CG14]{cruise2014probabilistic}
J.~Cruise and A.~Ganesh.
\newblock Probabilistic consensus via polling and majority rules.
\newblock {\em Queueing Systems}, 78(2):99--120, 2014.

\bibitem[DS18]{doty2018stable}
D.~Doty and D.~Soloveichik.
\newblock Stable leader election in population protocols requires linear time.
\newblock {\em Distributed Computing}, 31(4):257--271, 2018.

\bibitem[DV12]{draief2012convergence}
M.~Draief and M.~Vojnovi{\'c}.
\newblock Convergence speed of binary interval consensus.
\newblock {\em SIAM Journal on control and Optimization}, 50(3):1087--1109,
  2012.

\bibitem[GK10]{GK10}
S.~Gilbert and D.~R. Kowalski.
\newblock Distributed agreement with optimal communication complexity.
\newblock In {\em Proceedings of the twenty-first annual ACM-SIAM symposium on
  Discrete Algorithms}, pages 965--977. SIAM, 2010.

\bibitem[GP16a]{GP16}
M.~Ghaffari and M.~Parter.
\newblock A polylogarithmic gossip algorithm for plurality consensus.
\newblock In {\em PODC}, 2016.

\bibitem[GP16b]{ghaffari2016polylogarithmic}
M.~Ghaffari and M.~Parter.
\newblock A polylogarithmic gossip algorithm for plurality consensus.
\newblock In {\em Proceedings of the 2016 ACM Symposium on Principles of
  Distributed Computing}, pages 117--126. ACM, 2016.

\bibitem[HP01]{hassin2001distributed}
Y.~Hassin and D.~Peleg.
\newblock Distributed probabilistic polling and applications to proportionate
  agreement.
\newblock {\em Information and Computation}, 171(2):248--268, 2001.

\bibitem[Kar66]{karlin66}
S.~Karlin.
\newblock {\em A first course in stochastic processes}.
\newblock Academic Press, New York-London, 1966. \MR{0208657}

\bibitem[KM11]{kanoria2011majority}
Y.~Kanoria and A.~Montanari.
\newblock Majority dynamics on trees and the dynamic cavity method.
\newblock {\em The Annals of Applied Probability}, 21(5):1694--1748, 2011.

\bibitem[KSSV00]{karp2000randomized}
R.~Karp, C.~Schindelhauer, S.~Shenker, and B.~Vocking.
\newblock Randomized rumor spreading.
\newblock In {\em Foundations of Computer Science, 2000. Proceedings. 41st
  Annual Symposium on}, pages 565--574. IEEE, 2000.

\bibitem[KT75]{karlin66-2}
S.~Karlin and H.~M. Taylor.
\newblock {\em A first course in stochastic processes}.
\newblock Academic Press [A subsidiary of Harcourt Brace Jovanovich,
  Publishers], New York-London, second edition, 1975. \MR{0356197}

\bibitem[KU18]{kosowski2018population}
A.~Kosowski and P.~Uzna'ski.
\newblock Population protocols are fast.
\newblock {\em CoRR Vol. abs/1802.06872 (2018). showeprint [arxiv]}, 2018.

\bibitem[Kut02]{kutin-mcdiarmid}
S.~Kutin.
\newblock {\em Algorithmic stability and ensemble-based learning}.
\newblock PhD thesis, University of Chicago, 2002.
\newblock Chapter 3, \url{people.cs.uchicago.edu/~kutin/publications/thesis/}.

\bibitem[MNRS14]{mertzios2014determining}
G.~B. Mertzios, S.~E. Nikoletseas, C.~L. Raptopoulos, and P.~G. Spirakis.
\newblock Determining majority in networks with local interactions and very
  small local memory.
\newblock In {\em International Colloquium on Automata, Languages, and
  Programming}, pages 871--882. Springer, 2014.

\bibitem[MNT14]{mossel2014majority}
E.~Mossel, J.~Neeman, and O.~Tamuz.
\newblock Majority dynamics and aggregation of information in social networks.
\newblock {\em Autonomous Agents and Multi-Agent Systems}, 28(3):408--429,
  2014.

\bibitem[MSW11]{miranda2011tls}
P.~Miranda, M.~Siekkinen, and H.~Waris.
\newblock Tls and energy consumption on a mobile device: A measurement study.
\newblock In {\em Computers and Communications (ISCC), 2011 IEEE Symposium on},
  pages 983--989. IEEE, 2011.

\bibitem[MT17]{mossel2017opinion}
E.~Mossel and O.~Tamuz.
\newblock Opinion exchange dynamics.
\newblock {\em Probability Surveys}, 14:155--204, 2017.

\bibitem[NIY99]{nakata1999probabilistic}
T.~Nakata, H.~Imahayashi, and M.~Yamashita.
\newblock Probabilistic local majority voting for the agreement problem on
  finite graphs.
\newblock In {\em International Computing and Combinatorics Conference}, pages
  330--338. Springer, 1999.

\bibitem[PVV09]{perron2009using}
E.~Perron, D.~Vasudevan, and M.~Vojnovic.
\newblock Using three states for binary consensus on complete graphs.
\newblock In {\em INFOCOM 2009, IEEE}, pages 2527--2535. IEEE, 2009.

\bibitem[SCHK13]{shang2013upper}
S.~Shang, P.~Cuff, P.~Hui, and S.~Kulkarni.
\newblock An upper bound on the convergence time for quantized consensus.
\newblock In {\em INFOCOM, 2013 Proceedings IEEE}, pages 600--604. IEEE, 2013.

\bibitem[VN51]{vN51}
J.~Von~Neumann.
\newblock Various techniques used in connection with random digits.
\newblock {\em National Bureau of Standards Applied Math Series}, 12(36-38):5,
  1951.

\end{thebibliography}

\end{document}